%% file: main.tex
\pdfoutput=1

\documentclass[11pt]{article}

\usepackage[letterpaper,margin=0.97in]{geometry}
\usepackage{bm, amsmath, amssymb, amsthm, amsfonts}

\usepackage{scrextend}
\input{header}

\definecolor{light-gray}{gray}{0.8}

\newtheorem{theorem}{Theorem}[section]
\newtheorem{lemma}[theorem]{Lemma}

\newtheorem{claim}[theorem]{Claim}

\newtheorem{definition}{Definition}[section]

\usepackage{algcompatible}

\begin{document}

\title{Distributed Maximum Matching Verification in CONGEST}

\author{
Mohamad Ahmadi\\
\small University of Freiburg\\
\small mahmadi@cs.uni-freiburg.de
\and
Fabian Kuhn\\
\small University of Freiburg\\
\small kuhn@cs.uni-freiburg.de
}

\date{}

\maketitle

\begin{abstract}
  We study the maximum cardinality matching problem in a standard
  distributed setting, where the nodes $V$ of a given $n$-node network
  graph $G=(V,E)$ communicate over the edges $E$ in synchronous
  rounds. More specifically, we consider the distributed CONGEST
  model, where in each round, each node of $G$ can send an
  $O(\log n)$-bit message to each of its neighbors. We show that for
  every graph $G$ and a matching $M$ of $G$, there is a randomized
  CONGEST algorithm to \emph{verify} $M$ being a maximum
  matching of $G$ in time $O(|M|)$ and disprove it in time $O(D + \ell)$, where $D$ is the
  diameter of $G$ and $\ell$ is the length of a shortest augmenting path. We hope that our algorithm constitutes a
  significant step towards developing a CONGEST algorithm to
  \emph{compute} a maximum matching in time $\tilde{O}(s^*)$, where
  $s^*$ is the size of a maximum matching.
\end{abstract}

\setcounter{page}{0}
\thispagestyle{empty}

\input{introduction}

\input{overview}

\input{aug-finding}

\input{analysis}

\bibliographystyle{alphaabbr}
\bibliography{references}

\end{document}

%% file: header.tex
\usepackage{cite}
\usepackage{framed}
\usepackage[framemethod=tikz]{mdframed}
\usepackage{appendix}
\usepackage{graphicx}
\usepackage{color}
\usepackage{varwidth}
\usepackage{wrapfig}
\usepackage[ruled,vlined,linesnumbered,noend]{algorithm2e}
\usepackage[noend]{algpseudocode}
\usepackage[textsize=tiny]{todonotes}
\usepackage{array}
\usepackage{caption}
\usepackage{subcaption}

\usepackage{enumitem}
\setitemize{noitemsep,topsep=3pt,parsep=3pt,partopsep=3pt}

\usepackage{xspace}
\usepackage{xifthen}

\definecolor{darkgreen}{rgb}{0,0.5,0}
\definecolor{darkblue}{rgb}{0,0,0.8}
\usepackage{hyperref}
\hypersetup{
    unicode=false,          % non-Latin characters in Acrobat’s bookmarks
    colorlinks=true,        % false: boxed links; true: colored links
    linkcolor=darkblue,          % color of internal links (change box color with linkbordercolor)
    citecolor=darkgreen,        % color of links to bibliography
    filecolor=magenta,      % color of file links
    urlcolor=cyan           % color of external links
}
\RequirePackage[]{silence}
\usepackage[nameinlink,capitalize]{cleveref}

\newcommand{\calS}{\ensuremath{\mathcal{S}}}

\newcommand{\calG}{\ensuremath{\mathcal{G}}}
\newcommand{\calE}{\ensuremath{\mathcal{E}}}

\newcommand{\calA}{\ensuremath{\mathcal{A}}}

\newcommand{\calC}{\ensuremath{\mathcal{C}}}
\newcommand{\calD}{\ensuremath{\mathcal{D}}}
\newcommand{\calM}{\ensuremath{\mathcal{M}}}

\newcommand{\ignore}[1]{}

\algnewcommand\algorithmicswitch{\textbf{switch}}
\algnewcommand\algorithmiccase{\textbf{case}}

\algdef{SE}[SWITCH]{Switch}{EndSwitch}[1]{\algorithmicswitch\ #1\ \algorithmicdo}{\algorithmicend\ \algorithmicswitch}%
\algdef{SE}[CASE]{Case}{EndCase}[1]{\algorithmiccase\ #1}{\algorithmicend\ \algorithmiccase}%
\algtext*{EndSwitch}%
\algtext*{EndCase}%
\newcommand{\CONGEST}{\ensuremath{\mathsf{CONGEST}}\xspace}
\newcommand{\LOCAL}{\ensuremath{\mathsf{LOCAL}}\xspace}

\newcommand{\eps}{\varepsilon}
\newcommand{\poly}{\operatorname{\text{{\rm poly}}}}

\newcommand{\set}[1]{\left\{#1\right\}}

\newcommand{\hide}[1]{}

\newcommand{\FullOrShort}{short}

\ifthenelse{\equal{\FullOrShort}{full}}{
	
  \newcommand{\fullOnly}[1]{#1}
  \newcommand{\shortOnly}[1]{}

}{
  
  \newcommand{\shortOnly}[1]{#1}
  \newcommand{\fullOnly}[1]{}

}

\newcommand{\para}[1]{\smallskip \noindent\textbf{#1}}

\newcommand{\proc}[1]{\text{\texttt{#1}}}

\renewcommand{\phi}{\varphi}

%% file: introduction.tex
\section{Introduction and Related Work}
\label{sec:introduction}

For a graph $G=(V,E)$, a matching $M\subseteq E$ is a set of pairwise
disjoint edges and the maximum matching problem asks for a matching
$M$ of maximum possible cardinality (or of maximum possible weight if
the edges are weighted). Matchings have been at the center of the
attention in graph theory for more than a century (see, e.g.,~\cite{lovasz86}). 
Algorithmic problems dealing with the
computation of matchings are among the most extensively
studied problems in algorithmic graph theory. The problem of finding a
maximum matching is on the one hand simple enough so that it can be
solved efficiently \cite{edmonds65a,edmonds65b}, on the other hand the
problem has a rich mathematical structure and led to many important
insights in graph theory and theoretical computer science. Apart from
work in the standard sequential setting, the problem has been studied
in a variety of other settings and computational models. Exact or
approximate algorithms have been developed in areas such
as online algorithms (e.g., \cite{KarpVV90,EmekKW16}), streaming
algorithms (e.g., \cite{mcgregor05}), sublinear-time algorithms (e.g.,
\cite{ymh09,mansour13}), classic parallel algorithms (e.g.,
\cite{karp85,fischer93}), as well as also the recently popular massively
parallel computation model (e.g.,
\cite{czumaj2017round,assadi2019coresets,ghaffari2019sparsifying}). In
this paper, we consider the problem of verifying whether a given
matching is a maximum matching in a standard distributed setting,
which we discuss in more detail next.

\vspace{.2cm}
\para{Distributed maximum matching:} In the distributed context,
the maximum matching problem is mostly studied for networks in the
following synchronous message passing model. The network is modeled as
an undirected $n$-node graph $G=(V,E)$, where each node hosts a
distributed process and the processes communicate with each other over
the edges of $G$. As it is common practice, we 
identify the nodes with their processes and think of the nodes
themselves as the distributed agents. Time is divided into synchronous
rounds and in each round, each node $v\in V$ can perform some
arbitrary internal computation, send a message to each of its
neighbors in $G$, and receive the messages of the neighbors (round $r$ is assumed to start at time $r-1$ and end at time $r$). If the
messages can be of arbitrary size, this model is known as the \LOCAL
model \cite{linial92,peleg00}. In the more realistic \CONGEST model~\cite{peleg00},  in each round, each node can
send an arbitrary $O(\log n)$-bit message to each of its neighbors. 

\para{Our contribution:} As the main result of our paper, we give
a distributed maximum matching verification algorithm.
\begin{theorem}\label{thm:main}
  Given an undirected graph $G=(V,E)$ and a matching $M$
  of $G$, there is a randomized distributed \CONGEST model algorithm
  to test whether $M$ is a maximum matching.  If $M$ is a
  maximum matching, the algorithm verifies this in time
  $O(|M|)$, otherwise, the algorithm disproves it in
  time $O(D+\ell)$, where $D$ is the diameter of $G$ and
  $\ell$ is the length of a shortest augmenting path.
\end{theorem}
Our main technical contribution is a distributed algorithm that, given
a matching $M$ and a parameter $x$, determines if there is an
augmenting path of length at most $x$ in $O(x)$ rounds of the \CONGEST
model. If there is an augmenting path of length at most $x$, the
algorithm identifies two free (i.e., unmatched) nodes $u$ and $v$
between which such a path exists. We note that if the algorithm can be
extended to also \emph{construct} an augmenting path of length at most
$x$ between $u$ and $v$ in time $\tilde{O}(x)$, it would directly lead
to an $\tilde{O}(s^*)$-round algorithm for computing a maximum
matching, where $s^*$ is the size of a maximum matching. The reason
for this follows from the classic framework of Hopcroft and
Karp~\cite{hopkarp}.  It is well-known that if we are given a matching
$M$ of size $s^*-k$ for some integer $k\geq 1$, there is an
augmenting path of length less than $2s^*/k$ \cite{hopkarp}. Hence, if
we can find such a path and augment along it in time linear in the
length of the path, we get a total time of $O(s^*\log s^*)$ by summing
over all values of $k$ from $1$ to $s^*$. The same approach has been
used in \cite{akr18} to compute a maximum matching in time
$O(s^*\log s^*)$ in bipartite graphs. While finding a shortest
augmenting path is quite straightforward in bipartite graphs, getting
an efficient \CONGEST model algorithm for general
graphs turns out to be much more involved. We therefore hope that our
algorithm for finding the length and the endpoints of some shortest
augmenting path provides a significant step towards also efficiently
constructing a shortest augmenting path in the \CONGEST model and
therefore to obtaining an $\tilde{O}(s^*)$-time \CONGEST algorithm to
find a maximum matching. Before we discuss our
algorithm, the underlying techniques and the main challenges in more
detail in \Cref{sec:overview}, we next give a brief summary of the
history of the distributed maximum matching problem.

\vspace{.2cm}
\para{Distributed maximal matching algorithms:}
While except for \cite{akr18}, there is no previous work on exact
solutions for the distributed maximum matching problem, there is a
very extensive and rich literature on computing approximate solutions
for the problem. The most basic way to approximate maximum matching is
by computing a \emph{maximal matching}, which provides a
$1/2$-approximation for the maximum matching problem. The work
on distributed maximal matching algorithms started with the classic
randomized parallel maximal matching and maximal independent set
algorithms from the 1980s~\cite{alon86,itai86,luby86}. While these
algorithms were originally described for the PRAM setting, they
directly lead to randomized $O(\log n)$-round algorithms in the
\CONGEST model. It was later shown by Ha\'{n}\'{c}kowiak,
Karo\'{n}ski, and Panconesi~\cite{hanckowiak98,hanckowiak99} that
maximal matching can also be solved deterministically in
polylogarithmic time in the distributed setting. The current best
deterministic algorithm in the \CONGEST model (and also in the \LOCAL
model) is by Fischer~\cite{rounding} and it computes a maximal
matching in $O(\log^2\Delta\log n)$ rounds, where $\Delta$ is the
maximum degree of the network graph $G$. At the cost of a higher
dependency on $\Delta$, the dependency on $n$ can be reduced and it
was shown by Panconesi and Rizzi~\cite{panconesi01} that a maximal
matching can be computed in $O(\Delta+\log^* n)$ rounds. The best
known randomized algorithm is by Barenboim et al.~\cite{barenboim12}
and it shows that (by combining with the result of \cite{rounding}) a
maximal matching can be computed in $O(\log\Delta) + O(\log^3\log n)$
rounds in the \CONGEST model. The known bounds in the \CONGEST model
are close to optimal even when using large messages. It is known that
there is no randomized
$o\big(\frac{\log\Delta}{\log\log\Delta}+\sqrt{\frac{\log n}{\log\log
    n}}\big)$-round
maximal matching algorithm in the \LOCAL model~\cite{kuhn16_jacm}. A
very recent result further shows that there are also no randomized
$o\big(\Delta + \frac{\log\log n}{\log\log\log n}\big)$-round
algorithm and no deterministic
$o\big(\Delta+\frac{\log n}{\log\log n}\big)$-round algorithms to
compute a maximal matching~\cite{linearinDeltaLower}.

\vspace{.2cm}
\para{Distributed matching approximation algorithms:}
There is a series of papers that target the distributed maximum
matching problem directly and that provide results that go beyond the
$1/2$-approximation achieved by computing a maximal matching. Most of
them are based on the framework of Hopcroft and Karp~\cite{hopkarp}:
after $O(1/\eps)$ iterations of augmenting along a (nearly) maximal
set of vertex-disjoint short augmenting paths, one is guaranteed to
have a $(1-\eps)$-approximate solution for the maximum matching
problem. The first distributed algorithms to use this approach are an
$O(\log^{O(1/\eps)}n)$-time deterministic \LOCAL algorithm for
computing a $(1-\eps)$-approximation in graphs of girth at least
$2/\eps-2$~\cite{czygrinow03} and an $O(\log^4 n)$-time deterministic
\LOCAL algorithm for computing a $2/3$-approximation in general
graphs~\cite{czygrinow2004}. The first approximation algorithms in the
\CONGEST model are by Lotker et al.~\cite{lotker08}, who give a
randomized algorithm to compute a $(1-\eps)$-approximate maximum
matching in time $O(\log n)$ for every constant $\eps>0$. For
bipartite graphs, the running time of the algorithm depends
polynomially on $1/\eps$, whereas for general graphs it depends
exponentially on $1/\eps$.\footnote{In the \LOCAL model, the algorithm
  can be implemented in time $O(\log(n)/\poly(\eps))$ also for general
  graphs. This was independently also shown in a concurrent paper by
  Nieberg~\cite{nieberg08}.} The algorithm was recently improved by
Bar Yehuda et al.~\cite{yehuda17}, who give an algorithm with time
complexity $O\big(\frac{\log\Delta}{\log\log\Delta}\big)$ for
computing a $(1-\eps)$-approximation. As in \cite{lotker08}, the time
depends polynomially on $1/\eps$ in bipartite graphs and exponentially
on $1/\eps$ in general graphs. Note that the time dependency on
$\Delta$ in \cite{yehuda17} matches the lower bound of
\cite{kuhn16_jacm}. In \cite{akr18}, Ahmadi et al.\ give a
deterministic
$O\big(\frac{\log\Delta}{\eps^2}+\frac{\log^2\Delta}{\eps}\big)$-round
\CONGEST maximum matching algorithm that has an approximation factor
of $(1-\eps)$ in bipartite graphs and an approximation factor of
$(2/3-\eps)$ in general graphs. Unlike the previous algorithms, the
algorithm of \cite{akr18} is not based on the framework of Hopcroft
and Karp. Instead, the algorithm first computes an almost optimal
fractional matching and it then rounds the fractional solution to an
integer solution by adapting an algorithm of \cite{rounding}. There
also exist deterministic distributed algorithms to
$(1-\eps)$-approximate maximum matching in polylogairhtmic
time~\cite{even15,focs17,derandomization}, these algorithm however
require the \LOCAL model. The algorithms of
\cite{akr18,focs17,derandomization} directly also work for the maximum
weighted matching problem. Other distributed algorithms that compute
constant-factor approximations for the weighted maximum matching
problem appeared in
\cite{wattenhofer04,hoepman06,lotker08,lotker09,rounding,yehuda17}.
We note that none of the existing approximation algorithms can be used
to solve the exact maximum matching problem in time $o(|E|)$ in the
\CONGEST model.

\vspace{.2cm}
\para{Additional related work:} Our result can also be seen in
the context of some recent interest in the complexity of computing
exact solutions to distributed optimization problems. In particular,
it was recently shown that several problems that are closely related
to the maximum matching problem have near-quadratic lower bounds in
the \CONGEST model. In \cite{censorhillel_disc17}, it is shown that
computing an optimal solutions to the maximum independent set and the
minimum vertex cover problem both require time $\tilde{\Omega}(n^2)$
in the \CONGEST model. In \cite{bachrach_podc19}, similar
$\tilde{\Omega}(n^2)$ lower bounds are proven for other problems, in
particular for computing an optimal solution to the minimum dominating
set problem and for computing a $(7/8+\eps)$-approximation for maximum
independent set. Consequently, for maximum independent set, minimum
vertex cover and minimum dominating set, the trivial $O(|E|)$-time
\CONGEST model algorithm is almost optimal. If our result can be extended to actually find the maximum matching in almost
linear time, it would show that this is not true for the
maximum matching problem.

\vspace{.2cm}
\para{Mathematical notation:} Before giving an outline of our
algorithm in \Cref{sec:overview}, we introduce some
graph-theoretic notation that we will use throughout the remainder of
the paper. A walk $W$ from node $u$ to a node $v$ in a graph $G=(V,E)$
is a sequence of nodes $\langle u=v_1, v_2, \dots , v_k=v \rangle$
such that for all $j<k$, $\set{v_j, v_{j+1}}\in E$.  A path $P$ is a
walk that is cycle-free, i.e., a walk where the nodes are pairwise
distinct.  Let $\mathbb{V}(W)$ denote the multi-set of the nodes in a
walk $W$ and let $|W|$ denote the length of the walk $W$, i.e.,
$|W|=|\mathbb{V}(W)|-1$.  For simplicity, we write $v\in W$ if
$v\in \mathbb{V}(W)$.  Moreover, we say an edge $e$ is on walk $W$ and
write $e\in W$ if $e$ is an edge between two consecutive nodes in
$W$. For two walks $W_1=\langle u_1,\dots,u_s\rangle$ and
$W_2=\langle v_1,\dots,v_t\rangle$ with $u_s=v_1$, we use
$W_1\circ W_2$ to denote the concatenation of the walks
$W_1$ and $W_2$.  Further, for a path
$P=\langle u_1, u_2, \dots , u_i, \dots , u_j, \dots , u_k\rangle$, we
use $P[u_i, u_j]$ to denote the consecutive subsequence of $P$
starting at node $u_i$ and ending at node $u_j$, i.e., the subpath of
$P$ from $u_i$ to $u_j$. We use parentheses instead of square brackets to exclude the starting or ending node from the subpath, e.g., $P(u_i, u_j]$, $P[u_i, u_j)$ or $P(u_i, u_j)$.

%%% Local Variables:
%%% mode: latex
%%% TeX-master: "main"
%%% End:

%% file: overview.tex
%\vspace{-.2cm}
\section{Outline of Our Approach}
\label{sec:overview}

For a graph $G$, it is well-known that a matching $M$ is a maximum matching of $G$ if and only if there is no augmenting path in $G$ w.r.t.\ $M$. 
	By performing a broadcast/convergecast, the size of the given matching can be learnt by all nodes in the graph in time linear in $D$, the diameter of $G$. 
	After all nodes learn the size of the given matching, the algorithm looks for an augmenting path of length at most $r$ in phases, where $r$ is initially set to $D$ and it doubles from each phase to the next. 
	The algorithm stops as soon as either $r > 4 |M|$ or it detects a shortest augmenting path of length at most $r$. 
	Note that the length of an augmenting path cannot be more than $2|M|+1$. 
	Therefore, if the algorithm does not find a shortest augmenting path, then there is no augmenting path in $G$ with respect to $M$. 
	The efficiency of the algorithm depends on how fast one can detect the existence of a shortest augmenting path of length at most $\ell$ in $G$ for an integer $\ell$. 
	The following main technical result states that this central challenging task can be accomplish efficiently. 	
	\begin{lemma}\label{lem:path}
		Given an arbitrary graph $G$ and a matching $M$ of $G$, there is a randomized algorithm to detect whether there exists an augmenting path of length at most $\ell$ in $O(\ell)$ rounds of the \CONGEST model, with high probability. 
	\end{lemma}
	Considering the above lemma, let us now study the time complexity of the algorithm. 
	First, let us assume that the given matching $M$ is a maximum matching. 
	Then, it takes $O(D)$ rounds to learn the size of the given matching and $O(|M|)$ rounds to look for shortest augmenting paths in phases since $\sum_{i = \log D}^{2 + \log |M|} O(2^i) \leq  O(|M|)$. 
	However, since $\Omega(D)$ is a lower bound for the size of a maximum matching, the overall time complexity of the algorithm to verify that $M$ is a maximum matching is $O(|M|)$. 	
	Now let us assume that $M$ is not a maximum matching and $\ell$ is the length of a shortest augmenting path. 
	Then, it takes $O(D)$ rounds to learn the size of the given matching and $O(\ell)$ rounds for the algorithm to look for and eventually detect a shortest augmenting path in phases. 
	Hence, the overall time complexity of the algorithm to disprove $M$ being a maximum matching is $O(\ell + D)$. 
	This overall implies \Cref{thm:main}.
		   
\vspace{-.1cm}
\subsection{Detecting a Shortest Augmenting Path: The Challenges}
\label{sec:challenges}

Let $G=(V,E)$ be a graph, let $M\subseteq E$ be a matching of $G$, and let $f\in V$ be a free node (i.e., an unmatched node). Assume that we want to find a shortest augmenting path $P$ connecting $f$ with another free node $f'$. If the graph $G$ is bipartite, such a path can be found by doing a breadth first search (BFS) along alternating paths from $f$. This works because in bipartite graphs, for every node $v$ on a shortest augmenting path $P$ connecting $f$ with another free node $f'$, the subpath $P[f,v]$ connecting $f$ and $v$ is also a shortest alternating path between $f$ and $v$. In \cite{vazirani13}, Vazirani calls this property, which holds in bipartite graphs, the \emph{BFS-honesty property}. If the BFS-honesty property holds, to find a shortest alternating path from a free node $f$ to a node $v$, it suffices to know shortest alternating paths from $f$ to all the nodes along this path. The BFS-honesty property does not hold in general graphs. A simple example that shows this is given in \Cref{fig:BFS-honesty}. The shortest alternating path connecting node $f$ with $u$ is of length $3$. The shortest alternating path connecting $f$ with $v$ is of length $5$ and it contains node $u$, however the subpath connecting $f$ with $u$ on the alternating path to $v$ is of length $4$.

\begin{figure}[h] 
  \centering	
  \begin{tikzpicture}
    \shade[shading=ball, ball color=black] (-3,0) circle (.08);
    \node at (-3, .3) {$f$};
    \shade[shading=ball, ball color=black] (-2,0) circle (.08);
    \shade[shading=ball, ball color=black] (-1,0) circle (.08);
    \shade[shading=ball, ball color=black] (0,0) circle (.08);
    \node at (0, .3) {$u$};
    \shade[shading=ball, ball color=black] (1,0) circle (.08);
    \shade[shading=ball, ball color=black] (2,0) circle (.08);
    \shade[shading=ball, ball color=black] (-.5,-.7) circle (.08);
    \node at (1, .3) {$v$};
    \draw[dotted,line width=.4mm] (-3,0) - - (-2,0);
    \draw[line width=.4mm] (-2,0) - - (-1,0);
    \draw[dotted, line width=.4mm] (-1,0) - - (0,0);
    \draw[dotted, line width=.4mm] (0,0) - - (1,0);
    \draw[line width=.4mm] (1,0) - - (2,0);
    \draw[dotted, line width=.4mm] (-1,0) - - (-.5,-.7);
    \draw[line width=.4mm] (-.5,-.7) - - (0,0);
    \node at (-0.5, -1.0) {$w$};
  \end{tikzpicture}	
  \caption{The BFS-honesty property does not hold in general graphs (solid lines depict edges in the matching, dotted lines depict edges not in the matching). }
  \label{fig:BFS-honesty}
\end{figure}
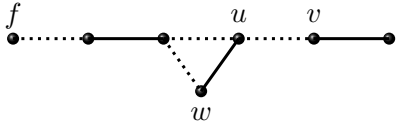
\vspace{-.2cm}
To show the use of the DFS-honesty property in the distributed setting more clearly, we next sketch the algorithm of \cite{akr18} for finding a shortest augmenting path in a bipartite graph. The algorithm essentially works as follows. Every free node $f\in V$ in parallel starts its own BFS exploration of $G$ along alternating paths. The exploration of a free node $f$ is done by propagating its ID (i.e., $f$) along alternating paths from $f$, where the ID is propagated by one more hop in each synchronous round. Whenever a node $u$ receives the IDs of two different free nodes $f$ and $f'$ in the same round, it only forwards the ID of one of them. %The algorithm guarantees that after $r$ rounds, every node that has a shortest alternating path of length $r'\leq r$ to a free node, has found an alternating path of length $r'$ to some free node. 
Note that each node only forwards a single free node ID and it only forwards this ID once (in the round after it first receives it). This is sufficient if the BFS-honesty property holds. 
Moreover, this guarantees that IDs only traverse alternating paths and avoid traversing cycles. 
Assume that the shortest augmenting path in the graph is of length $\ell=2k+1$. let $P$ be such a path and assume that $\set{u,v}$ is the middle edge of $P$. Note that this implies that the shortest alternating paths of nodes $u$ and $v$ are both of length $k$. Hence, $u$ and $v$ receive the ID of a free node exactly in round $k$ and they will both forward that ID along edge $\set{u,v}$ in round $k+1$. When this happens, $u$ and $v$ learn about the fact that they are in the middle of an augmenting path of length $2k+1$ and that path can be constructed simply by following back the edges on which the alternating BFS traversals reached nodes $u$ and $v$. 
%{\color{blue}Note that this BFS exploration does not guarantee that each and every middle edge of all paths in a maximal set of disjoint shortest augmenting paths detects the path. However, it adequately guarantees that at least one shortest augmenting path is detected over its middle edge. %(I added this, but I am not sure if it is necessary or is ``out of domain'')} 

Let us now discuss some of the challenges when adapting this ID dissemination protocol to general graphs. For simplicity, assume that we are only doing the BFS exploration from a single free node $f$. Consider again the example in \Cref{fig:BFS-honesty}. We have seen that the shortest alternating path from $f$ to $v$ passes through node $u$, however the subpath from $f$ to $u$ is not the shortest alternating path from $f$ to $u$. In fact, while the shortest alternating path from $f$ to $u$ reaches $u$ on an unmatched edge, in order to reach node $v$, we have to use the shortest one of the alternating paths from $f$ to $u$ that reach $u$ on a matched edge. This suggests that each node $v$ should keep track of both kinds of shortest paths from node $f$ and that $v$ should forward $f$ twice.  A natural generalization of the protocol would thus be the following: After receiving $f$ on a shortest alternating path ending in an unmatched edge of $v$, $v$ forwards $f$ on its matched edge and after receiving $f$ on a shortest alternating path ending in the matched edge of $v$, $v$ forwards $f$ on its unmatched edges. One would hope that this lets each node detect both kinds of shortest alternating paths from node $f$. However, as \Cref{fig:distinguish} shows, this is not necessarily true. While in the \Cref{fig:sub-1}, when $v$ receives $f$ over its matched edge, the ID was indeed forwarded on a shortest alternating path from $f$ to $v$. However, in \Cref{fig:sub-2}, the exploration passes through an odd cycle and node $v$ is only reached on an alternating walk instead of an alternating path. In the example of \Cref{fig:sub-2}, node $w$ should detect that the BFS traversal passed through the odd cycle and $w$ should therefore not forward $f$ over its matched edge. However, it is not clear how $w$ should distinguish between the cases in \Cref{fig:sub-1} and \Cref{fig:sub-2}.  Note that in the BFS traversal of \Cref{fig:sub-1}, $f$ is not only forwarded on the alternating path to $w$, but it is also forwarded through the odd cycle as in \Cref{fig:sub-2}. In fact, the example of \Cref{fig:distinguish} is still a relatively simple case as odd cycles can be nested, and closed odd walks can look much more complicated than just passing through a single odd cycle. Detecting whether and when to forward the ID of a free node is the main algorithmic challenge that we face.	

\begin{figure}[t]
		\centering
		\begin{subfigure}{.4\textwidth}
 			\centering
  				\begin{tikzpicture}
		\shade[shading=ball, ball color=black] (-3,0) circle (.08);
		\node at (-3, .3) {$f$};
		\shade[shading=ball, ball color=black] (-2,0) circle (.08);
		\shade[shading=ball, ball color=black] (-1,0) circle (.08);
		\shade[shading=ball, ball color=black] (0,0) circle (.08);
		\shade[shading=ball, ball color=black] (1,0) circle (.08);
		\shade[shading=ball, ball color=black] (2,0) circle (.08);
		\node at (2, .3) {$v$};
		\draw[dotted, line width=.4mm] (-3,0) - - (-2,0);
		\draw[line width=.4mm] (-2,0) - - (-1,0);
		\draw[dotted, line width=.4mm] (-1,0) - - (0,0);
		\draw[dotted, line width=.4mm] (0,0) - - (1,0);
		\draw[line width=.4mm] (1,0) - - (2,0);
		\shade[shading=ball, ball color=black] (0,-.7) circle (.08);
		\node at (.3,-.6) {$w$};
		\node at (-1,.3) {$z$};
		\shade[shading=ball, ball color=black] (-.5,-1.3) circle (.08);
		\shade[shading=ball, ball color=black] (.5, -1.3) circle (.08);
		\shade[shading=ball, ball color=black] (-.5, -2) circle (.08);
		\shade[shading=ball, ball color=black] (.5, -2) circle (.08);
		\draw[line width=.4mm] (0,0) - - (0,-.7);
		\draw[dotted, line width=.4mm] (0,-.7) - - (-.5,-1.3);
		\draw[dotted, line width=.4mm] (.5,-1.3) - - (0,-.7);
		\draw[line width=.4mm] (-.5,-1.3) - - (-.5,-2);
		\draw[dotted, line width=.4mm] (-.5,-2) - - (.5,-2);
		\draw[line width=.4mm] (-.5,-2) - - (-.5,-1.3);
		\draw[line width=.4mm] (.5,-2) - - (.5,-1.3);
		\shade[shading=ball, ball color=black] (-.85,-.4) circle (.08);
		\shade[shading=ball, ball color=black] (-.7, -.8) circle (.08);
		\draw[dotted, line width=.4mm] (-1,0) - - (-.85,-.4);
		\draw[line width=.4mm] (-.85,-.4) - - (-.7, -.8);
		\draw[dotted, line width=.4mm] (-.7, -.8) - - (-.5, -1.3);
		
		\draw[color=red, ->] (-3,-.2) - - (-2,-.2);
		\draw[color=red, ->] (-2,-.2) - - (-1.1,-.2);
		\draw[color=red, ->] (-1.1,-.2) - - (-.97,-.5);
		\draw[color=red, ->] (-.97,-.5) - - (-.83,-.9);
		\draw[color=red, ->] (-.83,-.9) - - (-.7,-1.3);
		\draw[color=red, ->] (-.7,-1.3) - - (-.7,-2.2);
		\draw[color=red, ->] (-.7,-2.2) - - (.7,-2.2);
		\draw[color=red, ->] (.7,-2.2) - - (.7,-1.3);
		\draw[color=red, ->] (.7,-1.3) - - (.2,-.7);
		\draw[color=red, ->] (.2,-.7) - - (.2,-.2);
		\draw[color=red, ->] (.2,-.2) - - (1,-.2);
		\draw[color=red, ->] (1,-.2) - - (2,-.2);

				\end{tikzpicture}
  			\caption{}
  			\label{fig:sub-1}
		\end{subfigure}
		\begin{subfigure}{.4\textwidth}
 			\centering
  				\begin{tikzpicture}
		\shade[shading=ball, ball color=black] (-3,0) circle (.08);
		\node at (-3, .3) {$f$};
		\shade[shading=ball, ball color=black] (-2,0) circle (.08);
		\shade[shading=ball, ball color=black] (-1,0) circle (.08);
		\shade[shading=ball, ball color=black] (0,0) circle (.08);
		\shade[shading=ball, ball color=black] (1,0) circle (.08);
		\shade[shading=ball, ball color=black] (2,0) circle (.08);
		\node at (2, .3) {$v$};
		\draw[dotted, line width=.4mm] (-3,0) - - (-2,0);
		\draw[line width=.4mm] (-2,0) - - (-1,0);
		\draw[dotted, line width=.4mm] (-1,0) - - (0,0);
		\draw[dotted, line width=.4mm] (0,0) - - (1,0);
		\draw[line width=.4mm] (1,0) - - (2,0);
		\shade[shading=ball, ball color=black] (0,-.7) circle (.08);
		\node at (.3,-.6) {$w$};
		\shade[shading=ball, ball color=black] (-.5,-1.3) circle (.08);
		\shade[shading=ball, ball color=black] (.5, -1.3) circle (.08);
		\shade[shading=ball, ball color=black] (-.5, -2) circle (.08);
		\shade[shading=ball, ball color=black] (.5, -2) circle (.08);
		\draw[line width=.4mm] (0,0) - - (0,-.7);
		\draw[dotted, line width=.4mm] (0,-.7) - - (-.5,-1.3);
		\draw[dotted, line width=.4mm] (.5,-1.3) - - (0,-.7);
		\draw[line width=.4mm] (-.5,-1.3) - - (-.5,-2);
		\draw[dotted, line width=.4mm] (-.5,-2) - - (.5,-2);
		\draw[line width=.4mm] (-.5,-2) - - (-.5,-1.3);
		\draw[line width=.4mm] (.5,-2) - - (.5,-1.3);
		\draw[color=red, ->] (-3,-.2) - - (-2,-.2);
		\draw[color=red, ->] (-2,-.2) - - (-1,-.2);
		\draw[color=red, ->] (-1,-.2) - - (-.2,-.2);
		\draw[color=red, ->] (-.2,-.2) - - (-.2,-.7);
		\draw[color=red, ->] (-.2,-.7) - - (-.7,-1.3);
		\draw[color=red, ->] (-.7,-1.3) - - (-.7,-2.2);
		\draw[color=red, ->] (-.7,-2.2) - - (.7,-2.2);
		\draw[color=red, ->] (.7,-2.2) - - (.7,-1.3);
		\draw[color=red, ->] (.7,-1.3) - - (.2,-.7);
		\draw[color=red, ->] (.2,-.7) - - (.2,-.2);
		\draw[color=red, ->] (.2,-.2) - - (1,-.2);
		\draw[color=red, ->] (1,-.2) - - (2,-.2);
			\end{tikzpicture}
  			\caption{}
  			\label{fig:sub-2}
		\end{subfigure}
	\caption{Main challenge: Nodes need to be able to distinguish whether the BFS exploration reaches them on an alternating path or only on alternating walks.}
	\label{fig:distinguish}
	\end{figure}
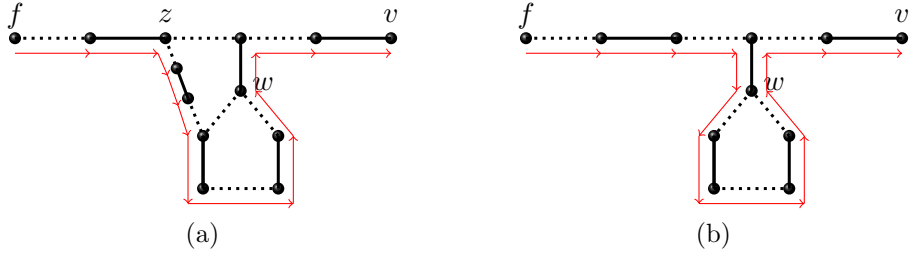
A second challenge come from the fact that we need to do alternating BFS explorations from all free nodes and it is not obvious how to coordinate these parallel BFS explorations while keeping the message size small. In the bipartite case, it was enough for each node $v$ to only participate in the BFS exploration of a single free node $f$ and to discard all other BFS traversals that reach node $v$. It is not clear whether the same thing can also be done in general graph. Luckily it turns out to still be sufficient if each node $v$ only participates in the BFS exploration of the first free node that reaches $v$. Proving that this is sufficient is however more involved than that in the bipartite case.	
\subsection{The Free Node Clustering}
\label{sec:clusteringoutline}

We start the outline of our algorithm to detect a shortest augmenting path by describing the required outcome of the alternating BFS exploration in general graphs in more detail. As mentioned above, we intend to in parallel perform BFS explorations starting from all the free nodes $f_1,\dots,f_{\rho}$. We will show that it is sufficient for each node $v\in V$ to participate in the BFS exploration for exactly one free node $f_i$. This implies that at each point in time, the BFS explorations of the different free nodes $f_1,\dots,f_{\rho}$ induce a clustering of the nodes in $V$. There is a cluster for each free node $f_i$, and each node $v\in V$ is either contained in exactly one of the $\rho$ clusters or it is not contained in any cluster (i.e., has not been reached by any of the explorations). We call this induced clustering the \emph{free node clustering}. The clustering is computed in synchronous rounds and we will guarantee that it satisfies the following properties.

\begin{itemize}
\item[(C1)] Consider some node $v\in V$ and some round number $r\geq 1$. If $v$ has not joined any cluster in the first $r-1$ rounds, if $v$ has an alternating path $P$ of length $r$ to a cluster center $f$ (i.e., a free node $f$), and if all nodes of $P$ except node $v$ are in the cluster of node $f$ after $r-1$ rounds, then $v$ joins the cluster of $f$ or some other cluster with the same property. That is, $v$ joins a cluster in round $r$ such that afterwards, it has an alternating path $P'$ of length $r$ to its cluster center such that $P'$ is completely contained in the cluster that $v$ joined. Let $C$ be the cluster that $v$ joins and let $U\subseteq C$ be the set of neighbors $u$ of $v$ such that the cluster contains an alternating path of length $r$ from $v$ through $u$ to the cluster center. The set $U$ is called the \emph{predecessors} of $v$. If $v$ joins a cluster in round $r$, we say that $v$ is \emph{$r$-reachable} (i.e., $v$'s shortest alternating paths to its cluster center that are completely contained in the cluster are of length $r$). If $v$'s adjacent edge on the shortest alternating path of $v$ is an unmatched edge (i.e., if $r$ is odd), we say that $v$ is \emph{$r$-$0$-reachable} and otherwise, we say that it is \emph{$r$-$1$-reachable}.
\item[(C2)] Assume that $v$ is $r_v$-reachable. Let $r>r_v$ be the first round after which the cluster of $v$ contains an alternating path $P$ of length $r$ connecting $v$ with $f$ such that if $v$ is $r_v$-$0$-reachable, $P$ starts with a matched edge at node $v$ and if $v$ is $r_v$-$1$-reachable, $P$ starts with an unmatched edge at node $v$ (if such a round $r$ exists). Then, after $r$ rounds of the construction, $v$ is aware of the existence of such a path. If $v$ is $r_v$-$0$-reachable, we say it is also \emph{$r$-$1$-reachable} and if it is $r_v$-$1$-reachable, we say that it is also \emph{$r$-$0$-reachable}.
\end{itemize}
To put it differently, a node in a cluster is called $r$-$0$-reachable ($r$-$1$-reachable) if there is a shortest odd-length (even-length) alternating path of length $r$ from the cluster center to the node that is completely contained in the cluster.
The clustering after $r$ rounds of the construction will be called the \emph{$r$-radius free node clustering}. For the precise definition of the clustering and of the related terminology, we refer to \Cref{sec:clustering-definitions}. We will see that the $r$-radius free node clustering can be constructed in $r$ rounds in the \CONGEST model. We give an outline of the distributed construction of the clustering in the following \Cref{sec:clusteringalgoutline}. The details of the distributed construction and its analysis appear in \Cref{sec:3,sec:analysis}. Before discussing the distributed construction, we next sketch how the free node clustering can be used to detect an augmenting path and why it is sufficient for detecting a shortest augmenting path.

%%%% Comment
\iffalse
\\[.1cm]
\noindent
\textit{Inner-clustering properties:} For every node $v$ in the cluster $C_f$ centered at a free node $f$,
\begin{itemize}
	\item There is an alternating path from $f$ to $v$ that is completely contained in cluster $C_f$. 
	\item Node $v$ knows the length of the shortest alternating path from $f$ to $v$ that is completely contained in $C_f$ and includes a free edge of $v$ (if any such path exists). 
	\item Node $v$ knows the length of the shortest alternating path from $f$ to $v$ that is completely contained in $C_f$ and includes the matched edge of $v$ (if any such path exists). 
	\item Node $v$ knows all its neighbors on the shortest alternating paths from $f$ to $v$. 
\end{itemize}
\noindent
\textit{Intra-clustering property:}
\begin{itemize}
	\item There is a shortest augmenting path that is partitioned into two consecutive subpaths, each of which is completely contained in a different cluster. 
\end{itemize}

	In this subsection, let us avoid arguing the distributed construction of the clustering. 
	Rather we first discuss the possibility of having a clustering with the inner-clustering properties. 
	Then, we discuss that having the inner-properties implies the intra-property of such clustering. 
	Later in Subsection \ref{}, we argue the distributed construction of the clustering.  
\fi%%%%%

\vspace{.2cm}
\para{Detecting augmenting paths:}
After computing the $r$-radius free node clustering for a sufficiently large radius $r$, we can use it to find an augmenting path as follows. Let $u$ and $v$ be two neighbors in $G$ such that $u$ and $v$ are in different clusters (say for free nodes $f$ and $f'$). Assume that for two integers $\ell_u,\ell_v\geq 0$ one of the following conditions hold:
\begin{enumerate}
\item The edge $\set{u,v}$ is in the matching, $u$ is $\ell_u$-$0$-reachable (in its cluster), and $v$ is $\ell_v$-$0$-reachable (in its cluster).
\item The edge $\set{u,v}$ is not in the matching, $u$ is $\ell_u$-$1$-reachable (in its cluster), and $v$ is $\ell_v$-$1$-reachable (in its cluster).
\end{enumerate}
In both cases the matching directly implies that there exists an augmenting path of length $\ell_u+\ell_v+1$ between the free nodes $f$ and $f'$. Further, after $\max\set{\ell_u,\ell_v}+1$ rounds, $u$ and $v$ are aware of the existence of this path. 
%We have to be a little bit careful because various paths of potentially different lengths might be detected a different times during the construction. However, because $u$ and $v$ are also aware of the length of the detected augmenting path, it is not difficult to construct a shortest detectable path by prioritizing shorter paths over longer paths during the detection phase. The details of how to construct and augment along a shortest augmenting path (based on an existing free node clustering) appear in \Cref{sec:path-setup}.

\vspace{.2cm}
\para{Detectability of a shortest augmenting path:} 
It remains to show that the free node clustering allows to find some shortest augmenting path. Assume that the length of a shortest augmenting path in $G$ with respect to the given matching $M$ is $2k+1$ for some integer $k\geq 0$. For an augmenting path $P=\langle f=v_0,\dots,v_{\ell}=f'\rangle$ of length $\ell=2k+1$ between two free nodes $f$ and $f'$, we let $i\geq0$ and $j\geq 0$ be two integers such that $i$ is the largest integer such that all nodes in $P[v_0,v_i]$ are in the cluster of $f$ and such that $j$ is the largest integer such that all the nodes in $P[v_{\ell-j},v_\ell]$ are in the cluster of $f'$. We define the \emph{rank} of the augmenting path $P$ as $i+j$. Note that the path $P$ is detectable if and only if it has rank $2k$. We therefore need to show that there exists a shortest augmenting path of rank $2k$.

To prove that there is a shortest augmenting path of rank $2k$, we assume that $P$ is a shortest augmenting path of maximal rank and that the rank of $P$ is less than $2k$ and we show that this leads to a contradition: this either allows to construct an augmenting path of length less than $2k+1$ or it allows to construct an augmenting path of length $2k+1$ of larger rank. The actual proof is somewhat technical. It consists of two steps. If we assume, w.l.o.g., that $i\leq j$, we first show inductively that all the nodes $v_{i+1},\dots,v_{\max\set{i+1,j}}$ are in the cluster of $f'$. If $j\geq \ell-j-1$, we have proven that $v_{\ell-j-1}$ is in the cluster of $f'$, which is a contradiction to the choice of $j$. Otherwise, node $v_{\ell-j-1}$ is in a cluster $f''\neq f'$ (it is however possible that $f''=f$). We can now derive the desired contradiction by a careful concatenation of parts of the paths connecting $f''$ with $v_{\ell-j-1}$, parts of the augmenting path between $f$ and $f'$, and parts of a path between $f'$ and $v_{i+1}$ that was constructed in the earlier inductive argument. The details of the arguments appear in \Cref{sec:path-setup}. %\fabian{Maybe, there is a nice way to give a little bit more details here, without writing too much. Also, the last part of the proof of \Cref{lem:detection1} should probably be done with a bit more details. Maybe, one could add a hand-waving argument that $v_{i+1}$ needs to be in the cluster of $f'$ (with a path $P_{i+1}$ of length at most $i+1$)}

%%\vspace{-.3cm}
\subsection{Distributed Construction of the Free Node Clustering}
\label{sec:clusteringalgoutline}

	We focus on a single step (round) of the the distributed construction of the free node clustering.
	To that end, consider graph $G$ and matching $M$, and assume that the first $r-1$ rounds of the clustering construction have been done successfully and the introduced clustering properties (C1) and (C2) of \Cref{sec:clusteringoutline} hold. 
	Therefore, for all integers $t<r$ and $\vartheta \in \set{0,1}$, every $t$-$\vartheta$-reachable node correctly detects the fact that it is $t$-$\vartheta$-reachable and knows its predecessors. 
	Then, let us explain the outline of the approach towards implementing the $r^{th}$ step of the distributed construction of the clustering. 
	
	Let us first focus on maintaining property (C1). 
	To satisfy (C1), every $(r-1)$-$\vartheta$-reachable node sends its cluster ID over its adjacent matched edge if $\vartheta = 0$ and over its adjacent unmatched edges if $\vartheta = 1$. 
	This way, for every node that receives a cluster ID over its adjacent edge, by joining the corresponding cluster, there would be an alternating path of length $r$ from the cluster center to the node such that the path is completely contained in the cluster. 
	This maintains property (C1) for $r$ steps of the clustering and it can be achieved in a distributed setting as explained. 
	However, maintaining property (C2) is the main challenge as we try to elaborate in the sequel.
	 
	Let us consider a node $v$ in the cluster centered at some free node $f$ after $r-1$ steps of the clustering construction such that it is $r'$-$\vartheta$-reachable for some integers $r' < r$ and $\vartheta \in \set{0,1}$. 
	Let us then assume that after the nodes have joined their corresponding clusters in the $r^{th}$ step, there is an alternating path of length $r$ from $f$ to $v$ that has completely fallen into the cluster centered at $f$ such that the path contains an adjacent matched edge of $v$ if $\vartheta = 0$ and an adjacent unmatched edge of $v$ otherwise.
	Therefore, $v$ should learn about the existence of such a path to maintain property (C2). 
	Nodes like $v$ can be reached within the cluster from their corresponding cluster center in two ways; first through an alternating path from the cluster center to $v$ such that the path is completely contained in the cluster and contains a matched edge of $v$, and second through a similar path but containing an unmatched edge of $v$. 
	Let us call these nodes that can be reached via both kinds of paths \textit{bireachable nodes}. 
	
	Let us define an \textit{odd cycle} to be an alternating walk of odd length that is completely contained in a cluster and starts and ends at the same node. 
	Node $v$ that is the first and last node of an odd cycle is called the stem of the odd cycle.
	An odd cycle is said to be \textit{minimal} if it has no consecutive subsequence that is an odd cycle. 
	Note that a minimal odd cycle can still have a consecutive subsequence that is an even-length cycle. 
	An odd cycle is moreover said to be \textit{reachable} if either the stem is the cluster center or there is an alternating path from the cluster center to the stem of the odd cycle such that (1) it is completely contained in the cluster, (2) it is edge-disjoint from the odd cycle, and (3) it includes the matched edge of the stem of the odd cycle. 
	You can see examples of reachable minimal odd cycles in \Cref{fig:sub-1}, one with stem $w$ and another one with stem $z$.
	All the nodes of an odd cycle except the stem are said to be \textit{strictly inside} the odd cycle.
	Then, one can show that a node is bireachable if and only if it is strictly inside a reachable minimal odd cycle.
	We only need this simple observation to explain the intuition behind our approach for maintaining property (C2) and in \Cref{sec:analysis} we formally prove the correctness of the approach. 
	As an example, node $w$ is strictly inside the reachable minimal odd cycle with stem $z$ in \Cref{fig:sub-1} and hence bireachable, but $w$ is not bireachable in \Cref{fig:sub-2}. 
			
	To help the nodes to distinguish whether they are strictly inside a reachable minimal odd cycle or not, we define a flow circulation protocol throughout each cluster.
	Let us consider the very simple example of a reachable minimal odd cycle in \Cref{fig:flow1}.
	When the cluster ID is sent over the middle edge of this odd cycle (i.e., $e$) in both directions in the same round, we consider a flow generation of unit size over the edge and we call it the flow of $e$. 
	Then, half of the generated flow is sent back towards the stem of the odd cycle on each of the two paths. 
	When the stem receives the whole unit flow of edge $e$ in a single round, it learns that it is the stem of an odd cycle for which the flow is generated and discards the flow (it avoids sending the flow further). 
	Whereas all the other nodes inside the odd cycle receive a flow of value less than $1$.
	They interpret this incomplete flow receipt as being strictly inside a reachable minimal odd cycle. 
	Moreover, they interpret the round in which they receive an incomplete flow for the first time as the length of an existing alternating path from the cluster center. 
	Therefore, to maintain property (C2), a node detects the length of its shortest alternating path through its matched (unmatched) edge by receiving an incomplete flow for the first time if its shortest alternating path contains its unmatched (matched) edge. 
	
	\begin{figure}[t]
		\centering
		\begin{subfigure}{.4\textwidth}
 			\centering
  				\includegraphics[width=4.5cm]{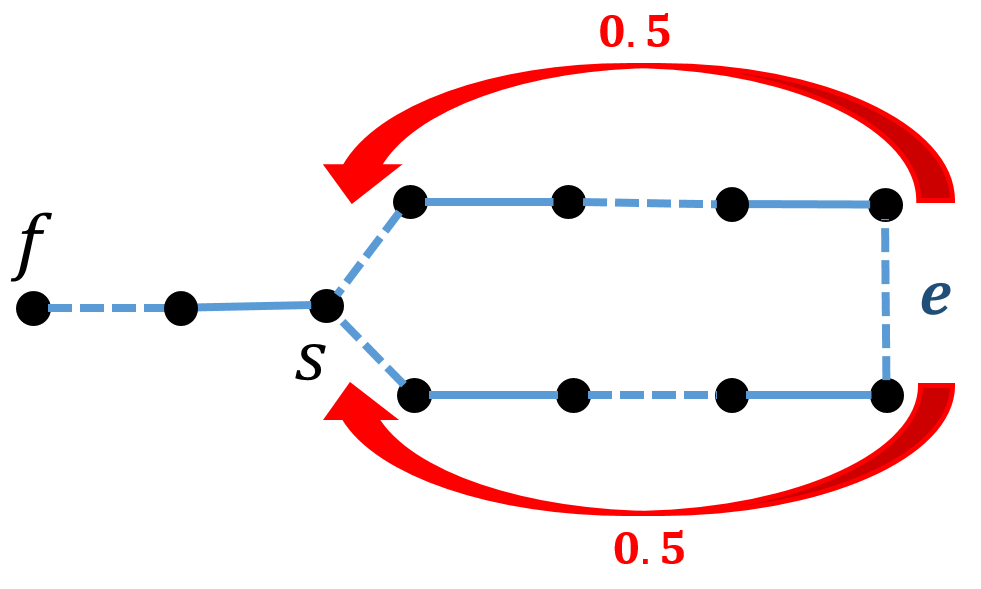}
  			\caption{}
  			\label{fig:flow1}
		\end{subfigure}
		\begin{subfigure}{.4\textwidth}
 			\centering
  				\includegraphics[width=7cm]{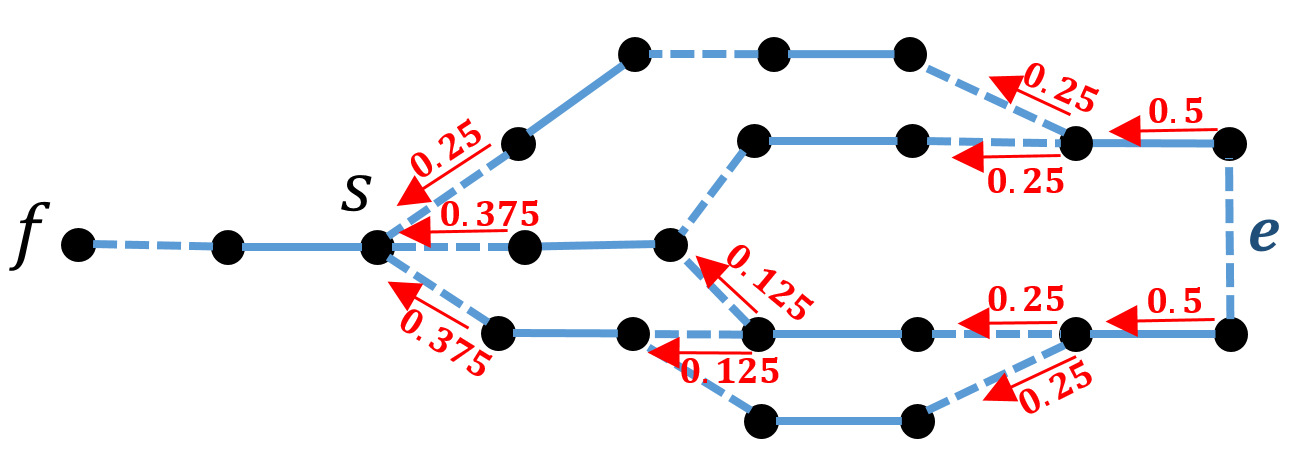}
  			\caption{}
  			\label{fig:flow2}
		\end{subfigure}
	\caption{Flow circulation in reachable minimal odd cycles}
	\label{fig:flow}
	\end{figure}

	Many such odd cycles might share a common middle edge as the cycles in \Cref{fig:flow2} that share edge $e$ as their middle edge.
	Then, it is enough that every node divides the value of the received flow of $e$ and sends them backwards until the cycle's stem, i.e., node $s$, receives the whole unit flow of $e$. 
	However, in case of having many interconnected and nested reachable minimal odd cycles that do not share a single middle edge, an edge might carry the flows of many different edges in the same round. This is a problem when implementing the idea in the \CONGEST model as we cannot bound the number of flows that have to be sent concurrently over an edge. Instead of separately sending flows generated at different edges $e$, we therefore sum all flows that have to go over the same edge and only send aggregate values. Ideally, we would like to have the following desired differentiation; a node that receives an aggregated flow whose size is not an integer, learns that it is strictly inside at least one reachable minimal odd cycle, and a node that receives an aggregated flow whose size is an integer learns that it is the stem of at least one reachable minimal odd cycle but not strictly inside any such cycle and it discards the flow. 
	To avoid that the sum of a set of fractional flows for different edges sums to an integer, we can use randomization. Instead of always equally splitting a flow that has to be sent over several edges, we randomly split the flow. This guarantees that w.h.p., flows only sum up to an integer if they consist of all parts of all involved separate flows. Unfortunately, this is still not directly implementable in the \CONGEST model because we might need to split a single flow a  polynomial in $n$ many times and $O(\log n)$ bits then are not sufficient to forward the flow value with sufficient accuracy. In order to apply the idea in the \CONGEST model, we instead use flow values from a sufficiently large (polynomial size) finite field. In \Cref{sec:congest-adapt}, we show that this suffices to w.h.p.\ obtain the same behavior as if flows for each edge were sent separately. 
	Aggregating flows thus allows to satisfy the congestion requirement, it however causes a number of further challenging problems, which we present and discuss next.
	 
	Let us consider the rather basic example of having only two nested reachable minimal odd cycles in \Cref{fig:flow3}. 
	Let $C$ denote the odd cycle with stem $s$ and $C'$ denote the odd cycle with stem $s'$. 
	Observe that $e$ is the middle edge of $C$ and $e'$ is the middle edge of $C'$. 
	The received flow of $e$ by $x$ must be sent to $y_1$ whereas the received flow of $e'$ by $x$ must be sent to $y_2$, which requires node $x$ to treat the two flows differently. 
	That is, node $x$ must recognize that the received flow of $e$ corresponds to the odd cycle containing the alternating path ending at edge $\set{y_1, x}$, and the received flow of $e'$ corresponds to the odd cycle containing the alternating path ending at edge $\set{y_2, x}$.
	This cannot be achieved due to the flow aggregation enforced by the congestion restriction. 
	Therefore, node $x$ is not capable of correctly directing the flows along the right paths so that the flows of $e$ only traverse the paths of cycle $C$ and the flows of $e'$ only traverse the paths of cycle $C'$.
		
	To resolve this issue and be able to still aggregate the flows, nodes should be able to treat all flows in the same way.
	Therefore, since every node knows its predecessors, we would like to establish the generic regulation of always sending flows only towards all predecessors no matter what the flow is. 
	However, by letting node $x$ send the received flow of $e'$ to its only predecessor $y_1$, the nodes in the alternating path between $x$ and $s'$ through $y_2$ do not anymore receive any flow of $e'$. 
	To fix this and keep node $x$ free of treating flows differently, we eliminate the flow generation over $e'$ and simulate it by generating a flow over $e''$. 
	That is, we shift the flow generation of cycle $C'$ from $e'$ to $e''$.
	Then, half of the flow of $e''$ is sent by $y_2$ to its predecessor $y_4$, and half of it is sent by $x$ to its predecessor $y_1$. 
	%However, by generating flow over $e''$ instead of $e'$, we still would like keep the desired effects that would have satisfied by the circulation of the flow of $e'$. 
	
	Previously, we let the flow generation only occur over the middle edge of an odd cycle, that can be easily recognized when an edge carries the same cluster ID in opposite directions in the same round.
	Now by having flow generation over both middle edges like $e$ as well as non-middle edges like $e''$, we need a more involved flow generation regulation. 
	Let every node send its cluster ID in at most one round over its matched edge and in at most one round over its unmatched edges. 
	A node sends its cluster ID over its matched edges in round $r$ if it is $(r-1)$-$0$-reachable, and it sends its cluster ID over its unmatched edges in round $r'$ if it is $(r'-1)$-$1$-reachable.
	We let a flow be generated over an edge when the two endpoints are not each other's predecessors and they both send the same cluster ID to each other, no matter if they are sent in the same round or not.
	Neither the endpoints of $e$ nor those of $e''$ are each other's predecessors while the endpoints send $f$ to each other over $e$ and $e''$. 
	Therefore, flow generation occur over both $e$ and $e''$, where $e$ is an example of a middle edge that the endpoints send cluster IDs in the same round, and $e''$ is an example of a non-middle edge for a shifted flow generation that the endpoints send cluster IDs in different rounds. 
	Also note that since $y_3$ is the predecessor of $y_6$, a flow is not anymore generated over $e'$ within this new regulation. 
	%Moreover, this ID transmission and flow generation guarantees that flow generation occurs over each edge at most once. 

	 \begin{figure}[t]
	 	\centering
  			\includegraphics[width=8cm]{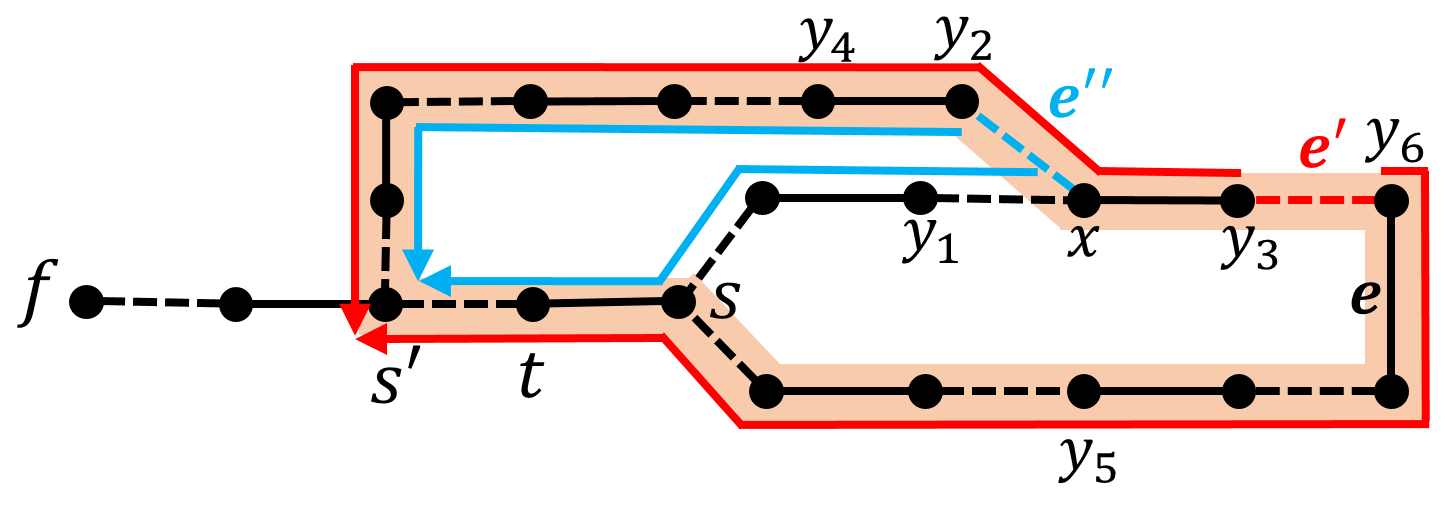}
  		\caption{Nested odd cycles: flow simulation of edge $e'$ on edge $e''$}
		\label{fig:flow3}
	 \end{figure}
 
	Now to see that shifting flow generation maintains the desired effects and avoids any side effects, let us compare the two cases of flow generation over $e'$ and its simulation over $e''$. 
	Each half of the flow of $e'$ is sent towards $s'$, one along the path between $y_6$ and $s'$ through $y_5$ and one along the path between $y_3$ and $s'$ through $y_4$ as depicted by arrows in \Cref{fig:flow3}. 
	In the simulation, each half of the flow of $e''$ is also sent towards $s'$, one along the path between $y_2$ and $s'$ through $y_4$ and one along the path between $x$ and $s'$ through $y_1$ again as depicted by arrows in \Cref{fig:flow3}.
	We need the simulation to serve the purposes of the flow generation of $e'$. 
	However, there are two crucial differences that might question the desired effects of flow $e'$ if we run the simulation instead. 
	To explain the first difference, consider the nodes inside odd cycle $C$. 
	Nodes like $y_1$ do not receive flows of $e'$ but receive flows of $e''$ in the simulation. 
	Moreover, nodes like $y_5$ receive flows of $e'$ but not flows of $e''$ in the simulation. 
	The second difference is that node $s'$ as the stem of $C'$ receives the whole unit flow of $e'$ in a single round as desired to perceive the fact that it is the corresponding stem and discards the flow. 
	However, in the simulation since $e''$ is not the middle edge of $C'$ and the flows are sent along paths with different lengths, $s'$ does not receive the whole unit flow of $e''$ in a single round. 
	This avoids node $s'$ to perceive the fact that it is the stem of an odd cycle and the flow is further sent by $s'$. 
	Let us see how crucial these differences are and how we can resolve them. 
	
	Regarding the first difference, the decisive observation is that whenever a node receives a proper fraction of a flow for the first time in round $r$, it detects the existence of an alternating path of length $r$. 
	Therefore, only the first receipt of such flow is important and must be at the right time for a node. 
	All those nodes in $C$ that differ in receiving the corresponding flows in the flow circulation of $e'$ and $e''$ already have received a proper fraction of flow $e$, and hence the receipt and the time of receiving later flows are irrelevant to them. 
	
	However, the second difference is crucial and needs to be resolved. 
	Note that the half flow of $e''$ is sent along the path between $y_2$ and $s'$ through $y_4$ that is the same path traversed by the half flow of $e'$. 
	Therefore, if $y_2$ sends the half flow of $e''$ to $y_4$ immediately in the next round of receiving the cluster ID from $x$, it reaches $s'$ at exactly the same time as the half flow of $e'$ would have reached $s'$. 
	Now assume that $x$ would also have sent the other half flow of $e''$ along the path between $x$ and $s'$ through $y_5$.
	If $x$ would have sent the flow in the very next round of receiving the cluster ID from $y_2$, then the flow of $e''$ would also have  reached $s'$ in exactly the same round as the flow of $e'$ would have reached $s'$.
	However, since $x$ actually sends this flow along the path between $x$ and $s'$ through $y_1$, it reaches $s'$ sooner. 
	This time difference is the difference of the length of the shortest alternating path between $x$ and $f$ ending in the matched edge of $x$ and such a path ending in the unmatched edge of $x$. 
	This difference is known by $x$, and $x$ can therefore delay sending the flow by this number of rounds and repair the unwanted side effects of the simulation. 
	Note that $x$ cannot send the flow in the very next round of receiving the cluster ID from $y_2$ since it has not yet decided at that time to send its cluster ID to $y_2$ and hence cannot yet recognize the flow generation over $e''$. 
	Therefore, it has to anyway send the flow along a shorter path, e.g., the path through $y_1$. 
	
	This discussed simple example inevitably abstracts away some details. 
	In the example, flows are only sent over alternating paths.
	However, if nodes always send flows to their predecessors, flows do not necessarily traverse alternating paths and the paths along which flows are sent might have consecutive unmatched edges. 
	Then, along a path that a flow is forwarded, every node that has two adjacent unmatched edges on the path delays forwarding the flow. 
	We postpone further details to \Cref{sec:3,sec:analysis}.

%%	

%%% Local Variables:
%%% mode: latex
%%% TeX-master: "main"
%%% End:

%% file: aug-finding.tex
% Algorithm
%\vspace{-.3cm}
\section{Shortest Augmenting Path Detection}\label{sec:3}

In this section, we present the algorithm to detect a shortest
augmenting path in time linear in the length of the path in the \CONGEST model. 
	The organization of this section is as follows. In
\Cref{sec:clustering-definitions}, we formally define the free node
clustering that was described in \Cref{sec:clusteringoutline}. We
define the clustering by giving a deterministic sequential algorithm that
constructs the clustering in a step-by-step manner. Note that
this deterministic algorithm is only for the purpose of providing a
precise definition of the clustering.  Then, in \Cref{sec:path-setup},
we show that given such a clustering, at least one shortest augmenting path can be detected in
a single round of the \CONGEST model.  In
\Cref{sec:distributed-clustering}, we provide a distributed algorithm
to construct the free node clustering. Due to lack of space, the
analysis of the distributed free node clustering algorithm appears in
in \Cref{sec:analysis}.  For the
sake of simplicity, we first consider no restriction on the size of sent
messages when we describe the algorithm and present its analysis in  \Cref{sec:analysis}.
We postpone dealing with the message size restriction to
\Cref{sec:congest-adapt}, where we then show how one can employ
randomness to adapt the algorithm to the \CONGEST model (along the
lines described in \Cref{sec:clusteringoutline}).

%%%%%%%%%% SECTION: r-Radius Free Node Clustering %%%%%%%%%%
\subsection{ The \boldmath{$r$}-Radius Free Node Clustering}\label{sec:clustering-definitions}	
		
To define the $r$-radius free node
clustering of a graph $G$ w.r.t.\ a given matching $M$ of $G$,
we introduce a deterministic $r$-step algorithm, which we 
henceforth call the FNC algorithm,.  The free nodes
$f_1, \dots, f_\rho$ are the cluster centers.  For
all $i$, let $C_i$ denote the cluster that is centered at free node
$f_i$.  Initially every cluster $C_i$ only contains $f_i$ and during
the execution of the algorithm more nodes potentially join the
cluster.  For consistency, we assume that there exists a
step $0$ in which every free node joins the cluster centered at
itself, i.e., initially $\forall i\in[1, \rho]: C_i=\set{f_i}$.
Then, in every step $t\geq 1$, every node that has not yet joined any
cluster, concurrently joins the cluster centered at $f_i$ if and only
if $f_i$ is the minimum-ID free node from which $v$ has an alternating
path $P$ of length $t$ such that
$\mathbb{V}(P)\setminus \set{v} \subseteq C_i$.\footnote{Note that
  there might be some nodes in $G$ that never join any cluster in any
  step of the FNC algorithm.}  Throughout, let $C_i(t)$ denote the set
of nodes in cluster $C_i$ after $t$ steps of the FNC algorithm. We define $\calC(r) := \set{C_1(r), \dots , C_\rho(r)}$ to be the
$r$-radius free node clustering of $G$.  See
\Cref{alg:clustering} for the details of the FNC algorithm.

\begin{algorithm}[h]
  \caption{$r$-Radius Free Node Clustering (FNC) Algorithm}\label{alg:clustering}
  \SetKwInOut{Input}{Input}\SetKwInOut{Output}{Output}
  \Input{Graph $G=(V,E)$, matching $M$ of $G$, and integer $r\geq 1$}
  \Output{$r$-Radius free node clustering of $G$}
  \BlankLine

  $V' \leftarrow V\setminus\{f_1, \dots, f_\rho\}$ \;
  \ForAll{$i\in \set{1, \dots, \rho}$}{ 
  	$C_i\leftarrow\set{f_i}$ \; 
  }
  \ForAll{$t\in\set{1, 2, \dots, r}$}{   
  	\ForAll{$i\in \set{1, \dots, \rho}$}{
		$B_i := \emptyset$\;
	} 
	\ForAll{$v\in V'$}{
		$X \leftarrow \{j \ | \ v$ has an alternating path $P$ of length $t$ s.t. $\mathbb{V}(P)\setminus \set{v}\subseteq C_j\}$\; 
		
		\If{$X \neq \emptyset$}{
			$j' \leftarrow \underset{j\in X}{\mathrm{argmin} \ f_j}$\;
			$B_{j'} \leftarrow B_{j'} \cup \set{v}$\; 
			$V' \leftarrow V' \setminus \set{v}$\;
		}
	}
	\ForAll{$i\in \set{1, \dots, \rho}$}{
		$C_i := C_i \cup B_i$\;
	}
  }
  \Return{$\mathcal{C}(r) \leftarrow \set{C_1, \dots,C_\rho}$};
\end{algorithm}

\noindent	
To simplify the discussions, we introduce the following definitions
and terminology.  In the following definitions, $v$ is an
arbitrary node in $G$ and $\vartheta$ is an arbitrary integer in $\set{0,1}$.  We say that
$P$ is a \textit{path of $v$} or $v$ \textit{has a path} $P$ if $P$ is
a path starting at a free node and ending at node $v$.
		
\begin{definition}[Uniform Paths]
  We say that a path $P$ is uniform at time $t\geq 0$ if
  $\mathbb{V}(P)\subseteq C_i(t)$ for some
  $i\in\set{1,\dots,\rho}$. When the time $t$ is clear from the
  context, we just say that $P$ is uniform.
\end{definition}

\noindent
The following lemma is a simple observation about uniform alternating paths.

\begin{lemma}\label{lem:uniform}
	Let $P$ be an alternating path of length $r$ from any free node to any node. 
	If there is any time (possibly larger than $r$) at which $P$ is uniform, then $P$ is uniform at time $r$.
\end{lemma}
\begin{proof}
	For a free node $f_i$ and a node $v$, let $\langle v=v_r, \dots, v_0=f_i \rangle$ be the given path $P$. 
	Let $t$ be an integer such that $P$ is completely contained in cluster $C_i$ after $t$ steps of the FNC algorithm.   
	That is, all nodes of $P$ eventually join the cluster centered at $f_i$. 
	Thus, since nodes commit to the clusters they join, the nodes of $P$ cannot join any other clusters. 
	Let us now by induction show that for all $r'\leq r$, $v_{r'}\in C_i(r')$. 
	Since $v_0\in C_i(0)$, $v_1$ joins $C_i$ in the first step of the FNC algorithm, i.e., $v_1\in C_i(1)$.  
	Now fix an arbitrary integer $r'\leq r$ such that for all $r'' < r'$, $v_{r''}\in C_i(r'')$.  
	Node $v_{r'}$ can only join cluster $C_i$. 
	If it already joined the cluster in the first $r'-1$ steps of the FNC algorithm, then $v_{r'} \in C_i(r')$ since for all $t'\geq 0$, $C_i(t') \subseteq C_i(t'+1)$. 
	Otherwise, node $v_{r'}$ must join cluster $C_i$ in step $r'$ of the FNC algorithm, i.e., $v_{r'} \in C_i(r')$. 
	Hence, we can conclude that for all $r'\leq r$, $v_{r'}\in C_i(r')$.
	That is, for all $r'\leq r$, $v_{r'}\in C_i(r)$, and consequently $P$ is uniform at time $r$.
	
\end{proof}

\begin{definition}[Almost Uniform Paths]
  We say that a path $P$ of $v$ is almost uniform (at time $t$) if
  $\mathbb{V}(P)\setminus \set{v}$ is uniform (at time $t$). Note that every uniform path of $v$ is also almost uniform.
\end{definition}

\begin{definition}[$\vartheta$-Edges]
  A free (i.e., an unmatched) edge is called a $0$-edge, and a matched edge is called a $1$-edge.   
\end{definition}
		
\begin{definition}[$\vartheta$-Paths]
  An alternating path $P$ of $v$ is called a $\vartheta$-path of $v$ if $P$ contains a $\vartheta$-edge adjacent to $v$. 
\end{definition}

%\begin{definition}[$\vartheta$-Parents]
%  We say $u$ is a $\vartheta$-Parent of $v$ if $u$ is the neighbor
%  of $v$ on a shortest uniform $\vartheta$-path of $v$. 
%\end{definition}
		
\begin{definition}[Predecessors]
  We say $u$ is a \textit{predecessor} of $v$ if $u$ is the neighbor
  of $v$ on a shortest uniform alternating path of $v$. %Then, $v$ is called a \textit{successor} of $u$.
\end{definition}

%Note that the set of $v$'s predecessors is either the set of $v$'s $0$-parents or the set of its $1$-parents. 
		
\begin{definition}[Reachability]
  For an integer $r\geq 0$, we say that $v$ is $r$-$\vartheta$-reachable if
  $v$ has a shortest uniform $\vartheta$-path of length $r$.
  Moreover, we say that $v$ is $r$-reachable if it has a shortest uniform
  alternating path of length $r$.
\end{definition} 
For the sake of consistency, we assume that every free node is $0$-$0$-reachable and $0$-$1$-reachable. 

%%%%%%

\subsection{Detecting a Shortest Augmenting Path}\label{sec:path-setup}

	In this section, we show how the existence of a shortest
        augmenting path of length at most $\ell$ can be detected in a
        single round of the \CONGEST model if the nodes of $G$ are provided with the $\ell$-radius free node clustering with respect to a given matching $M$ of $G$ and if in addition the $\ell$-radius free node clustering is \textit{well-formed} in the following sense.

\begin{definition}[Well-Formed Clustering]
  The $r$-radius free node clustering $\mathcal{C}(r)$ is said to be
  well-formed in a distributed setting if for all $r'\leq r$,
  $\vartheta\in \set{0,1}$ and $i$, every $r'$-$\vartheta$-reachable
  node $v\in C_i(r)$, beyond knowing its cluster ID, knows its predecessors and the
  fact of being $r'$-$\vartheta$-reachable.
\end{definition}
 
\noindent 
	Before proving the aforementioned claim that is stated in the next lemma, consider the following definition. 

\begin{definition}[Rank of an Augmenting Path]
  For an integer $\ell$, consider an arbitrary augmenting path $P$ of
  length $\ell$ between any pair of free nodes $f_{s}$ and $f_{t}$.
  Let us name the nodes of $P$ as
  $\langle f_{s}=u_0, \dots, u_{\ell}=f_{t} \rangle$ and let 
  $i$ and $j$ be the largest integers such that the subpaths
  $P[f_{s}, u_{i}]$ and $P[u_{\ell-j}, f_{t}]$ are uniform.  Then, we
  define the \textit{rank} of $P$ to be $i+j$.
\end{definition}

	\begin{lemma}\label{lem:detection}
		Let all nodes of a given graph $G$ be provided with the well-formed $r$-radius free node clustering with respect to a given matching $M$ of $G$.
		If there is a shortest augmenting path of length $\ell\leq r$, then a shortest augmenting path can be detected in a single round of the \CONGEST model. 
	\end{lemma}
	\begin{proof}
		Let us assume that there is a shortest augmenting path that is partitioned into two consecutive subpaths such that each subpath is completely contained in a different cluster. 
		Then, since all the nodes are provided with the well-formed $r$-radius free node clustering, the two neighboring nodes at the end of the subpaths can detect the existence of the shortest augmenting path in just a single round of the \CONGEST model.
		To do so, they only need to inform each other about their cluster IDs and their reachabilities in a single round.
		Hence, to prove the lemma it is sufficient to prove that there is a shortest augmenting path that is partitioned into two consecutive subpaths such that each subpath is completely contained in a different cluster. 
		
	Let us assume that the shortest augmenting path is of length $2k+1$.
	Throughout the proof, we consider that the graph is provided with an $r$-radius free clustering for some $r\geq 2k+1$. 
  	Then, to prove the lemma, it is enough to show that there is a shortest augmenting path of rank $2k$.
  	Let us first prove the following helping claim. 
\begin{addmargin}[2em]{0em}
    \begin{claim}\label{clm:shortest-path}
      If an arbitrary node $v$ has an almost uniform $\vartheta$-path of length $k$ that is not uniform, then $v$ has a uniform
      $\vartheta$-path of length at most $k$.
    \end{claim}
    \begin{proof}[Proof of \Cref{clm:shortest-path}]
    	Let $P_1$ be a path from a free node $f_i$ to $v$ that is almost uniform but not uniform. 
	At time $r$, $P_1[f_i, w)$ is uniform since $P_1$ is almost uniform. 
	Moreover, $P_1[f_i, w)$ is of length $k-1$.
      	Therefore, based on \Cref{lem:uniform}, $P_1[f_i, v)$ is uniform at time $k-1$, i.e., $\mathbb{V}(P_1[f, v))\subseteq C_i(k-1)$. 
	Since $v$ has an almost uniform alternating path, it joins some cluster in step $k$ of the FNC algorithm if it has not yet joined any cluster.  
	Hence, we can say that $v$ joins some cluster $C_j$ in the first $k$ steps of the FNC algorithm, i.e., $v\in C_j(k)$ for some $j$.  
	Therefore, $v$ has a shortest uniform alternating path $P_2$ of length at most $k$.  
	Note that $P_1[f_i, v)$ is completely contained in $C_i$, and $P_2$ in $C_j$.
	Hence $P_1[f_i, v)$ and $P_2$ are disjoint. 
	Thus, if $P_2$ is a $(1-\vartheta)$-path of $v$, then the concatenation of $P_1$ and $P_2$ would be an augmenting path of length less than $2k+1$, which contradicts the assumed length of the shortest augmenting path.  	Therefore, $P_2$ is a uniform $\vartheta$-path of $v$, whose length is at most $k$.
	
    \end{proof}      
  \end{addmargin}

  \noindent For the sake of contradiction, let us assume that the
  highest rank of any shortest augmenting path in $\calG$ is $T<2k$.  Let
  $P=\langle f_s=v_0, v_1, \dots, v_{2k+1} =f_t \rangle$ be an
  arbitrary shortest augmenting path of rank $T$.  Then, let $i$ and
  $j$ be the largest integers such that all the nodes of $P[f_s, v_i]$
  and $P[v_{2k+1-j}, f_t]$ are respectively in clusters $C_s$ and
  $C_t$, which implies $i+j=T<2k$.  Without loss of generality, let us
  assume that $i\leq j$.  To show a contradiction to our assumption on
  the highest rank of any shortest augmenting path, let us first prove
  the following claim.
\begin{addmargin}[2em]{0em}
    \begin{claim}\label{clm:allinFt}.
      For $\psi:= \max \left\{ i+1, \ \min\set{j, 2k-j} \right\}$, all the nodes of $P[v_{i+1}, v_\psi]$ are in $C_t$. 
    \end{claim}
    \begin{proof}[Proof of \Cref{clm:allinFt}]
      	We prove the claim by induction on $z\in [i+1, \psi]$.  Let us first show that $v_{i+1} \in C_t$ as the induction base.  
	Due to the choice of
      $i$, node $v_{i+1}$ is not in cluster $C_s$.  Therefore,
      $P[f_s, v_{i+1}]$ is an almost uniform $(i\mod{2})$-path of
      $v_{i+1}$ from $f_s$ while $v_{i+1}$ is in a cluster centered at a
      different free node than $f_s$.  Based on
      \Cref{clm:shortest-path}, $v_{i+1}$ thus has a shortest
      uniform $(i\mod{2})$-path $P_{i+1}$ of length at most $i+1$.
      
      \begin{center}		
	\begin{tikzpicture}
          \draw (-5,0) -- (5, 0);
          \draw[dashed] (0,.5) -- (0, -.5);
          \shade[shading=ball, ball color=blue] (-5,0) circle (.08);
          \node at (-5,.3) {$f_s$};  
          \shade[shading=ball, ball color=red] (5,0) circle (.08);
          \node at (5,.3) {$f_t$}; 
          \shade[shading=ball, ball color=black] (-1,0) circle (.08);
          \node at (-1,.3) {$v_\psi$};
          \shade[shading=ball, ball color=blue] (-3,0) circle (.08);
          \node at (-3,.3) {$v_i$};
          \shade[shading=ball, ball color=black] (-2.5,0) circle (.08);
          \node at (-2.35,.3) {$v_{i+1}$};
          \shade[shading=ball, ball color=red] (1,0) circle (.08);
          \node at (1.5,.3) {$v_{\tiny{2k-j+1}}$};
          \draw[color=blue] (-5, -.2) to [bend right=20] (-3, -.2);
          \draw[color=red] (1, -.2) to [bend right=20] (5, -.2);
          \node[blue] at (-4,-.2) {$\in C_s$};
          \node[red] at (3,-.35) {$\in C_t$};
          \node[blue] at (-4,.1) {$\dots$};
          \node[red] at (3,.1) {$\dots$};
          \node at (-1.7,.1) {$\dots$};
          \node at (0,.1) {$\dots$};
	\end{tikzpicture}	
      \end{center}
      
      For the sake of contradiction, let us assume that $v_{i+1}$ is
      in cluster $C_m$ for $m\neq t$.  Moreover, path
      $P[v_{i+1}, f_t]$ is an $((i+1)\mod{2})$-path of length $2k-i$ of
      $v_{i+1}$.  Therefore, $P'$, the concatenation of $P_{i+1}$ and
      $P[v_{i+1}, f_t]$, is an alternating walk of length at most
      $2k+1$ between $f_m$ and $f_t$.  Note that $P'$ cannot be a path
      since otherwise it would be an augmenting path that is of length less
      than $2k+1$ or rank more than $i+j$.  It then contradicts at
      least one of the two choices of $k$ or $T$.  Therefore,
      $P_{i+1}$ and $P[v_{i+1}, f_t]$ must have a common node in
      addition to $v_{i+1}$.  Any common node of $P_{i+1}$ and
      $P[v_{i+1}, f_t]$ must be in $P[v_{i+1}, v_{2k-j}]$.  That is
      because $P[v_{2k-j+1}, f_t] \in C_t$ and $P_{i+1} \in C_m$,
      where $m\neq t$.  Let $v_x\in P[v_{i+2}, v_{2k-j}]$ be the
      closest node to $f_m$ on $P_{i+1}$.  Note that
      $max\set{|P[f_s, v_x]|, |P[v_x, f_t]|}\leq 2k-i$ since
      $i+1<x\leq 2k-j$ and $i\leq j$.  Moreover, since
      $|P_{i+1}|\leq i+1$ and $v_x\neq v_{i+1}$, path
      $P_{i+1}[v_x, f_m]$ is of length at most $i$.  Hence, since
      $m\neq t$ and $m\neq s$, the concatenation of
      $P_{i+1}[v_x, f_m]$ and either $P[f_s, v_x]$ or $P[v_x, f_t]$ is
      an augmenting path of length less than $2k+1$, which contradicts
      the choice of $k$.  Therefore, $v_{i+1}\in C_t$, i.e., $m=t$.
      Note that if $\psi=i+1$, the proof of the
      \Cref{clm:allinFt} is already completed.  Therefore, for
      the remainder of the claim's proof, let us assume
      $\psi \neq i+1$, i.e., $\psi=\min\set{j, 2k-j}$, and conclude
      the proof by showing the induction step.

      Regarding the induction step, for an arbitrary integer
      $z\in [i+1, \psi-1]$, let us assume that all the nodes of
      $P[v_{i+1}, v_z]$ are in cluster $C_t$ and prove that node
      $v_{z+1}$ is in $C_t$ too.  Let us first show that $v_{z+1}$ has
      an almost uniform $(z\mod{2})$-path of length at most $z+1$ from
      $f_t$.  Let $v_y\in P[v_{i+1}, v_z]$ be the closest node to
      $f_t$ on $P_{i+1}$.  $P[f_s, v_y]$ is of length less than $j$
      since $y\leq z< \psi\leq j$.  Moreover, since
      $|P_{i+1}|\leq i+1$, path $P_{i+1}[v_y, f_t]$ is of length at
      most $i+1$.  Hence, if $P_{i+1}[v_y, f_t]$ is a
      $((y+1)\mod{2})$-path of $v_y$, then the concatenation of
      $P[f_s, v_y]$ and $P_{i+1}[v_y, f_t]$ would be an augmenting
      path of length $j+i+1<2k+1$, contradicting the choice of $k$.
      Therefore, $P_{i+1}[v_y, f_t]$ is a $(y\mod{2})$-path of $v_y$.
      Hence, $P'$, the concatenation of $P_{i+1}[v_y, f_t]$ and
      $P[v_y, v_{z+1}]$ is an almost uniform $(z\mod{2})$-path of
      $v_{z+1}$ from $f_t$ whose length is at most $z+1$.
			
      For the sake of contradiction, let us assume that $v_{z+1}$ is
      in cluster $C_q$ for $q\neq t$.  Then, $v_{z+1}$ is in a
      different cluster than $C_t$ and it has an almost uniform
      $(z\mod{2})$-path of length at most $z+1\leq k$ from $f_t$.
      Hence, based on \Cref{clm:shortest-path}, $v_{z+1}$ has a
      uniform $(z\mod{2})$-path of length at most $z+1$ from $f_q$,
      denoted by $P_{z+1}$.  Now let $v_{y'}\in P[v_{z+1}, v_{2k-j}]$
      be the closest node to $f_q$ on $P_{z+1}$.
      $P_{z+1}[v_{y'}, f_q]$ must be a $(y'\mod{2})$-path of $v_{y'}$
      since otherwise the concatenation of $P[v_{y'}, f_t]$ and
      $P_{z+1}[v_{y'}, f_q]$ would be an augmenting path that is of length
      less than $2k+1$ or rank more than $T$.  Moreover, let
      $v_{y''}\in P[v_{i+1}, v_{y'}]$ be the closest node to $f_t$ on
      $P_{i+1}$.  Then, $P_{i+1}[v_{y''}, f_t]$ is a
      $((y''+1)\mod{2})$-path of $v_{y''}$ since otherwise the
      concatenation of $P_{i+1}[v_{y''}, f_t]$ and $P[f_s, v_{y''}]$
      would be an augmenting path of length less than $2k+1$ or rank
      more than $T$.  Note that $|P_{i+1}[v_{y''}, f_t]|\leq i+1$ and
      $|P_{z+1}[v_{y'}, f_q]|\leq z+1\leq \psi \leq j$.  Therefore,
      since $P[v_{y''}, v_{y'}]$ is of length less than $2k-i-j$, the
      concatenation of $P_{i+1}[v_{y''}, f_t]$, $P[v_{y''}, v_{y'}]$
      and $P_{z+1}[v_{y'}, f_q]$ is an augmenting path of length less
      than $2k+1$ or rank more than $T$, contradicting the choice of
      $k$ or $T$.     
       
    \end{proof}
  %\end{minipage}
\end{addmargin}  

  \noindent Now let us first consider the case $j\geq k$.  This
  implies that $2k-j = \min\set{j, 2k-j}$.  Moreover, since $i+j<2k$,
  it holds that $i+1 \leq 2k-j$.  Therefore, $\psi = 2k-j$.  Then,
  based on Claim 2, $v_{2k-j}\in C_t$, which
  contradicts the choice of $j$ and concludes the proof.  Let us then
  assume that $j<k$.  Due to the choice of $j$, let us assume that
  $v_{2k-j}$ is in cluster $C_w$ for $w\neq t$.  Moreover, path
  $P[v_{2k-j}, f_t]$ is an almost uniform $((2k-j)\mod{2})$-path of
  $v_{2k-j}$ from $f_t$ whose length is $j+1$.  Therefore, based on
  Claim 1, $v_{2k-j}$ has a uniform
  $((2k-j)\mod{2})$-path $P_{2k-j}$ of length at most $j+1$ from $f_w$.  Let $v_x\in P[v_{\psi+1}, v_{2k-j}]$ be the closest
  node to $f_w$ on $P_{2k-j}$.  Then, $P_{2k-j}[v_x, f_w]$ is a
  $(x\mod{2})$-path of $v_x$ since otherwise the concatenation of
  $P_{2k-j}[v_x, f_w]$ and $P[v_x, f_t]$ would be an augmenting
  path of length less than $2k+1$.  
  Based on Claim 2, node $v_{i+1}$ is in cluster $C_t$ but not $C_s$.  Therefore,
      $P[f_s, v_{i+1}]$ is an almost uniform but not uniform $(i\mod{2})$-path of
      $v_{i+1}$ from $f_s$.  Based on
      Claim 1, $v_{i+1}$ thus has a shortest
      uniform $(i\mod{2})$-path $P_{i+1}$ of length at most $i+1$.  
  Let $v_{x'}\in P[v_{i+1}, v_{x}]$
  be the closest node to $f_t$ on $P_{i+1}$.  Then,
  $P_{i+1}[v_{x'}, f_t]$ is a uniform $((x'+1)\mod{2})$-path of
  $v_{x'}$ whose length is at most $i+1$.  Then, the concatenation of
  $P_{i+1}[v_{x'}, f_t]$, $P[v_{x'}, v_x]$ and $P_{2k-j}[v_x, f_w]$ is
  an augmenting path that is of length less than $2k+1$ or rank more than
  $i+j$ between $f_t$ and $f_w$, contradicting the choice of $k$ or
  $T$.
    
\end{proof}			
				 
\subsection{Distributed Free Node
  Clustering}\label{sec:distributed-clustering}

In this section, we present a distributed deterministic algorithm whose
$r$-round execution provides the well-formed $r$-radius free node
clustering.  This algorithm uses large messages. However, in
\Cref{sec:congest-adapt}, we show how to use randomness to adapt this algorithm
to the \CONGEST model.  The algorithm makes all the free nodes
(cluster centers) propagate their own IDs along their shortest
alternating paths. It can essentially be seen as a multi-source
breadth first search graph exploration.
%, where the the IDs of the free nodes are propagated to the nodes that have an alternating path of length $i$ to a free node in round $i$. 
To correctly develop
the well-formed free node clustering, it is crucial that the free node
IDs, that we call \textit{tokens}, only traverse paths but not walks
with cycles.  %However, enforcing the tokens to avoid traversing odd-length cycles---which are inevitable in general graphs---turns out to be a challenging task.  
The algorithm succeeds in preventing the tokens to traverse odd-length cycles by a technique of generating and circulating flows throughout the network as we see in the sequel.

\vspace{.3cm}
\para{Distributed \boldmath{$r$}-Radius Free Node Clustering: DFNC
  Algorithm}
  \vspace{.2cm}

\noindent
The algorithm is run for $r$ rounds.  Let the following
variables be maintained by the nodes during the
execution; $r^{(0)}_v$ and $r^{(1)}_v$ keep track of the $v$'s
reachabilities, $cid_v$ holds the cluster ID of $v$, and $pred_v$
holds the set of $v$'s predecessors.  At the beginning of the
execution, for every free node $v$, variables $r^{(0)}_v$ and
$r^{(1)}_v$ are set to $0$, variable $cid_v$ is set to $v$, and
$pred_v$ is set to $\emptyset$.  Moreover, for every matched node, all
these variables are initially undefined and set to $\perp$.  Every
node $v$ participates in the token dissemination based on the
following simple rule. For an arbitrary integer $t\geq 1$, in round
$t$:
	\begin{itemize}
		\item If $r^{(0)}_v=t-1$, then $v$ sends $cid_v$ over its adjacent $1$-edge (if any). Otherwise, if $r^{(1)}_v=t-1$, then $v$ sends $cid_v$ over all its adjacent $0$-edges (if any).
	\end{itemize}
\noindent
	Then, based on the above simple rule, every node sends tokens to its neighbors in at most two rounds, at most once over its $1$-edge and at most once over its $0$-edges. 
	Token forwarding for a node depends on its variables considering the above simple rule. 
	We already explained how the variables are set for a free node.
	Therefore, in the first round of the execution every free node sends its ID to all its neighbors. 
	Now let us explain how a matched node sets its variables, i.e., which cluster it joins and how it detects its reachabilities and its predecessors. 
	Let round $t$ be the first round in which a node $v$ receives tokens. 
	Let $\tau_1, \dots, \tau_j$ be the tokens that $v$ receives from its neighbors in round $t$.
	Then, $v$ sets $cid_v$ to $\min_i \tau_i$, and subsequently sets $pred_v$ to the set of all its neighbors that sent $cid_v$ to $v$ in round $t$. 
	Now let us explain how node $v$ sets its variables $r^{(1)}_v$ and $r^{(0)}_v$. 
	There are two types of messages sent by the nodes throughout the execution; tokens (i.e., free node IDs) and \textit{flow messages}.  
	Node $v$ sets $r^{(1)}_v$ and $r^{(0)}_v$ based on the received tokens and flow messages. 
	Before we explain how these variables are set by $v$, let us first define the flow messages by explaining flow generation and circulation throughout the network.
	\\[.1cm]
	\noindent
	\textit{Flow Generation:}
		A \textit{flow} is a key-value pair, where the key is an edge and the value is a real number in $[0,1]$. 
		Any flow whose key is some edge $e$ is simply called a \textit{flow of edge $e$}.
		A \textit{flow message} is then defined to be a set of flows (i.e., key-value pairs) that are sent by a node to its neighbor in a round. 
		A flow generation is an event that can only happen
                over an edge for which both endpoints belong to the same cluster. 
		Let us fix an arbitrary edge $e=\set{u,w}$ where both endpoints belong to the same cluster. 
		Then, we say that a flow is generated over edge $e$ if and only if (1) none of $u$ or $w$ is the other one's predecessor, and (2) both $u$ and $w$ send tokens to each other. 
		Note that we only consider at most one flow generation for every edge throughout the whole execution.
		Let us assume that $u$ and $w$ are not each other's predecessors, $u$ sends token to $w$ in round $r_u$, and $w$ sends token to $u$ in round $r_w$. 
		Then, the flow generation over $e$ is defined as an event in which $u$ receives a singleton flow message $\set{(e, 1/2)}$ over $e$ in round $r_w$ and $w$ receives a singleton flow message $\set{(e, 1/2)}$ over $e$ in round $r_u$. 
		It is important to note that nodes $u$ and $w$ might send tokens to each other in different rounds, i.e., $r_u\neq r_w$.	
		However, they cannot perceive the flow generation over $e$ before they make sure that $e$ carries tokens in both directions.
		In particular, if $r_w<r_u$, node $u$ cannot in round $r_w$ perceive the flow receipt over $e$ since it does not yet know whether it will send a token to $w$.
		Hence, $u$ will perceive this flow receipt of round $r_w$ later in round $r_u$ in which it decides to send token over $e$ and then knows that the edge carries tokens in both directions. 
		However, we will see that $u$ does not need to know about the flow receipt of round $r_w$ before round $r_u$. 
	\\[.1cm]
	\noindent	
	\textit{Flow Circulation:} 
	No matter if it receives a flow over its adjacent $0$-edges or its adjacent $1$-edge, every node always forwards the received flow to its predecessors by equally splitting the flow value among them.
	When the edge over which $v$ receives a flow and the edges connected to its predecessors are not all $0$-edges (see \Cref{fig:forwarding-1} and~\ref{fig:forwarding-2}), $v$ forwards the flow immediately in the next round after receiving the flow.
	Otherwise (see \Cref{fig:forwarding-3}), it delays forwarding the flow for $r^{(1)}_v - r^{(0)}_v$ rounds.
	A node furthermore avoids forwarding the whole flow of a
        single edge $e$ in a single round (i.e., a flow of value $1$ of $e$).  
	Let us see the details of the flow circulation in the following. 
	
		\begin{figure}[t]
		\centering
		\begin{subfigure}{.3\textwidth}
 			\centering
  			\begin{tikzpicture}
				\draw[lightgray,-latex,line width=4mm] (-.4,0) -- (2.5,0);
				\node at (-.1,.5) {\scriptsize{flow}};
				\node at (0,.3) {\scriptsize{direction}};
			
				\shade[shading=ball, ball color=black] (0,0) circle (.1);
				\shade[shading=ball, ball color=black] (2,.2) circle (.1);
				\shade[shading=ball, ball color=black] (2,.7) circle (.1);
				\shade[shading=ball, ball color=black] (2,-.7) circle (.1);
				
				\draw[line width=.4mm, dashed] (1,0) - - (2,.2);
				\draw[line width=.4mm, dashed] (1,0) - - (2,.7);
				\draw[line width=.4mm, dashed] (1,0) - - (2,-.7);
				\node at (2,-.18) {$\vdots$};
				
				\node[red] at (2.2,.3) {$*$};
				\node[red] at (2.2,.8) {$*$};
				\node[red] at (2.2,-.6) {$*$};
				
				\coordinate (u) at (1,0);
				\shade[shading=ball, ball color=black] (u) circle (.1);
				\node at (1,.25) {\small{$u$}};

				\draw[line width=.4mm] (u) - - (0,0);			

			\end{tikzpicture}
  			\caption{}
  			\label{fig:forwarding-1}
		\end{subfigure}
		\begin{subfigure}{.3\textwidth}
 			\centering
  			\begin{tikzpicture}
				\draw[lightgray,-latex,line width=4mm] (-.4,0) -- (2.5,0);

				\shade[shading=ball, ball color=black] (2,0) circle (.1);
				\shade[shading=ball, ball color=black] (0,.2) circle (.1);
				\shade[shading=ball, ball color=black] (0,.7) circle (.1);
				\shade[shading=ball, ball color=black] (0,-.7) circle (.1);
				
				\draw[line width=.4mm, dashed] (1,0) - - (0,.2);
				\draw[line width=.4mm, dashed] (1,0) - - (0,.7);
				\draw[line width=.4mm, dashed] (1,0) - - (0,-.7);
				\node at (0,-.18) {$\vdots$};
				
				\node[red] at (2.2, .2) {$*$};
				
				\coordinate (u) at (1,0);
				\shade[shading=ball, ball color=black] (u) circle (.1);
				\node at (1,.25) {\small{$u$}};

				\draw[line width=.4mm] (u) - - (2,0);					
			\end{tikzpicture}
  			\caption{}
  			\label{fig:forwarding-2}
		\end{subfigure}
		\begin{subfigure}{.3\textwidth}
 			\centering
  			\begin{tikzpicture}
				\draw[lightgray,-latex,line width=4mm] (-.4,0) -- (2.5,0);
				\shade[shading=ball, ball color=black] (0,.2) circle (.1);
				\shade[shading=ball, ball color=black] (0,.7) circle (.1);
				\shade[shading=ball, ball color=black] (0,-.7) circle (.1);
				\shade[shading=ball, ball color=black] (2,.2) circle (.1);
				\shade[shading=ball, ball color=black] (2,.7) circle (.1);
				\shade[shading=ball, ball color=black] (2,-.7) circle (.1);		
				\draw[line width=.4mm, dashed] (1,0) - - (0,.2);
				\draw[line width=.4mm, dashed] (1,0) - - (0,.7);
				\draw[line width=.4mm, dashed] (1,0) - - (0,-.7);
				\node at (0,-.18) {$\vdots$};				
				\draw[line width=.4mm, dashed] (1,0) - - (2,.2);
				\draw[line width=.4mm, dashed] (1,0) - - (2,.7);
				\draw[line width=.4mm, dashed] (1,0) - - (2,-.7);
				\node at (2,-.18) {$\vdots$};				
				\node[red] at (2.2,.3) {$*$};
				\node[red] at (2.2,.8) {$*$};
				\node[red] at (2.2,-.6) {$*$};		
				\coordinate (u) at (1,0);
				\shade[shading=ball, ball color=black] (u) circle (.1);
				\node at (1,.25) {\small{$u$}};
			\end{tikzpicture}
  			\caption{}
  			\label{fig:forwarding-3}
		\end{subfigure}
	\caption{The 3 possibilities of flow forwarding. The $u$'s predecessors are marked by asterisks.}
	\label{fig:d}
	\end{figure}
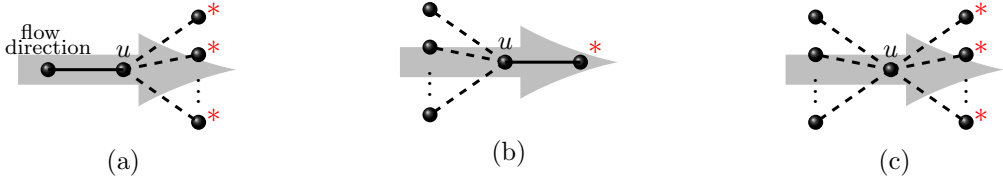

	Let $I_v(t)$ denote the set of all the flows that a node $v$ receives in a round $t$, i.e., the union of all the received flow messages by $v$ in round $t$.
	Moreover, let $O_v(t,e)$ denote the \textit{output buffer} of a node $v$ for its adjacent edge $e$ in a round $t$, which is initially an empty set and eventually sent as a flow message over $e$ by $v$ in round $t$.
	Now let us fix an arbitrary node $v$, where $E_v$ is the set of $v$'s adjacent edges that are connected to its predecessors.
	Node $v$ updates its output buffers in two ways; (1) it updates them with respect to the received flows and (2) it updates them to avoid forwarding the whole unit flow of a specific edge in a single round. 
	Regarding the former case, fix an arbitrary round $t$ in which $v$ receives flows. 
	If the edges over which $v$ receives flows and the edges in $E_v$ are all $0$-edges (\Cref{fig:forwarding-3}), let $t' := t + r^{(1)}_v - r^{(0)}_v +1$.
	Otherwise (see \Cref{fig:forwarding-1} and~\ref{fig:forwarding-2}), let $t':= t+1$. 
	Then, for every $(e,f)\in I_v(t)$, $v$ updates its output buffers as follows:
	$$ \forall e'\in E_v: \ \ O_{v}(t', e') \leftarrow O_{v}(t', e') \cup (e, \frac{f}{|E_v|}) \text{ .}$$	
	Now regarding the latter way of output buffers update, fix an arbitrary round $t'$.
	At the beginning of round $t'$, let $O_v(t')$ be the set of all the flows in the output buffers of $v$ for round $t'$, i.e., $O_v(t') := \bigcup_{e'\in E_v} O_v(t', e')$.
	Then, let $S_v(e,t')$ be the sum of flow values of a specific edge $e$ that are sent by $v$ in round $t'$, i.e., $S_v(e,t'):= \sum_{(e,f)\in O_v(t')} f$.
	For every edge $e$, if $S_v(e,t') = 1$, node $v$ removes all flows of edge $e$ from all its output buffers of round $t'$.
	That is, $v$ removes all flows of $e$ and we say that $v$ \textit{discards} the flow of $e$. 
	After discarding all such flows, for every $e'\in E_v$, $v$ forwards $O_v(t', e')$ over edge $e'$ in round $t'$ if $O_v(t', e') \neq \emptyset$. 
	\\[.2cm]
	\noindent
	\textit{Setting Variables $r^{(1)}_v$ and $r^{(0)}_v$ (Reachability Detection):}
		We say that round $t$ is an \textit{incomplete round} for $v$ if node $v$ sends flow in round $t+1$. 
		Let $t$ be the first round in which $v$ receives tokens or be the first incomplete round for $v$.
		If $t$ is an even integer, $v$ assigns $t$ to $r^{(1)}_v$, otherwise, $v$ assigns $t$ to $r^{(0)}_v$.
		Note that the first round that a node receives a token (if any) is before its first incomplete round (if any) since it receives flows from the nodes it has already sent tokens to.

%%% Local Variables:
%%% mode: latex
%%% TeX-master: "main"
%%% End:

%% file: analysis.tex
% Analysis
\section{The Analysis of the DFNC Algorithm}\label{sec:analysis} 
	In this section, we show that an $r$-round execution of the DFNC algorithm provides the well-formed $r$-radius free node clustering as stated in the following lemma.
	
	\begin{lemma}\label{lem:DFNC-proof}
		For any integer $r$, an $r$-round execution of the DFNC algorithm on a graph $G$ and a matching $M$ of $G$ provides the well-formed $r$-radius free node clustering of $G$ with respect to $M$. 
	\end{lemma}
	
	\noindent
	We show the correctness of this lemma by induction on $r$. 
	For the entire \Cref{sec:analysis}, we fix an arbitrary graph $\calG=(V_{\calG},E_{\calG})$ and an arbitrary matching $\calM$ of $\calG$, where $\set{f_1,\dots,f_{\rho}}\subseteq V_{\calG}$ are the free nodes.
	Thus, throughout this section, when we refer to unmatched (free) nodes, matched nodes, augmenting paths, etc., they are always considered in $\calG$ with respect to $\calM$. 
	Let $\calE$ also be an arbitrary execution of the DFNC algorithm on $\calG$ and $\calM$.  
	For all $i$ and $t$, let $C_i(t)$ denote the set of the nodes in the cluster centered at free node $f_i$ in the $t$-radius free node clustering of $\calG$ with respect to $\calM$. 
	For all $i$ and $t$, let $D_i(t)$ also denote the set of the nodes that join the cluster centered at free node $f_i$ in the first $t$ rounds of $\calE$.
	Then, for all $t$, $\calC(t):= \set{C_1(t), \dots, C_\rho(t)}$ and $\calD(t):= \set{D_1(t), \dots, D_\rho(t)}$ are respectively the $t$-radius free node clustering and the clustering provided by the first $t$ rounds of $\calE$.
	
	We provide the necessary arguments for the induction proof in three sections. 
	In Section I, we start with proving the induction base, i.e., $\calD(0) = \calC(0)$. 
	Thereafter, we consider the following assumption as the induction hypothesis of the proof:
	
\vspace{.2cm}
\begin{addmargin}[2em]{2em}
	\textbf{\textit{I.H.}} \ For every $t<r$, the first $t$ rounds of $\calE$ provides the well-formed $t$-radius free node clustering of $\calG$ with respect to $\calM$.
\end{addmargin}
\vspace{.2cm}

\noindent
	We will then show that the clustering provided by the first $r$ rounds of $\calE$ is the $r$-radius free node clustering, i.e., $\calD(r) = \calC(r)$. 
	Moreover, it will be shown that the variables holding the set of predecessors of the nodes that join clusters in round $r$ of $\calE$ are correctly set. 
	We continue to prove that the provided clustering is also well-formed in the next two sections.
	In Section II, we show that for every node $v$, $r^{(\vartheta)}_v = r$ if $v$ is $r$-$\vartheta$-reachable.
	In Section III, we show that for every node $v$, $v$ is $r$-$\vartheta$-reachable if $r^{(\vartheta)}_v = r$. 
	Putting all these pieces together completes the induction argument and proves \Cref{lem:DFNC-proof}.
	%
%	\\[.4cm]
%	\noindent
%	\textbf{\large{Section I.}}
%	\\[.1cm]
%	\noindent
\subsection{Section I}
	Let us start with the proof of the induction base in following lemma. 
	\begin{lemma}\label{lem:base}
		$\calD (0)$ is the well-formed $0$-radius free node clustering of $\calG$, i.e., $\calC(0) = \calD(0)$. 
	\end{lemma}
	\begin{proof}
		Consider the state of the nodes right at the beginning of $\calE$, i.e., at time $0$. 
		Then, the provided clustering is the set of singleton sets containing the free nodes, which is the same as the $0$-radius free node clustering of $\calG$, i.e., $\calC(0) = \calD(0)$. 
		Every free node $v$ is $0$-$0$-reachable and $0$-$1$-reachable while $r^{(0)}_v = 0$ and $r^{(1)}_v = 0$. 
		Moreover, we have no other $0$-$0$-reachable or $0$-$1$-reachable nodes except the free nodes in $\calG$. 
		Furthermore, every free node $v$ has no predecessors while $pred_v = \emptyset$.
		Therefore, the clustering provided by the DFNC algorithm at the beginning of the execution, i.e., at time $0$, is the well-formed $0$-radius free node clustering of $\calG$.
		
              \end{proof}

	Now let us show that the provided clustering after $r$ rounds of $\calE$ is actually the $r$-radius free node clustering. 
	
	\begin{lemma}\label{lem:proper-clustering}
		Assuming that \textit{I.H.} holds, the clustering provided by the first $r$ rounds of $\calE$ is the $r$-radius free node clustering, i.e., $\calD (r) = \calC(r)$. 
	\end{lemma}
	\begin{proof}
		In addition to having for all $t<r$, $\calD(t) = \calC(t)$, the provided clustering by the first $r-1$ rounds of $\calE$ (i.e. $\calD(r-1)$) is assumed to be well-formed. 
		We only need to show that for all $i$, $C_i(r)\setminus C_i(r-1) = D_i(r)\setminus D_i(r-1)$. 
		To that end, for any $i\in [1,\rho]$, fix an arbitrary node $c\in C_i(r)\setminus C_i(r-1)$ and an arbitrary node $d\in D_i(r)\setminus D_i(r-1)$. 
		It is then enough to show that $c\in D_i(r)\setminus D_i(r-1)$ and $d\in C_i(r)\setminus C_i(r-1)$. 
		
		Let us first show that $c\in D_i(r)\setminus D_i(r-1)$. 
		To do so, we show that $c$ joins the cluster centered at $f_i$ in the $r^{th}$ round of $\calE$. 
		Considering the DFNC algorithm, it is then enough to show the following three points; (1) $c$ does not receive any token before round $r$, (2) $c$ receives token $f_i$ in round $r$, and (3) $f_i$ is the minimum one among all the tokens that $c$ receives in round $r$.
		Regarding the first point, since $c\in C_i(r)\setminus C_i(r-1)$, node $c$ joins cluster $C_i$ in the $r^{th}$ step of the FNC algorithm. 
		Therefore, $c$ is not in any cluster of $\calC(r-1)$ and hence not in any cluster of $\calD(r-1)$.
		That is, $c$ does not join any cluster in the first $r-1$ rounds of $\calE$.
		Hence, node $c$ does not receive any token before round $r$ as every node joins a cluster in the first round of receiving a token. 
		Regarding the second point, note that $c$ joins cluster $C_i$ in the $r^{th}$ step of the FNC algorithm.
		Therefore, $c$ has an almost uniform alternating path $P_c$ of length $r$ from $f_i$ in clustering $\calC(r-1)$, and consequently in $\calD(r-1)$ of $\calG$.
		Let $w$ be $c$'s neighbor on $P_c$ and connected to $c$ by a $\vartheta$-edge for some $\vartheta\in \set{0,1}$.
		Let us show that $P_c[f_i, w]$ is a shortest uniform $(1-\vartheta)$-path of $w$ in $\calD(r-1)$. 
		For the sake of contradiction, let us assume that $w$ has a uniform $(1-\vartheta)$-path $P_c'$ of length $t<r-1$. 
		Then, based on \Cref{lem:uniform}, $P_c'$ is uniform at time $t$ and hence $\mathbb{V}(P_c') \subseteq C_i(t)$. 
		Node $c$ is not in any cluster of $\calC(t)$ and thus not in $P_c'$. 
		Thus, the concatenation of $P_c'$ and $\langle w, c\rangle$ is an almost uniform alternating path of length less than $r$ of $c$ at time $t$. 
		Therefore, $c$ would join a cluster at latest in round $t$ of the FNC algorithm, which contradicts $c$ not being in any cluster of $\calC(r-1)$. 
		Hence, $P_c[f_i, w]$ is a shortest uniform $(1-\vartheta)$-path of $w$ in $\calD(r-1)$.
		Thus, $w$ is $(r-1)$-$(1-\vartheta)$-reachable and hence $r_{w}^{(1-\vartheta)} = r-1$. 
		Therefore, $w$ sends token $cid_w=f_i$ to $c$ in round $r$. 
		Now regarding the last point, for the sake of contradiction, let us assume that $c$ receives a token $f_j<f_i$ over a $\vartheta$-edge, say $\set{c, w'}$, in round $r$. 
		Then, $r^{(1-\vartheta)}_{w'} = r-1$ in $\calE$. 
		Then, since $\calD(r-1)$ is well-formed, $w'$ is $(r-1)$-$(1-\vartheta)$-reachable in clustering $\calD(r-1)$ and consequently in $\calC(r-1)$ of $\calG$. 
		Let $P_c''$ be a shortest uniform $(1-\vartheta)$-path of $w'$ in clustering $\calC(r-1)$ of $\calG$. 
		Based on \Cref{lem:uniform}, $P_c''$ is uniform at time $r-1$, i.e., $\mathbb{V}(P_c'')\subseteq C_j(r-1)$.
		However, $c\not \in P_c''$ because $c\in C_i(r) \setminus C_i(r-1)$ and hence $c\not\in C_j(r-1)$. 
		Thus, the concatenation of $P_c''$ and $\langle w',c\rangle$ is an almost uniform alternating path of length $r$ of $c$ from $f_j$ in $\calC(r-1)$, where $f_j<f_i$. 
		Therefore, $c$ must have not joined $C_i$ in the $r^{th}$ step of the FNC algorithm, which contradicts the choice of $c$. 
		
		Now let us show that $d\in C_i(r)\setminus C_i(r-1)$.
		To do so, we show that $d$ joins the cluster centered at $f_i$ in the $r^{th}$ step of the FNC algorithm execution on $\calG$. 
		Therefore, it is enough to show the following three points; (1) $d$ does not join any cluster in the first $r-1$ steps of the FNC algorithm, (2) $d$ has an almost uniform alternating path of length $r$ in clustering $\calC(r-1)$ and (3) $f_i$ is the minimum-ID free node from which $d$ has an almost uniform alternating path of length $r$ in clustering $\calC(r-1)$. 
		Regarding the first point, note that $d$ is not in $\calD(r-1)$ and consequently not in $\calC(r-1)$. 
		Therefore, $d$ does not join any cluster in the first $r-1$ steps of the FNC algorithm. 
		Regarding the second point, since $d\in D_i(r)\setminus D_i(r-1)$, $d$ joins the cluster centered at $f_i$ in the $r^{th}$ round of $\calE$. 
		Therefore, for some node $z$ and integer $\varphi \in \set{0,1}$, $d$ receives token $f_i$ from $z$ over its $\varphi$-edge in round $r$, and hence $r^{(1-\varphi)}_z = r-1$.
		As a result, since the first $r-1$ rounds of $\calE$ provides the well-formed $(r-1)$-radius free node clustering, $z$ must have a shortest uniform $(1-\varphi)$-path $P_d$ of length $r-1$ in clustering $\calD(r-1)$ and consequently in $\calC(r-1)$.
		Based on \Cref{lem:uniform}, $P_d$ must be uniform at time $r-1$, whereas $d$ is not in any cluster of $\calD(r-1)$. 
		Therefore, $d\not \in P_d$ and hence the concatenation of $P_d$ and $\langle z, d \rangle$ is an almost uniform $\varphi$-path of length $r$ in $\calC(r-1)$. 
		Now regarding the last point, for the sake of contradiction, let us assume that $d$ has an almost uniform alternating path $P_d'$ of length $r$ from a free node $f_{j'}< f_i$ in $\calC(r-1)$ and consequently in $\calD(r-1)$.
		Let $z'$ be $d$'s neighbor on $P_d'$. 
		Let us show that $P_d'[f_{j'}, z']$ is a shortest uniform $(1-\varphi)$-path of $z'$. 
		Let $z'$ have a shorter uniform $(1-\varphi)$-path $P_d''$ of length $t'<r-1$. 
		Then, based on \Cref{lem:uniform}, $P_d''$ is uniform at time $t'$ in clustering $\calD(t')$. 
		Then, since $d$ is not in any cluster of $\calD(t')$, the concatenation of $P_d''$ and $\langle z', d\rangle$ would be an almost uniform alternating path of length less than $r$, and hence $d$ would join some cluster before round $r$ of $\calE$, which is contradictory. 
		Therefore, $P_d'[f_{j'}, z']$ is a shortest uniform $(1-\varphi)$-path of length $r-1$ of $z'$.
		Therefore, $r_{z'}^{(1-\varphi)} = r-1$, and hence $z'$ sends token $f_{j'}$ to $d$ in round $r$. 
		This contradicts node $d$ joining $C_i$ in round $r$ since $f_{j'}<f_j$. 
		
	\end{proof}	
				
	As a final step of this section, we state in the following lemma that predecessors are correctly set in the execution of the DFNC algorithm. 		
	\begin{lemma}\label{lem:step2}
		Assuming that \textit{I.H.} holds, for all $i$ and every node $v\in D_i(r)$, $pred_v$ is properly set to $v$'s predecessors in round $r$ of $\calE$. 
	\end{lemma}
	\begin{proof}
		In addition to having for all $t<r$, $\calD(t) = \calC(t)$, the provided clustering by the first $r-1$ rounds of $\calE$ (i.e. $\calD(r-1)$) is assumed to be well-formed.
		Here we show that for all $i$ and every node $v\in D_i(r)\setminus D_i(r-1)$, variable $pred_v$ is properly set to the set of $v$'s predecessors in round $r$ of $\calE$. 
		Fix an arbitrary node $v\in D_i(r)\setminus D_i(r-1)$ for any $i$. 
		We only need to show that $v$ receives token $f_i$ from node $w$ in round $r$ if and only if $w$ is a predecessor of $v$. 
	
		First, let us fix an arbitrary predecessor $w_1$ of $v$, and show that $w_1$ sends token $f_i$ to $v$ in round $r$. 
		Let $P$ be a shortest uniform alternating path of $v$ on which $w_1$ is $v$'s neighbor. 
		Note that $|P| = r$. 
		Let $P$ be a $\vartheta$-path for some $\vartheta \in \set{0,1}$. 
		Now let us show that $P[f_i, w_1]$ is a shortest uniform $(1-\vartheta)$-path of $w_1$. 
		For the sake of contradiction, let us assume that $w_1$ has uniform $(1-\vartheta)$-path $P'$ of length $t < r-1$.  
		Based on \Cref{lem:uniform}, path $P'$ is uniform at time $t$, i.e., $\mathbb{V}(P')\subseteq C_i(t)$.
		However, $v\notin C_i(t)$ since $v\in D_i(r)\setminus D_i(r-1)$, where $t<r-1$. 
		Therefore, $v\not\in P'$, and hence the concatenation of $P'$ and $\langle w_1, v \rangle$ would be an almost uniform $\vartheta$-path of length less than $r$ of $v$. 
		Hence, $v$ must have joined some cluster before round $r$, which is contradictory. 
		Therefore, $P[f_i, w_1]$ is a shortest uniform $(1-\vartheta)$-path of $w_1$.
		Hence, $w_1$ is $(r-1)$-$(1-\varphi)$-reachable. 
		Since $\calD(r-1)$ is well-formed, it thus holds that $r^{(1-\varphi)}_{w_1} = r-1$. 
		As a result, $w_1$ sends $f_i$ to $v$ in round $r$. 
	
		Second, let us fix an arbitrary node $w_2$ that sends token $f_i$ to $v$ in round $r$ over a $\vartheta$-edge for some $\vartheta \in \set{0,1}$. 
		%We would like to show that $w_2$ is a predecessor of $v$.
		Thus, $cid_{w_2} = f_i$ and $r^{(1-\vartheta)}_{w_2} = r-1$. 
		Therefore, since $\calD(r-1)$ is the well-formed $(r-1)$-radius free node clustering of $\calG$, $w_2$ has a shortest uniform alternating path $P''$ of length $r-1$, which is a $(1-\vartheta)$-path. 
		Based on \Cref{lem:uniform}, $P''$ is uniform at time $r-1$ whereas node $v$ has not yet joined any cluster. 
		Hence, $v\not \in P''$ and therefore the concatenation of $P''$ and $\langle w_2, v \rangle$ is analternating path. 
		It will become a uniform $\vartheta$-path of length $r$ of $v$ after node $v$ joins $D_i$ in round $r$ of $\calE$. 
		Therefore, $w_2$ is the neighbor of $v$ on a uniform $\vartheta$-path of $v$, and hence a predecessor of $v$. 
		
		\end{proof} 	
%	
%	\noindent
%	\textbf{\large{Section II.}}
%	\\[.1cm]
%	\noindent
\subsection{Section II}
	In this section, we show that after $r$ rounds of the DFNC algorithm execution, for all $\vartheta\in \set{0,1}$, every $r$-$\vartheta$-reachable node properly detects its $r$-$\vartheta$-reachability as stated in the following Lemma:
	\begin{lemma}\label{lem:step3-odd}
		Assuming that \textit{I.H.} holds, for every node $v$ and integer $\vartheta \in \set{0,1}$, $r^{(\vartheta)}_v = r$ in execution $\calE$ if $v$ is $r$-$\vartheta$-reachable. 
	\end{lemma}
	To prove this Lemma, we need to first point out a few observations about flow circulation throughout the network while running the DFNC algorithm. 
	To that end, we provide a series of necessary definitions and helper lemmas in the sequel.  
	Since every node forwards the received flows to their predecessors, the flows do not necessarily traverse alternating paths.
	We call the paths along which a node sends flows towards the cluster center the node's \textit{shortcuts}.
		Consider the following definition for a more precise description of a shortcut. 
		\begin{definition}[Shortcuts]
			For any matched node $v$ and free node $f_i$, a uniform path $P:=\langle v = v_0, v_1, \dots , v_\ell = f_i\rangle$ is called a shortcut of $v$ if for all $j<\ell$, $v_{j+1}$ is the predecessor of $v_j$. 
		\end{definition}	
	\begin{lemma}\label{lem:reach}
		Assuming that \textit{I.H.} holds, let distinct nodes $u$ and $v$ respectively be $r_u$-reachable and $r_v$-reachable for $r_v\leq r_u \leq r$. 
		Then, $v$ has no shortcut containing $u$. 
	\end{lemma}
	\begin{proof}
		For the sake of contradiction, let us assume that $v$ has a shortcut $S$ that contains $u$. 
		Let us name the nodes in $S[v, u]$ as $\langle v = w_0, \dots , w_m = u \rangle$ for some integer $m$. 
		To show a contradiction, we prove by induction that $r_v > r_u$. 
		To that end, we prove that for all $j<m$, the reachability of $w_j$ is strictly greater than that of $w_{j+1}$. 
		For the induction base, we prove that the reachability of $v=w_0$ is strictly greater than that of $w_1$.
		Considering \textit{I.H.} and \Cref{lem:proper-clustering}, $\calD(r) = \calC(r)$. 
		Therefore, since $v$ is $r_v$-reachable, $v$ must receive a token for the first time in round $r_v$. 
		Hence, since $w_1$ is the predecessor of $v$, $w_1$ must send a token to $v$. 
		Therefore, $w_1$ must have a received a token before round $r_v$. 
		Let $t<r_v\leq r$ be the first round that $w_1$ receives a token.
		Due to \textit{I.H.}, $w_1$ is $t$-reachable.
		Therefore, the reachability of $v=w_0$ is strictly greater than that of $w_1$.
		Now since the reachability of $w_1$ is less than that of $v$ and hence less than $r$, we can inductively employ \textit{I.H.} and prove that the reachability of $w_m = u$ is less than that of $w_0 = v$. 
		That is, $r_v > r_u$, which is contradictory. 
		
              \end{proof}
	
	\begin{lemma}\label{lem:reach2}
		Assuming that \textit{I.H.} holds, let a $t$-reachable node $v$ have a shortcut $S$ that contains a $t'$-$\vartheta$-reachable node $u$ for any integers $t', t\leq r$ and $\vartheta\in \set{0,1}$.
		If $u$'s adjacent edge on $S[v, u]$ is a $(1-\vartheta)$-edge, then $t' <t$.  
	\end{lemma}
	\begin{proof}
		Let $v$ be in cluster $C_i$, and hence $S$ is a path between $v$ and $f_i$. 
		Let $w$ be $u$'s neighbor in $S[v, u]$, where consequently $\set{w, u}$ is a $(1-\vartheta)$-edge. 
		Node $u$ is the predecessor of $w$.
		Therefore, $u$ is $w$'s neighbor in a shortest uniform alternating path $P$ of $w$. 
		Since $\set{w, u}$ is a $(1-\vartheta)$-edge and $w$ is connected to its predecessors, namely $u$, with $(1-\vartheta)$-edges, every shortest uniform alternating path of $w$ must be a $(1-\vartheta)$-path. 
		Thus, $P$ is a $(1-\vartheta)$-path.  
		Hence $P[u, f_i]$ is a uniform $\vartheta$-path of $u$. 
		Since $u$ is $t'$-$\vartheta$-reachable, it holds that $|P[u, f_i]| \geq t'$. 
		Therefore, $|P|>t'$, that is $w$ is $t''$-reachable for some $t'<t''$. 
		If $w = v$, the proof is concluded.
		Otherwise, based on \Cref{lem:reach}, the reachability of $w$ is strictly less than that of $v$, which concludes $t'<t$. 
		
		\end{proof}

	To prove \Cref{lem:step3-odd}, we benefit from some specific way of marking some of the shortcuts and accordingly labeling some of the nodes.
	\\[.3cm]
	\noindent
	\textbf{Shortcut marking. } Let $P$ be an arbitrary uniform $\vartheta$-path of any node $v$ in the cluster centered at some free node $f_i$. 
	Then, we \textit{mark} a shortcut $S$ of a node $z$ with respect to $P$ when for every edge $\set{u,w}\in P$ that $u$ is the predecessor of $w$, if $S$ contains $w$ then $w$'s neighbor on $S[w, f_i]$ is $u$. 
	\\[.3cm]
	\noindent
	\textbf{Node labeling. } Considering path $P$, we label a node with $P^+_v$ when it has a marked shortcut with respect to $P$ that contains $v$. 
	Further, we label a node with $P^-_v$ when it has no marked shortcut with respect to $P$ that contains $v$.
	Note that each node can have either label $P^+_v$ or $P^-_v$, but not both. 

	\begin{lemma}\label{lem:good-edge}
		Assuming that \textit{I.H.} holds, let an arbitrary node $v$ in cluster $C_i$ be $r$-$\vartheta$-reachable and $r'$-$(1-\vartheta)$-reachable, where $r'<r$.
		There is an edge $\set{w, u}$ on every shortest uniform $\vartheta$-path $P$ of $v$ such that all the nodes in $P[u , v]$ are labeled $P^+_v$, node $w$ is labeled $P^-_v$ and $|P[f_i , w]|\geq r'$.
	\end{lemma}
	\begin{proof}
		Throughout the proof consider all marking and labeling with respect to path $P$. 
		Let us name the nodes in $P$ as $\langle v=v_0, \dots , v_r=f_i \rangle$.	
		Then, let $t$ be the integer such that $|P[v_t, f_i]| = r'$. 
		Since $r'<r$, $P[v, v_t]$ is of length at least $1$.  
		Since $P[v_t, f_i]$ is of length $r'$, node $v_t$ is $r''$-reachable for some $r''\leq r'$ whereas $v$ is $r'$-reachable. 
		Then, based on \Cref{lem:reach}, $v_t$ has no shortcut containing $v$ and is consequently labeled $P^-_v$. 
		Note that every node in $P$ has a marked shortcut.
		Therefore, since every shortcut of $v$ obviously contains $v$, node $v$ has a marked shortcut containing itself and thus labeled $P^+_v$. 
		Path $P[v, v_t]$ has one endpoint labeled $P^+_v$ and one endpoint labeled $P^-_v$.
		Therefore, since all the nodes in $P[v, v_t]$ have either label $P^+_v$ or $P^-_v$, there is an edge in $P[v, v_t]$ whose endpoints have different labels. 
		Therefore, edge $\set{u, w}$ is the closest edge to $v$ on $P$ whose endpoints have different labels.
		
              \end{proof}

	Let us next study the time it takes for a flow to traverse shortcuts. 
	To do so, we define the \textit{promoted length} of a shortcut (or a consecutive subpath of a shortcut) as the time it takes for a flow to traverse the path. 
	The traversal time of a path by a flow is actually the sum of the length of the path and all the delays caused by the inner nodes in flow forwarding along the path. 
	However, we would like to have a generalized definition for any walks rather than having the definition only for shortcuts. 
	Let us first formally define the delay by a node along a walk as follows. 
	For ease of discussion, for any $\varphi \in \set{0, 1}$, we say that a node is $\infty$-$\varphi$-reachable when it has no uniform $\varphi$-path in $\calG$. 
	\begin{definition}[Delay]
		Consider any walk $P$ and node $v\in P$, where $v$ is $r_0$-$0$-reachable and $r_1$-$1$-reachable. 
		The delay by node $v$ along walk $P$, denoted by $d(P, v)$, is defined to be $r_1 - r_0$ if $v$ has two adjacent $0$-edges on $P$ and $0$ otherwise.		 
	\end{definition}

	\begin{definition}[Promoted Length]
		The promoted length of walk $P$ is denoted and defined by $\|P\| := |P| + \sum_{v\in P}d(P, v)$. 
	\end{definition}

	\begin{lemma}\label{lem:promoted}
		Assuming that \textit{I.H.} holds, for any $\ell\leq r$, the promoted length of a shortcut of an $\ell$-reachable node equals $\ell$. 
	\end{lemma}
	\begin{proof}
		Fix an arbitrary $\ell$-reachable node $u$. 
		Without loss of generality, let $u$ be $\ell$-$0$-reachable. 
		Let $S$ be an arbitrary shortcut of $u$.
		Further, let $\ell'$ be the promoted length of $S$, i.e., $\|S\| = \ell'$. 
		Then, the goal is to show that $\ell=\ell'$. 
		Let us first assume that $S$ is an alternating path, and hence, $S$ is a shortest uniform $0$-path of $u$. 
		Therefore, since the promoted length of an alternating path is equal to the length of the path, it holds that $\ell' = \ell$, concluding the proof. 
		Now let us assume that $S$ is not an alternating path. 
		Let $u_1, \dots , u_t$ be the nodes in $S$ that have two adjacent $0$-edges on $S$ such that for all $j$, $u_j$ is closer to $u$ on $S$ than $u_{j-1}$. 
		For every node $u_j$ that has two adjacent $0$-edges in $S$, there are integers $o_j$ and $z_j<o_j$ such that $u_j$ is $z_j$-$0$-reachable and $o_j$-$1$-reachable. 
		For every node $u_j$, if it has two adjacent $0$-edges in $S$, let $d_j := o_j -z_j$, and otherwise let $d_j := 0$.
		To prove that $\ell=\ell'$, it is enough to show that there is a shortest uniform $0$-path of length $\ell'$ of $u$. 
		To do so, we construct one in $t$ phases.
		
		Let $P$ initially be $S$, which is of promoted length $\ell'$.
		Then, in $t$ phases we gradually transform $P$ to a shortest uniform $0$-path of  $u$ while the promoted length of $P$ remains the same.
		To do so, in each phase, we update $P$ in a way that the number of nodes with two adjacent $0$-edges on $P$ reduces by $1$. 
		In every phase $j$, $1\leq j \leq t$, we change $P$ by replacing $P[u_j, f_i]$ with a shortest uniform $1$-path of length $|P[u_j, f_i]|+d_j$ of $u_j$. 
		Note that every node in $P[u, u_j)$ has a shortcut containing $u_j$ such that $u_j$'s adjacent edge in $P[u, u_j]$ is a $0$-edge. 
		Therefore, based on \Cref{lem:reach2}, for every node in $P[u, u_j)$, there is an integer $\ell''>o_j$ such that the node is $\ell''$-reachable. 
		Moreover, it is easy to see that for every node in $P[u_j, f_i]$, there is an integer $\ell''\leq o_j$ such that the node is $\ell''$-reachable.
		Therefore, $P[u, u_j)$ and $P[u_j, f_i]$ have no common node, and consequently, $P$ is still a path after its change in phase $j$. 
		Furthermore, since $u_j$ does not have two adjacent $0$-edges in $P$ anymore, the promoted length of $P$ remains the same.  
		
		In the last phase, we replace $P[u_t, f_i]$ with a shortest uniform $1$-path of $u_t$.
		Moreover, every node is the next node's successor in path $P[u, u_t]$.
		Therefore, $P$ is a shortest uniform alternating path of $u$ whose promoted length remained $\ell'$ throughout the $t$ phases as we argued. 
		Since $P$ is now an alternating path, its length is the same as its promoted length, i.e., $\ell'$. 
		Therefore, since a shortest uniform alternating path of $v$ must be of length $\ell$, we can conclude that $\ell = \ell'$. 
		
		\end{proof}
	
	\begin{lemma}\label{lem:discards}
		Assuming that \textit{I.H.} holds, let any node $v$ assigns a flow of an arbitrary edge $e$ to be sent in round $t\leq r+1$. 
		Then, $v$ does not assign any flow of $e$ to be sent in any round except round $t$. 
	\end{lemma}
	\begin{proof}
%	{\color{red}\large Updating ...}\\
		Let $u$ and $w$ be the two endpoints of $e$. 
		Let $r_u$ be the round in which $u$ sends token to $w$, and $r_w$ be the round in which $w$ sends token to $u$. 
		We first show that both $r_u$ and $r_w$ are at most $r$. 
		Since $v$ receives a flow of $e$, the flow must have traversed a shortcut of $u$ or $w$ and reached $v$ (note that as a special case when $v$ is one of $e$'s endpoints, the shortcut has length $0$). 
		Without loss of generality, let us assume that a flow of $e$ has traversed a shortcut $S$ of $u$ and reached $v$. 
		Hence, the flow is assigned to be sent by $w$ after round $r_u$.
		Therefore, node $v$ also assigns the flow to be sent after round $r_u$.
		That is $r_u < t$ and hence $r_u\leq r$.
		Let us show that it also holds that $r_w\leq r$. 
		To study the non-trivial case, let us assume that 	
		\begin{equation}\label{eq:s1}
			r_u<r_w\text{ .}
		\end{equation}	
		Moreover, let us assume that $w$ is $z'$-$\vartheta$-reachable and $z$-$(1-\vartheta)$-reachable for some integers $\vartheta\in \set{0, 1}, z$ and $z' < z$.
		A flow is generated over edge $e$.
		Hence, node $u$ is not a predecessor of $w$, and therefore $w$ can only receive a token from $u$ after round $z'$, i.e., 
		\begin{equation}\label{eq:s2}
			z' < r_u\text{ .}
		\end{equation}
		Therefore, since $r_u\leq r$, it holds that $z' < r$. 
		Hence, considering \textit{I.H.}, node $w$ has set its variable $r_w^{(\vartheta)}$ to $z'$ and must send token over its $(1-\vartheta)$-edges in round $z'+1$. 
		However, based on (\ref{eq:s1}) and (\ref{eq:s2}), $z' + 1\neq r_w$. 
		Node $w$ does not therefore send token to $u$ over $e$ in round $z' + 1$. 
		This concludes that $e$ is a $\vartheta$-edge (and hence $\vartheta = 0$).
		Therefore, node $w$ must send a token to $u$ in round $z+1$, i.e., 
		\begin{equation}\label{eq:s3}
			r_w = z + 1\text{ .}
		\end{equation}
		Moreover, node $w$ must delay flow of $e$ by $z - z'$ rounds. 
		This means that $t > r_u + z - z'$. 
		Considering (\ref{eq:s2}), it results in having $t> z + 1$. 
		Together with (\ref{eq:s3}), it concludes that $r_w < t$ and hence $r_w \leq r$. 
		We just proved that both $r_u$ and $r_w$ are at most $r$, i.e., 
		\begin{equation}\label{eq:s4}
			r_u \leq r \text{ \ \ and \ \ } r_w \leq r\text{ .}
		\end{equation}
		
		Let $\calS_u$ be the set of all the shortcuts of $u$ containing $v$, and $\calS_w$ be the set of all the shortcuts of $w$ containing $v$.
		A flow of $e$ that reaches $v$ should either traverse a shortcut in $\calS_w$ or $\calS_u$.
		Hence, in both cases we show that the flow is assigned by $v$ to be sent in the same round, i.e., $t$. 
		Without loss of generality, let the flow that is
                assigned to be sent by $v$ in round $t$ has traversed a shortcut $S\in \calS_w$ to reach $v$.
		Let $P$ be the concatenation of $\langle w, u \rangle$ and $S$. 
		Then, considering the time that it takes for the flow to traverse $S$ and be sent by $v$ in round $t$, it holds that
		\begin{equation}\label{eq:s5}
			t = r_u + d(P, w) + \|S[w,v]\| + d(S, v) + 1\text{ .}
		\end{equation}
	Note that as a special case when $v = w$, it holds that $\|S[w,v]\| = 0$ and $d(S, v) = 0$, and the equation (\ref{eq:s5}) properly shows the correct calculation of the round in which the flow is assigned to be sent by $v$.

		First let us show that every flow that traverses a shortcut in $\calS_w$ and reaches $v$ is assigned to be sent by $v$ in round $t$. 
		To that end, let us fix an arbitrary shortcut $S_w \in \calS_w$ and an arbitrary flow that traverses $S_w$.
		Let $P_w$ be the concatenation of $\langle u, w \rangle$ and $S_w$.
		Let us assume that this flow is assigned to be sent by $v$ in round $t_w:= r_u + d(P_w, w) + \|S_w\| + d(S_w, v) + 1$ and show that $t = t_w$:
		\begin{align}
			t &\overset{\mathrm{(\ref{eq:s5})}}{=} r_u + d(P, w) + \|S[w,v]\| + d(S, v) + 1  \nonumber \\
			  &= r_u + d(P_w, w) + \|S[w,v]\| + d(S, v) + 1  \nonumber \\
			  &= r_u + d(P_w, w) + \|S\| - \|S[v, f_i]\| + 1  \nonumber \\
			  &= r_u + d(P_w, w) + \|S_w\| - \|S_w[v, f_i]\| + 1  \nonumber \\
			  &= r_u + d(P_w, w) + \|S_w\| + d(S_w, v) + 1  \nonumber \\
			  &= t_w \nonumber
		\end{align} 
		The second equality above comes from the fact that a node is either connected to all its predecessors by $0$-edges or by a $1$-edge. 
		The forth equality above is because the promoted length of all shortcuts of a node have the same length, due to \Cref{lem:promoted}.
		
		Second let us show that every flow that traverses a shortcut in $\calS_u$ and reaches $v$ is assigned to be sent by $v$ in round $t$. 
		To that end, let us fix an arbitrary shortcut $S_u \in \calS_u$ and an arbitrary flow that traverses $S_u$.
		Let $P_u$ be the concatenation of $\langle w, u \rangle$ and $S_u$.
		Let us assume that this flow is assigned to be sent by $v$ in round $t_u:= r_w + d(P_u, u) + \|S_u[u,v]\| + d(S_u, v) + 1$ and show that $t = t_u$.
		Let $e$ be a $\varphi$-edge for some $\varphi\in \set{0,1}$.
		Based on equation (\ref{eq:s5}) and \textit{I.H.}, node $u$ must be $(r_u - 1)$-$(1-\varphi)$-reachable. 
		Moreover, the length of a shortest uniform $(1-\varphi)$-path of $u$ is of length $\|S_u\| + d(P_u, u)$. 
		Therefore, it holds that 
		\begin{equation}\label{eq:s6}
			r_u = \|S_u\| + d(P_u, u) + 1 \text{ .}
		\end{equation}
		Considering $w$'s shortcut $S$ and the path $P:= S\circ \langle u, w \rangle$, it similarly holds that 	
		\begin{equation}\label{eq:s7}
			r_w = \|S\| + d(P, w) + 1 \text{ .}
		\end{equation}
		Now let us show that $t = t_u$ as follows. 
		\begin{align}
		t &\overset{\mathrm{(\ref{eq:s5})}}{=} r_u + d(P, w) + \|S[w,v]\| + d(S, v) + 1  \nonumber \\
		  & = r_u + d(P, w) + \|S\| - \|S[v, f_i]\| + 1  \nonumber \\
		  & \overset{\mathrm{(\ref{eq:s7})}}{=} r_u + r_w - \|S[v, f_i]\| \nonumber \\
		  & \overset{\mathrm{(\ref{eq:s6})}}{=} \|S_u\| + d(P_u, u) + 1 + r_w - \|S[v, f_i]\| \nonumber \\
		  & = \|S_u\| + d(P_u, u) + 1 + r_w - \|S_u[v, f_i]\| \nonumber \\
		  & = r_w + d(P_u, u) + \|S_u[u,v]\| + d(S_u, v) + 1  \nonumber \\
		  &= t_u \nonumber
		\end{align}
	\end{proof}

	Now let us at the end of this section provide the proof of \Cref{lem:step3-odd} below:
	
	\begin{proof}[Proof of \Cref{lem:step3-odd}]
		Considering an $r$-$\vartheta$-reachable node $v$, we show that $r^{(\vartheta)}_v = r$ after $r$ rounds of $\calE$.
		Since $v$ is $r$-$\vartheta$-reachable, it has a shortest uniform alternating path of length at most $r$. 
		Let us first consider the case when $v$ is $r$-reachable, i.e., the shortest uniform alternating path of $v$ is of length $r$. 
		Since $\calD(r-1)$ is the well-formed $(r-1)$-radius free node clustering, $v$ does not join any cluster in the first $r-1$ rounds of $\calE$. 
		Moreover, based on \Cref{lem:proper-clustering}, $\calD(r)$ is the $r$-radius free node clustering. 
		Therefore, $v$ must join its cluster in round $r$ of $\calE$. 
		Hence, $r$ is the first round in which $v$ receives a token. 
		Note that the length of any $0$-path is odd, and the length of any $1$-path is even. 
		Therefore, $r$ is even if $\vartheta = 1$, and it is odd otherwise. 
		Thus, based on the DFNC algorithm, node $v$ sets $r_v^{(\vartheta)}$ to $r$ in round $r$ of $\calE$. 
		
		For the remainder of the proof, let us consider the case when $v$ is $r'$-reachable for $r'<r$, and hence, $v$ joins its cluster before round $r$. 
		That is, $v$ is $r$-$\vartheta$-reachable and $r'$-$(1-\vartheta)$-reachable. 
		Then, to show that $r^{(\vartheta)}_v=r$, we need to show that $r$ is the first incomplete round for $v$. 
		Since the first $r-1$ rounds of $\calE$ provides the well-formed $(r-1)$-radius free node clustering of $\calG$ and $v$ is $r$-$\vartheta$-reachable, $v$ cannot have an incomplete round before round $r$. 
		Therefore, it is enough to show that $v$ has an incomplete round in the first $r$ rounds of $\calE$. 
		
		Let $P$ be an arbitrary shortest uniform $\vartheta$-path of $v$. 
		Then, based on \Cref{lem:good-edge}, let $e := \set{u,w}$ be the edge on $P$ where all the nodes in $P[u, v]$ are labeled $P^+_v$, node $w$ is labeled $P^-_v$ and $|P[w, f_i]| \geq r'$. 
		Let us first show that a flow is generated over edge $e$.
		To this end, we need to show that $u$ and $w$ send tokens to each other and none of them is the other one's predecessor.  
		
		Let us first show that $u$ and $w$ are not each other's predecessor. 
		For the sake of contradiction, let us assume otherwise. 
		Let us first consider the case when $u$ is a predecessor of $w$. 
		Let $S_u$ be a marked shortcut of $u$ that contains $v$. 
		$S_u$ does not contain $w$, and hence the concatenation of $S_u$ and $\langle w, u \rangle$ is a marked shortcut of $w$ that contains $v$.
		Therefore, $w$ must have been labeled $P^+_v$, which is contradictory. 
		Now let us consider the case when $w$ is a predecessor of $u$.
		Then, $u$'s neighbor on $S_u$ must be $w$.
		Hence, $S_u[w, f_i]$ is a marked shortcut of $w$ that contains $v$.
		Therefore, $w$ must have been labeled $P^+_v$, which is again contradictory. 
		
		Next we show that $u$ and $w$ send tokens to each other. 
		Let $\set{u,w}$ be a $\varphi$-edge for an integer $\varphi\in \set{0,1}$. 
		First, let us show that $w$ sends token to $u$.
		Since $P[w, f_i]$ is a uniform $(1-\varphi)$-path of length less than $r$ of $w$, node $w$ is $\ell_w$-$(1-\varphi)$-reachable for some $\ell_w<r$. 
		Therefore, since $\calD(r-1)$ is the well-formed $r$-radius free node clustering, $r_w^{(1-\varphi)}<r$. 
		Hence, $w$ sends token to $u$ at latest in round $|P[w, f_i]| + 1\leq r$.
		Now, let us show that $u$ sends token to $w$. 
		Every node in $P[v, u]$ is labeled $P^+_v$ and hence has a shortcut containing $v$. 
		Therefore, based on \Cref{lem:reach}, for every node in $P(v, u]$, there is an integer $\ell>r'$ such that the node is $\ell$-reachable. 
		Let $P'$ be a shortest uniform $(1-\vartheta)$-path of $v$, which is of length $r'$.
		Then, every node in $P'(v, f_i]$ $\ell'$-reachable for some $\ell'<r'$. 
		Therefore, $P[v, u]$ and $P'$ have no common node except $v$.
		Then, the concatenation of $P[v, u]$ and $P'$ is a uniform $(1-\vartheta)$-path of $u$.
		Note that this concatenated path is of length less than $r$ since based on \Cref{lem:good-edge}, $|P[w, f_i]|\geq r'$. 
		Therefore, since $\calD(r-1)$ is the well-formed $r$-radius free node clustering, $r_u^{(1-\varphi)} <r$.
		Hence, $u$ sends token to $w$ in the first $r$ rounds of $\calE$.

		We show that $r$ is the first incomplete round for $v$ in which $v$ receives a flow of $e$ but not the whole flow. 
		Let $S^+$ be a marked shortcut of $u$ containing $v$. 
		Since every node has a marked shortcut and $w$ has no marked shortcut containing $v$, let $S^-$ be a marked shortcut of $w$ that does not contain $v$. 
		Since a flow of $e$ is sent along $S^-$ and this shortcut does not contain $v$, node $v$ does not receive the whole flow of $e$. 
		Hence, to show that $v$ receives a proper fraction of $e$ in round $r$, it is enough to show the following two facts; (1) None of the nodes in $S^+[u, v]$ discards the whole flow of $e$, and hence, $v$ receives a flow of $e$, (2) Node $v$ receives a flow of $e$ in the first $r$ rounds of $\calE$. 
		
		For the sake of contradiction, let $z\in S^+[u, v]$ be the node that discards the whole flow of $e$.
		Therefore, all the flows of $e$ are received by $z$, and hence, $S^-$ must contain $z$.
		Then, the concatenation of $S^-[w, z]$ and $S^+[z, f_i]$ is a marked shortcut of $w$ that contains $v$, which contradicts $w$ being labeled $P^-_v$. 
		As a result, there is no node in $S^+[u, v]$ that discards the whole flow of $e$. 
		Thus, since $v$ receives a flow of $e$ over $S^+[u, v]$. 
		
		Now let $u$ be $\ell_u$-$\varphi$-reachable and $\ell_u'$-$(1-\varphi)$-reachable (recall that $e$ was considered to be a $\varphi$-edge for a $\varphi \in \set{0,1}$). 
		Every node in $P[v, u]$ is labeled $P^+_v$ and thus has a shortcut containing $v$. 
		Therefore, based on \Cref{lem:reach}, for every node in $P(v, u]$, there is an integer $\ell_1>r'$ such that the node is $\ell_1$-reachable. 
		Moreover, it is easy to see that for every node in $P'$, there is an integer $\ell_2\leq r'$ such that the node is $\ell_2$-reachable.
		Therefore, the concatenation of $P'$ and $P[v, u]$ is a path, and in particular a uniform $(1-\varphi)$-path of length $r'+|P[v, u]|$ of $u$. 
		Hence, 
		\begin{equation}\label{eq:lu}
			\ell_u' \leq r'+|P[v, u]|\ \text{.}
		\end{equation}
		Let $P^+_v := \langle w, u \rangle \circ S^+$. 
		If $S^+$ is a $(1-\varphi)$-path of $u$, $u$ is $\ell_u'$-reachable, and therefore, $\|S^+\| = \ell_u'$ due to \Cref{lem:promoted}. 
		Otherwise, $u$ is $\ell_u$-reachable and $\|S^+\| = \ell_u = \ell_u' - d(P^+_v, u)$. 
		Since in case $S^+$ is a $(1-\varphi)$-path, it holds that $d(P^+_v, u) = 0$, we can overall conclude that
		\begin{align}\label{eq:s+}
			\|S^+\| &= \ell_u' - d(P^+_v, u) \nonumber \\
				   &\overset{(\ref{eq:lu})}{\leq} r' + |P[v, u]| - d(P^+_v, u) \ \text{.}
		\end{align}
		Moreover, letting $r_w$ be the round in which $w$ sends token to $u$, 
		\begin{equation}\label{eq:rw}
			r_w \leq |P[w, f_i]|+1\ \text{.}
		\end{equation}
		Note that the flow that traverse $S^+$ reaches $v$ in round $t:= r_w + d(P^+_v, u) + \|S^+[u,v]\|$. 
		Let us show that $t\leq r$ in the following: 
		\begin{align}\label{eq:sent}
			t &:= r_w + d(P^+_v, u) + \|S^+[u,v]\| \nonumber \\
			  &\overset{(\ref{eq:rw})}{\leq} |P[w, f_i]| + 1 + d(P^+_v, u) + \|S^+[u,v]\| \nonumber \\
			  &= |P[w, f_i]| + 1 + d(P^+_v, u) + \|S^+\| - \|S^+[v, f_i]\| \nonumber \\
			  &\overset{(\ref{eq:s+})}{\leq} |P[w, f_i]| + 1 + d(P^+_v, u) + r' + |P[v, u]| - d(P^+_v, u) - \|S^+[v, f_i]\| \nonumber \\
			  & = |P[w, f_i]| + 1 + r' + |P[v, u]| - \|S^+[v, f_i]\| \nonumber \\
			  &= |P[w, f_i]| + 1 + r' + |P[v, u]| - r' \nonumber \\
			  &= |P[w, f_i]| + 1 + |P[v, u]| \nonumber \\
			  &= |P| \nonumber \\
			  &= r \nonumber
		\end{align}
              \end{proof}

%	\noindent
%	\textbf{\large{Section III.}}
%	\\[.1cm]
%	\noindent
\subsection{Section III}
	In this section, we show that after $r$ rounds of the DFNC algorithm execution, every node that detects $r$-$\vartheta$-reachability is actually an $r$-$\vartheta$-reachable node as stated in the following Lemma:
	\begin{lemma}\label{lem:step4}
		Assuming that \textit{I.H.} holds, for arbitrary node $v$ and integer $\vartheta \in \set{0,1}$, node $v$ is $r$-$\vartheta$-reachable if $r^{(\vartheta)}_v = r$ in execution $\calE$. 
	\end{lemma} 
	A node $v$ sets its variable $r^{(\vartheta)}_v$ to an integer $r$ when round $r$ is either the first round in which $v$ receives a token or the first incomplete round for $v$. 
	In both cases we need to show that there actually exists a uniform alternating path of length $r$ for $v$. 
	We will see that the challenging case is when $v$ sets its variable $r^{(\vartheta)}_v$ to $r$ because of having round $r$ as its first incomplete round. 
	Let us first present the outline of proving the claim in this challenging case.
	If node $v$ sets its variable $r^{(\vartheta)}_v$ to $r$ because of having round $r$ as its first incomplete round, then there must clearly exist some edge $\set{u, w}$ such that a proper fraction of the flow of $\set{u, w}$ reaches $v$ in round $r$. 
	In this case, we show that there is actually a uniform alternating path of length $r$ for $v$ that contains $\set{u, w}$. 
	Let us assume that $v$ is $r'$-reachable. 
	We first show that there are two shortcuts of $u$ and $w$ such that they have no common $\ell$-reachable node for $\ell\geq r'$ and only one of them contains $v$ (see \Cref{fig:detection}).
	We show this claim by presenting \Cref{alg:shortcuts-construction} that actually constructs the shortcuts (stated in \Cref{lem:disjoint-shortcuts}). 
	Then, we transform each of the shortcuts to an alternating path. 
	To do so, let $P_u$ be the concatenation of the shortcut of $u$ and $\langle u, w\rangle$, and let $P_w$ be the concatenation of the shortcut of $w$ and $\langle u, w\rangle$.
	Then, let $x_1, \dots , x_8$ be the nodes having two adjacent $0$-edges on $P_u$ or $P_w$ in \Cref{fig:detection}.
	Note that $u$ and $w$ might also have two adjacent $0$ edges on the paths, e.g., $x_8$.
	Let nodes $x_1, \dots , x_8$ be in ascending order by their $1$-reachabilities. 
	Then, we present \Cref{alg:disc} to transform the shortcuts to alternating paths in phases. 
	In phase $i$, if $x_i\in P_w$, the algorithm replace $P_w[f, x_i]$ with a uniform $1$-path of $x_i$ that is disjoint from $P_u[f, u]$ and contains $v$.	
	Otherwise, the algorithm replace $P_u[f, x_i]$ with a uniform $1$-path of $x_i$ that is disjoint from $P_w[v, w]$.
	Doing that, $x_i$ does not anymore have two adjacent $0$-edges on the paths, and hence we reduce the number of nodes with two adjacent $0$-edges on the paths by $1$ in each phase. 
	Therefore, eventually $P_u$ and $P_w$ become alternating paths, and the concatenation of $P_u[f, u]$, $\langle u, w\rangle$ and $P_w[f, w]$ becomes a uniform alternating path of length $r$ of $v$. 
	We present the described procedure to change the shortcuts to the corresponding alternating paths in \Cref{alg:disc}.

	\begin{figure}[t] 
		\centering	
	\begin{tikzpicture}[scale=0.70]

		\coordinate (u) at (-.5,0);
		\shade[shading=ball, ball color=black] (u) circle (.08);
		\node at (u) [below left = 0.5mm of u] {\footnotesize{$x_8 = u$}};
		
		\coordinate (w) at (.5,0);
		\shade[shading=ball, ball color=black] (w) circle (.08);
		\node at (w) [below right = 0.5mm of w] {\footnotesize{$w$}};
		
		\draw (w) -- (u); 
		
		\coordinate (v) at (2,5);
		\shade[shading=ball, ball color=black] (v) circle (.08);
		\node at (v) [left = 0.5mm of v] {\footnotesize{$v$}};
		
		\coordinate (f) at (-2,6);
		\shade[shading=ball, ball color=black] (f) circle (.08);
		\node at (f) [above left = 0.5mm of f] {\footnotesize{$f_i$}};
		
		\draw[color=black] (w) to [bend right=20] (v);
		\draw[color=black] (v) to [bend right=20] (f);
		\draw[color=black] (u) to [bend left=20] (f);

		%% Y %%
		\coordinate (j) at (-1.25,1.3);
		\shade[shading=ball, ball color=black] (j) circle (.08);
		\node at (j) [left = 0.5mm of j] {\footnotesize{$x_6$}};
		
		\coordinate (y1) at (-1.6,2);
		\shade[shading=ball, ball color=black] (y1) circle (.08);
		\node at (y1) [left = 0.5mm of y1] {\footnotesize{$x_5$}};
		
		\coordinate (y2) at (-1.9,3.2);
		\shade[shading=ball, ball color=black] (y2) circle (.08);
		\node at (y2) [left = 0.5mm of y2] {\footnotesize{$x_3$}};
		
		%\node at (-2.3, 4) {\footnotesize{$\vdots$}};
		
		\coordinate (yq) at (-2.07,4.5);
		\shade[shading=ball, ball color=black] (yq) circle (.08);
		\node at (yq) [left = 0.5mm of yq] {\footnotesize{$x_1$}};
		
		%% Z %%
		\coordinate (z1) at (1.18,1);
		\shade[shading=ball, ball color=black] (z1) circle (.08);
		\node at (z1) [left = 0.5mm of z1] {\footnotesize{$x_7$}};
		
		\coordinate (z2) at (1.62,2);
		\shade[shading=ball, ball color=black] (z2) circle (.08);
		\node at (z2) [left = 0.5mm of z2] {\footnotesize{$x_4$}};
		
		%\node at (1.4, 2.8) {\footnotesize{$\vdots$}};
		
		\coordinate (zq) at (2, 3.5);
		\shade[shading=ball, ball color=black] (zq) circle (.08);
		\node at (zq) [left = 0.5mm of zq] {\footnotesize{$x_2$}};
		
		%% P %% 
		\draw[color=red, ultra thick] (-1.9,5.9) to [bend right=20] (-.4,0.1);
		\draw [color=red, ultra thick] (-.4,0.1) -- (.4,0.1);
		\node at (y2) [right = 0.5mm of y2] {\color{red}$\boldmath{P_u}$\color{black}};
		
		\draw[color=blue, ultra thick] (-1.9,6.1) to [bend left=20] (2.1,5.1);
		\draw[color=blue, ultra thick] (2.1,5.1) to [bend left=20] (.6,-.1);
		\draw [color=blue, ultra thick] (.6,-.1) -- (-.4,-0.1);
		\node at (z2) [right = 0.5mm of z2] {\color{blue}$\boldmath{P_w}$\color{black}};

	\end{tikzpicture}	
		\caption{$P_u$ and $P_w$ are initially the concatenation of the shortcuts of $u$ and $w$ with $\langle u, w \rangle$. Nodes $x_1, \dots , x_8$ are the nodes with two adjacent $0$-edges on $P_u$ and $P_w$ in ascending order by their $1$-reachabilities. }
		\label{fig:detection}
	\end{figure}
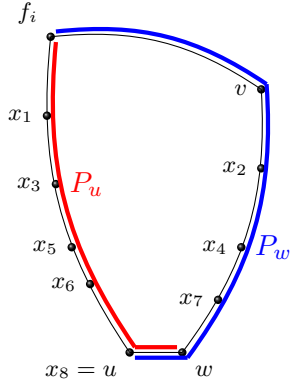

	Let us now start providing a few definitions and a few observations through a set of helper lemmas in the sequel, and finally at the end of this section we present the proof of \Cref{lem:step4}.
	For ease of discussion throughout this section, for every node $v$ in a cluster of the free node clustering, let $R(v)$ denote the reachability of $v$, i.e., $R(v):= l$ if $v$ is $l$-reachable.  
	
	\begin{lemma}\label{lem:disjoint-shortcuts}
          Assuming that \textit{I.H.} holds, let $v$ be
          an arbitrary $r'$-reachable node, and let $\set{u, w}$ be an
          arbitrary edge for which there exists a shortcut
          of $u$ and a shortcut of $w$ such that only one of them
          contains $v$.  Moreover, assume that there does not exist a
          node $x$ such that $x$ is $l$-reachable for $l\geq r'$ and
          such that all shortcuts of $u$ and $w$ contain $x$.  Then,
          there exists a shortcut of $u$ and a shortcut of $w$ with no
          common $l$-reachable node for any $l\geq r'$ such that only
          one of them contains $v$.
	\end{lemma}
	\begin{proof}
		Let us prove this lemma by constructing the desired shortcuts of $u$ and $w$. 
		We present an algorithm to actually construct the shortcuts (the pseudocode is given in \Cref{alg:shortcuts-construction}).
		Before we discuss the algorithm in detail, let us present the outline of the algorithm in the following. 
		Let us initially define paths $P_u:=\langle u \rangle$ and $P_w:=\langle w \rangle$.
		We update $P_u$ and $P_w$ in phases such that $P_u$ always remains a consecutive subpath of a shortcut of $u$ starting at $u$ and $P_w$ always remains a consecutive subpath of a shortcut of $w$ starting at $w$. 
		Let $t_u$ and $t_w$ be the variables that respectively maintain the last nodes of $P_u$ and $P_w$. 
		Variables $t_u$ and $t_w$ are initially set to $u$ and $w$ respectively, and they get updated after updating the paths in each phase. 
		Initially $P_u$ and $P_w$ are clearly disjoint. 
		Moreover, since at least one of $u$ or $w$ has a shortcut containing $v$, 	it initially holds that $\max\set{R(t_u), R(t_w)}>r'$ due to \Cref{lem:reach}. 
		In every phase, $P_u$ and $P_w$ are updated such that $\max\set{R(t_u), R(t_w)}$ decreases while $P_w$ and $P_u$ remain disjoint. 
		We will show that the phases are run in a way such that eventually $\max\set{R(t_u), R(t_w)} = r'$ and exactly one of paths $P_u$ or $P_w$ contains $v$. 
		Then, we construct the desired shortcuts $S_u$ and $S_w$ as follows. 
		$S_u$ is the concatenation of $P_u$ and a shortcut of $t_u$, and $S_w$ is the concatenation of $P_w$ and a shortcut of $t_w$.
		It is easy to see that $S_u$ and $S_w$ are shortcuts of $u$ and $w$ respectively. 
		Paths $S_u$ and $S_w$ have no common $l$-reachable node for any $l\geq r'$ since $P_u$ and $P_w$ are disjoint and also every node of the shortcuts of $t_u$ and $t_w$, except $t_u$ and $t_w$, is $l$-reachable for some $l<r'$. 
		Hence, further considering that exactly one of $P_u$ or $P_w$ contains $v$, it holds that exactly one of $S_u$ or $S_w$ contains $v$.
		In the rest of the proof, we will show how to update $P_u$ and $P_w$ in phases so that eventually after some phase $\max\set{R(t_u), R(t_w)} = r'$ and exactly one of paths $P_u$ or $P_w$ contains $v$.  
		  
		Let us label all nodes in $V_{\calG}$ as follows. 
		Every node that has a shortcut containing $v$ is labeled $v^{(1)}$, and every node that has a shortcut not containing $v$ is labeled $v^{(0)}$. 
		Note that a node can have both labels $v^{(1)}$ and $v^{(0)}$. 
		At the beginning of the execution, $t_u$ has label $v^{(\sigma)}$ and $t_w$ has label $v^{(1-\sigma)}$ for some integer $\sigma \in \set{0,1}$ since there exist a shortcut of $u$ and a shortcut of $w$ such that only one of them contains $v$.
		Therefore, to eventually have the paths $P_u$ and $P_w$ as desired, it is enough to decrease $\max\set{R(t_u), R(t_w)}$ in every phase while maintaining the following invariant:

\vspace{.2cm}
\begin{addmargin}[2em]{2em}
			\textit{Invariant $\mathcal{A}$: There exists an integer $\sigma\in\set{0,1}$ s.t. $t_u$ is labeled $v^{(\sigma)}$, $t_w$ is labeled $v^{(1-\sigma)}$, and $P_u \cap P_w = \emptyset$.}  
\end{addmargin}
\vspace{.2cm}

\noindent
	That is because at the beginning of the path construction execution, $\max\set{R(t_u), R(t_w)}>r'$ and invariant $\calA$ holds. 
	Moreover, based on \Cref{lem:reach}, the only node that is $l$-reachable for $l\leq r'$ and has a shortcut containing $v$ is node $v$ itself. 
	Therefore, one can observe that if one can always update the paths while maintaining invariant $\calA$, eventually $\max\set{R(t_u), R(t_w)} = r'$ after some updating phase.
		
		A single phase of updating $P_u$ and $P_w$ performs the following four steps. 
		However, it is crucial to note that a successful update in the second or third steps concludes the phase and avoids executing the next steps.  
		Without loss of generality, let us assume that at the beginning of the phase, it holds that $R(t_u) \leq R(t_w)$. 
		By symmetry, we can similarly state everything in the sequel for the case of having $R(t_w) \leq R(t_u)$.
		
		\begin{itemize}
			\item \textit{Step 1:} If $R(t_u) = R(t_w)$, $P_w \leftarrow P_w \circ \langle t_w, z\rangle$, where $z$ is a predecessor of $t_w$ that has label $v^{(1-\sigma)}$. 
			\item \textit{Step 2:} If $t_u$ has a predecessor $t\not\in P_w$ that is labeled $v^{(\sigma)}$, then $P_u \leftarrow P_u \circ \langle t_u, t\rangle$.
			\item \textit{Step 3:} Else, let $t\in P_w$ be a predecessor of $t_u$ that is labeled $v^{(\sigma)}$. If there is a shortcut $S$ of $u$ that is disjoint from $P_w[w, t]$, let $t'\in S$ be the closest node to $u$ on $S$ with $R(t')<R(t_u)$, and then $P_u \leftarrow S[u,t']$ and $P_w \leftarrow P_w[w, t]$. 
			\item \textit{Step 4:} Else, let $S'$ be any shortcut of $u$ or $w$ that does not contain $t$. Let $t_1\in P_w[w, t]\cap S'$ be the closets node to $t$ on $P_w$. Moreover, let $t_2$ be the closest node to $t_1$ in $S'[t_1, f]$ such that $R(t_2)<R(t_u)$. 
			Then,  $P_u \leftarrow P_u \circ \langle t_u, t\rangle$ and $P_w \leftarrow P_w[w, t_1] \circ S'[t_1, t_2]$.  
		\end{itemize}	
		
	Now let us show by induction that every phase successfully updates paths $P_u$ and $P_w$, where $\max\set{R(t_u), R(t_w)}$ decreases while invariant $\calA$ is maintained. 			
	To do so, consider an arbitrary phase of the path construction such that at the beginning of the phase, invariant $\calA$ holds and $\max\set{R(t_u), R(t_w)}>r'$. 
	Let $\hat{t}$ be $\max\set{R(t_u), R(t_w)}$ at the beginning of the phase. 
	We will show that after this updating phase, $\max\set{R(t_u), R(t_w)} < \hat{t}$ while invariant $\calA$ is maintained. 
	
	Consider an arbitrary node $s$, where $R(s)>R(v)$. 
	If $s$ has label $v^{(1)}$, then its neighbor on its shortcut that contains $v$ is labeled $v^{(1)}$(note that $s\neq v$ since $R(s)>R(v)$). 
	Hence, $s$ has a predecessor labeled $v^{(1)}$. 
	If $s$ has label $v^{(0)}$, then its neighbor on its shortcut that does not contain $v$ is labeled $v^{(0)}$, and hence, $s$ has a predecessor labeled $v^{(0)}$.
	Therefore, we can say that every node $s$ that is labeled $v^{(x)}$, for any $x\in \set{0,1}$, has a predecessor labeled $v^{(x)}$ if $R(s)>R(v)$.
	
	Now let us consider the first step.
	In case $R(t_w) = R(t_u)$, it updates $P_w$ such that $R(t_w)$ becomes less than $R(t_u)$.
	Let us assume that at the beginning of the phase, $R(t_u) = R(t_w)$. 
	Every node in $P_u$ is $l$-reachable for some $l\geq R(t_u)$. 
	However, every predecessor of $t_w$ is $l'$-reachable for some $l'<R(t_w)$. 
	Therefore, no predecessor of $t_w$ is in $P_u$. 
	Moreover, since $R(t_w)>R(v)$ and $t_w$ is labeled $v^{(1-\sigma)}$, $t_w$ has a predecessor labeled $v^{(1-\sigma)}$. 
	Hence, $P_w$ can be successfully updated such that $R(t_w)$ decreases while invariant $\calA$ remains true.  
	
	After \textit{Step 1}, we are sure that $R(t_w) \neq R(t_u)$ and hence $R(t_w) < R(t_u)$. 
	Now let us show that at least one of the three steps, from \textit{Step 2} to \textit{Step 4}, is successful in updating the paths. 
		If the second step is successful to update $P_u$, after the update, $R(t_u) = R(t) < \hat{t}$ while $t_w$ is still labeled $v^{(1-\sigma)}$ as $P_w$ is not changed.
		Therefore, since $t\not\in P_w$, invariant $\calA$ remains true and $\max\set{R(t_u), R(t_w)}$ decreases. 
		Now let us assume that the second step is not successful to update $P_u$. 
		Since $R(t_u)>R(v)$ and $t_u$ is labeled $v^{(\sigma)}$, node $t_u$ must have a predecessor $t$ that is labeled $v^{(\sigma)}$. 
		Since the second step is not successful, $t\in P_w$. 
		Node $t$ is the predecessor of $t_u$, and hence, $R(t)< R(t_u)\leq \hat{t}$.
		Moreover, due to the choice of $t'$, $R(t')<R(t_u)\leq \hat{t}$. 
		Therefore, if the third step is successful, by updating the paths, $t_u$ and $t_w$ are updated by setting $t_u\leftarrow t'$ and $t_w \leftarrow t$. 
		Hence, $R(t_u) = R(t') <\hat{t}$ and $R(t_w) = R(t) < \hat{t}$ while $P_u\cap P_w = \emptyset$. 
		
		Now let us assume that the second and third steps are not successful in updating the paths. 
		Let us show that the fourth step is then guaranteed to be successful in updating the paths. 
		Since $R(t)>R(v)$, based on the lemma's assumption, $t$ cannot be a common node of all the shortcuts of $u$ and $w$. 
		Therefore, there must be a shortcut $S'$ of $u$ or $w$ that does not contain $t$. 
		If $S'$ is a shortcut of $w$, it is easy to see that it has a common node with $P_w[w, t]$, e.g. $w$. 
		Otherwise, if $S'$ is a shortcut of $u$, then $S'$ must have a common node with $P_w[w, t]$ since the second step was not successful. 
		Hence, there exists the closest node $t_1\in P_w[w, t]\cap S'$ to $t$ on $P_w$. 
		Thus, $P_w[t_1, t]\cap S'[t_1, t_2] = \set{t_1}$. 
		Now let us show that $S'[t_1, t_2]\cap P_u = \emptyset$. 
		For the sake of contradiction, let us assume that $t'' \in S'[t_1, t_2]\cap P_u$ is the closest node to $t_1$ on $S'$. 
		Then, the concatenation $P_u[u, t'']$ and $S'[t'', t_2]$ is a shortcut of $u$ that is disjoint from $P_w[w, t]$, which contradicts the fact that the third step was not successful. 
		Therefore, $P_u$ and $P_w$ remain disjoint after the update while $\max\set{R(t_u), R(t_w)}$ decreases. 
		Note that since in \textit{Step 3} and \textit{Step 4}, we respectively choose $t'$ and $t_2$ as the closest nodes to $u$ and $t_1$ with reachability smaller than $R(t_u)$, we definitely have a phase in which $\max\set{R(t_u), R(t_w)} = r'$ and $v$ is the endpoint of one of the paths. 
		
              \end{proof}

\begin{algorithm}[htp]
  		\SetAlgoLined\DontPrintSemicolon
  		\SetKwFunction{algo}{Shortcut-Construction}\SetKwFunction{proc}{Update}
  		\SetKwProg{myalg}{}{}{}  %\SetKwProg{myalg}{Algorithm}{}{}
  		\myalg{\algo{$G, u, w, r'$}}{
			$t_w \leftarrow w$\\
			$t_u \leftarrow u$\\
			\While{$\max\set{R(t_u), R(t_w)} \neq r'$}{
				\texttt{Update}$(P_w, P_u, t_w, t_u)$
			}			
			$T_u \leftarrow$ a shortcut of $t_u$\\
			$T_w \leftarrow$ a shortcut of $t_w$\\
			$S_u \leftarrow P_u \circ T_u$\\
			$S_w \leftarrow P_w \circ T_w$\\
			
  			\nl \KwRet $S_u, S_w$\;}{}
  			\;
  		\setcounter{AlgoLine}{0}
  		\SetKwProg{myproc}{}{}{}  %  \SetKwProg{myproc}{Procedure}{}{}
  		\myproc{\proc{$P_w, P_u, t_w, t_u$}}{
				\textit{Step 1:}\\
  				\If{$R(t_u) = R(t_w)$}{
  					$P_w \leftarrow P_w \circ \langle t_w, z\rangle$ for a predecessor $z$ of $t_w$ that is labeled $v^{(1-\sigma)}$ \;
					$t_w \leftarrow z$ \;
  				}
				\textit{Step 2:}\\
				\If{$t_u$ has a predecessor $t\not\in P_w$ labeled $v^{(\sigma)}$}{
  					$P_u \leftarrow P_u \circ \langle t_u, t\rangle$ \;
					$t_u \leftarrow t$ \;
					\Return \;
  				}
				\textit{Step 3:}\\
				let $t\in P_w$ be a predecessor of $t_u$ that is labeled $v^{(\sigma)}$ \\
  				\If{$\exists$ shortcut $S$ of $u$ s.t. $S \cap P_w[w, t] = \emptyset$}{
  					 let $t'\in S$ be the closest node to $u$ on $S$ with $R(t')<R(t_u)$\;
					  $P_u \leftarrow S[u,t']$ \;
					  $t_u \leftarrow t'$\;
					  $P_w \leftarrow P_w[w, t]$\;
					  $t_w \leftarrow t$\;
					  \Return \;
  				}
				\textit{Step 4:}\\
				\Else{
					let $S'$ be any shortcut of $u$ or $w$ that does not contain $t$\;
					let $t_1\in P_w[w, t]\cap S'$ be the closets node to $t$ on $P_w$\;
					let $t_2$ be the closest node to $t_1$ in $S'[t_1, f]$ such that $R(t_2)<R(t_u)$\;
					$P_u \leftarrow P_u \circ \langle t_u, t\rangle$\;
					$t_u \leftarrow t$\;
					$P_w \leftarrow P_w[w, t_1] \circ S'[t_1, t_2]$\;
					$t_w \leftarrow t_2$\;
				}
		}
		\caption{The Shortcuts Construction Algorithm.}
		\label{alg:shortcuts-construction}
	 \end{algorithm} 
	 
%\newpage	 
	 
	In the following, we borrow a few notations and definitions from \cite{vazirani13}. 
	Let us henceforth consider an implicit direction on shortest $0$-paths and $1$-paths of the nodes from the cluster centers towards the nodes. 
	Let us consider a shortest uniform $\vartheta$-path $P$ of any node $u$ and a shortest uniform $\varphi$-path $Q$ of any node $w$ for $\vartheta, \varphi \in \set{0,1}$ to define the following notions. 
	
	\begin{definition}[Common Edge]
		An edge $e$ on both paths $P$ and $Q$ is called a \textit{common edge} of $P$ and $Q$.
	\end{definition}
	
	\begin{definition}[Forward/Backward Edge]
		Considering the implicit direction on $P$ and $Q$, if both $P$ and $Q$ traverse $e$ in the same direction, $e$ is called a \textit{forward edge} and otherwise a \textit{backward edge}.
	\end{definition}
	
	\begin{definition}[Separator]
		If $P$ and $Q$ have a common edge and the induced graph by paths $P$ and $Q$ gets disconnected by removing edge $e$, then $e$ is called a \textit{separator}.
	\end{definition}	
	Note that a separator edge of $P$ and $Q$ must be a forward edge. 
	\begin{definition}[Tenacity]
		For a $r'$-$\vartheta$-reachable and
                $r$-$(1-\vartheta)$-reachable node $v$, the \textit{tenacity} of $v$ is defined to be $r+r'$, denoted by $tn(v)$. 
		For a $\varphi$-edge $e=\set{u, w}$ that $u$ is
                $\ell_1$-$(1-\varphi)$-reachable and
                $\ell_2$-$(1-\varphi)$-reachable, the tenacity of $e$ is defined to be $\ell_1 + \ell_2 + 1$, denoted by $tn(e)$. 
	\end{definition}
	
%%% Comment
\iffalse	
	\begin{center}
		{\Large \color{red} The following lemma is the observation that simplifies the path concatenations.}
	\end{center}
\fi
%% Comment	
	
	\begin{lemma}\label{lem:sep}
		Assuming that \textit{I.H.} holds, let $P$ be a shortest uniform $\psi$-path of any node $w$, and $Q$ be a shortest uniform $\varphi$-path of any node $u$ for $\psi, \varphi \in \set{0,1}$.
		Let us assume that there exists a node $v\in P$ of tenacity greater than $|P|+|Q|$ such that $v\not\in Q$. 
		If $u\neq w$ or $\varphi \neq \psi$, the closest common edge of $Q$ and $P[v , w]$ to the cluster center on $Q$ is a separator. 
	\end{lemma}
	\begin{proof}
		Let $f$ be the cluster center of $u$ and $w$, and let $P[f, v]$ be a $(1-\vartheta)$-path of $v$ for any $\vartheta \in \set{0,1}$.
		Let $e:=\set{x', x}$ be the closest common edge of $Q$ and $P[v , w]$ to $f$ on $Q$ such that $x'$ is closer to $f$ than $x$ on $Q$.
		
		Let us first show that $e$ is a forward edge of $P$ and $Q$. 
		For the sake of contradiction, let us assume that $e$ is a backward edge of $P$ and $Q$. 
		Let $L$ be the concatenation of $Q[f , x']$ and $P[v , x']$. 
		Since $e$ is a backward edge, $L$ is an alternating walk. 
		Due to the choice of $e$, paths $Q[f , x']$ and $P[v , x]$ are disjoint and hence $L$ is a path.
		Moreover, since $P$ is alternating and $P[f, v]$ is a $(1-\vartheta)$-path, $v$'s adjacent edge on $L$ is a $\vartheta$-edge.
		Hence, $L$ is a uniform $\vartheta$-path of $v$.  
		Consequently, the shortest uniform $\vartheta$-path of $v$ is of length at most $|Q[f , x']| + |P[v , x']|$ while the shortest uniform $(1-\vartheta)$-path of $v$ is of length at most $|P[f , v]|$. 
		Hence, tenacity of $v$ should be at most $|P[f , v]| + |Q[f , x']| + |P[v , x']|$. 
		Since $|P[f , v]| + |P[v , x']| \leq |P|$ and $|Q[f , x']| < |Q|$, tenacity of $v$ should therefore be at most $|P|+|Q|$, which contradicts the assumed tenacity of $v$. 
		It concludes that $e$ is a forward edge of $P$ and $Q$.
		
		To show that $e$ is a separator of $P$ and $Q$, it is enough to show that $Q[f , x']\cap P[x , w] = \emptyset$ and $P[f , x']\cap Q[x , u] = \emptyset$. 
		The former equality holds as a direct implication of the choice of $e$. 
		Let us then show that $P[f , x']\cap Q[x , u] = \emptyset$.
		Let us first show that $P[f , v]\cap Q[x , u] = \emptyset$.
		For the sake of contradiction, let us assume that $P[f , v]$ and $Q[x , u]$ have a common edge. 
		Let $e' = \set{y', y}$ be the closest such edge to $f$ on $P$ such that $y'$ is closer to $f$ than $y$ on $P$. 
		Let also $e''= \set{z', z}$ be the closest common edge of $Q[x' , y]$ and $P[v , x]$ to $v$ on $P$ such that $z'$ is closer to $v$ than $z$ on $P$. 
		Note that if $Q[x , y]$ and $P[v , x']$ are disjoint, then $e$ and $e''$ are the same, i.e., $z' = x'$ and $z  = x$. 
		If $e \neq e''$, then by applying a similar argument to that in the second paragraph of this proof, we can show that $e''$ is also a forward edge. 
		
		Now let us show that $e'$ is also a forward edge. 
		For the sake of contradiction let us assume that $e'$ is a backward edge. 
		To show a contradiction, let us prove that $tn(v) < |P| + |Q|$.
		Since $e'$ is assumed to be a backward edge, path $Y$, the concatenation of $P[v , z']$, $Q[z' , y']$ and $P[f , y']$ is a uniform $\vartheta$-path of $v$.
		Then, to show that $tn(v) < |P| + |Q|$, it is enough to show that $|Y| \leq |Q|$.	
		Note that $|P[v , z']| + |P[f , y']| < |P[f , x']|$ and $|Q[z' , y']| < |Q[x' , u]|$. 
		Hence, to show that $|Y| \leq |Q|$, it suffices to prove that $|P[f , x']| \leq |Q[f , x']|$. 
		Since $Q[f , x']$ and $P[x' , w]$ are disjoint and $e = \set{x', x}$ is a forward edge of $P$ and $Q$, the concatenation of $Q[f , x']$ and $P[x' , w]$ is a uniform $\psi$-path of $w$. 
		Moreover, since $P$ is a shortest uniform $\psi$-path of $w$, $|Q[f , x']| + |P[x' , w]| \geq |P|$. 
		Therefore, $|P[f , x']| \leq |Q[f , x']|$. 
		This concludes that $e'$ is a forward edge. 
		
		Considering the choice of $e'$ and the fact that $e'$ is a forward edge, the concatenation of $P[f , y]$ and $Q[y , u]$, must be a uniform $\varphi$-path of $u$. 
		We already proved that $|P[f , x']| \leq |Q[f , x']|$.
		Therefore, since $|P[f , y]| < |P[f , x']|$, it holds that $|P[f , y]| \leq |Q[f , x']|$. 
		Moreover, $|Q[y , u]| < |Q[x' , u]|$. 
		Hence, $|P[f , y]| + |Q[y , u]| < |Q[f , x']| + |Q[x' , u]| = |Q|$. 
		This contradicts the assumption that $Q$ is a shortest uniform $\varphi$-path of $u$. 
		Hence, $Q[x , u]$ and $P[f , v]$ have no common edge.
		
		It is left to show that $Q[x , u]$ and $P[v , x']$ have no common edge. 
		For the sake of contradiction, let us assume that edge $\hat{e} = \set{s, s'}$ is the closest common edge of $Q[x , u]$ and $P[v , x']$ to $v$ on $P$, where $s'$ is closer to $v$ than $s$ on $P$. 
		With a similar argument to that in the second paragraph of this proof, one can show that $\hat{e}$ is a forward edge. 
		Therefore, path $Z$, the concatenation of $P[f , s]$ and $Q[s , u]$ is uniform $\varphi$-path of $u$.
		We already proved that $|P[f , x']| \leq |Q[f , x']|$. 
		Therefore, since $|P[f , s]| < |P[f , x']|$, it holds that $|P[f , s]| < |Q[f , x']|$.  
		Moreover, $|Q[s , u]| < |Q[x' , u]|$. 
		Therefore, $|P[f , s]| + |Q[s , u]| < |Q[f , x']| + |Q[x' , u]| = |Q|$.
		Therefore, $Z$ is of length less than $|Q|$ and a uniform $\varphi$-path of $u$. 
		This contradicts the assumption on the length of $Q$ as the shortest $\varphi$-path of $u$. 
		This leads to having $Q[x, u]$ and $P[f , x']$ with no common edge, which concludes the proof. 	
		
	\end{proof}

	\begin{lemma}\label{lem:token-time}
		Assuming that \textit{I.H.} holds, let an arbitrary node $w$ send a token to its neighbor $u$ in round $r_w\leq r$ of execution $\calE$. 
		Letting $S_w$ be a shortcut of $w$, it holds that $r_w = \|S_w \circ \langle w, u \rangle \|$.   
	\end{lemma}
	\begin{proof}
		Let $w$ be $\ell'$-$\vartheta$-reachable and $\ell$-$(1-\vartheta)$-reachable for some integers $\vartheta\in \set{0,1}$ and $\ell' < \ell$. 
		Based on \Cref{lem:promoted}, it holds that 		
		\begin{equation}\label{eq:t1}
			\| S_w \| = \ell' \ \text{.}
		\end{equation}
		
		Let us consider the following two possibilities separately:
		\begin{enumerate}
			\item[(a)] $\set{w, u}$ is a $(1-\vartheta)$-edge.\\
				Due to \textit{I.H.}, $w$ sends the token to $u$ in round $\ell' + 1$. 
				Moreover, since $w$'s adjacent edge on $S_w$ is a $\vartheta$-edge and $\set{w, u}$ is a $(1-\vartheta)$-edge, it holds that $d(S_w \circ \langle w, u \rangle, w) = 0$. 
				Hence, 
				\begin{align}
					r_w &\ = \ell' + 1 \nonumber \\
				   		   &\overset{(\ref{eq:t1})}{=} \|S_w\| + 1 \nonumber \\
						   & \ = \|S_w \circ \langle w, u \rangle \| - d(S_w \circ \langle w, u \rangle, w) - 1 + 1  \nonumber \\
						   & \ = \|S_w \circ \langle w, u \rangle \| \ \text{.} \nonumber
				\end{align}

			\item[(b)] $\set{w, u}$ is a $\vartheta$-edge. \\
				Due to \textit{I.H.}, $w$ sends the token to $u$ in round $\ell + 1$. 
				However, since both $w$'s adjacent edge on $S_w$ and $\set{w, u}$ are $\vartheta$-edges (i.e., $\vartheta = 0$), it holds that $d(S_w \circ \langle w, u \rangle, w) = \ell - \ell'$.
				Hence, 
				\begin{align}
					r_w &\ = \ell + 1 \nonumber \\
					       &\ = \ell - \ell' + \ell' + 1 \nonumber \\
					       &\ = d(S_w \circ \langle w, u \rangle, w) + \ell' + 1 \nonumber \\
				   	      &\overset{(\ref{eq:t1})}{=} d(S_w \circ \langle w, u \rangle, w) + \|S_w\| + 1 \nonumber \\
					      &\ = \|S_w \circ \langle w, u \rangle \| \ \text{.} \nonumber
				\end{align}
		\end{enumerate}	
              \end{proof}

		\begin{lemma}\label{lem:tenacity1}
			Assuming that \textit{I.H.} holds, let $v$ be an arbitrary $r'$-reachable node for which round $r$ is the first incomplete round. 
			If $v$ receives a proper fraction of the flow of an edge $e$ in round $r$ of execution $\calE$, then $tn(e) \leq tn(v)$. 
		\end{lemma}
	\begin{proof}
		%It is easy to see the lemma's correctness in case $v$ is an endpoint of $e$.
		%Hence, let us consider the non-trivial case where $e$ is not an adjacent edge of $v$.  
		To prove the lemma, we show that tenacity of $v$ is at least $r+r'$ and tenacity of $e$ is $r+r'$. 
		Let $f$ be the cluster center of node $v$. 
		
		Let us first show that tenacity of $v$ is at least $r+r'$. 
		For the sake of contradiction, let us assume otherwise. 
		Node $v$ is $r'$-reachable. 
		Let $\vartheta \in \set{0,1}$ be the integer such that $v$ is $r'$-$(1-\vartheta)$-reachable.
		Therefore, $v$ is $r''$-$\vartheta$-reachable for some $r''<r$. 
		Hence, $r''$ must be an incomplete round for $v$, which contradicts the assumption on $r$ as the first incomplete round $v$. 
		Hence, 
		\begin{equation}\label{eq:r}
			tn(v) \geq r + r'\ \text{.}
		\end{equation}	
		Next, we show that tenacity of $e$ is $r+r'$. 
		Let $u$ and $w$ be the two endpoints of $e$ such that $w$ has a shortcut $S$ along which a proper fraction of the flow of $e$ is sent and received by $v$ in round $r$. 
		Let $r_u$ be the round in which $u$ sends a token to $w$, and $r_w$ be the round in which $w$ sends a token to $u$. 
		Let $e$ be a $\psi$-edge for any $\psi \in \set{0,1}$. 
		Therefore, since the clustering is the well-formed $(r-1)$-radius free node clustering of $\calG$, $u$ and $w$ are respectively $(r_u - 1)$-$(1-\psi)$-reachable and $(r_w - 1)$-$(1-\psi)$-reachable. 
		Therefore, 
		\begin{equation}\label{eq:r0}
			tn(e) =  r_u + r_w - 1\ \text{.}
		\end{equation}		
		Let us now show that $d(S, v) = 0$. 
		Let $v_0$ and $v_1$ be $v$'s neighbors on $S$ such that $v_0$ is the predecessor of $v$. 
		If $\vartheta = 0$, then $v$ is connected to $v_0$ by a $1$-edge. 
		Then, $v$ must be connected to $v_1$ by a $0$-edge as it can only have one adjacent $1$-edge. 
		Therefore, it holds that $d(S, v) = 0$.
		Now let us consider the possibility of $\vartheta = 1$. 
		Hence, $v$ is connected to $v_0$ by a $0$-edge.
		Now we show that $v$ is connected to $v_1$ by a $1$-edge. 
		For the sake of contradiction, let us assume otherwise. 
		Then, since $v_1$ sends a flow of $e$ to $v$ in round $r$, node $v$ must have send a token to $v_1$ (over a $0$-edge) in some round $r'' < r$. 
		Therefore, due to \textit{I.H.}, $v$ must be a $(r'' - 1)$-$1$-reachable.
		Hence, round $r'' - 1$ must be an incomplete round for $v$, which contradicts round $r$ being the first incomplete round of $v$. 
		Therefore, we can conclude that all cases, it holds that 
		\begin{equation}\label{eq:r1}
			d(S, v) = 0 \ \text{.}
		\end{equation}
Now let us conclude the proof as follows:				
		\begin{align}	
			tn(e) &\ \ \overset{(\ref{eq:r0})}{=} \ r_u + r_w - 1 \ \nonumber \\
			&\overset{(\ref{lem:token-time})}{=} r_u + \|\langle u, w \rangle \circ S_w \| - 1 \ \nonumber \\
			&\ \ =\ \ \ r_u + \|\langle u, w \rangle \circ S_w[w, v] \| + d(S, w) + \|S[v, f]\| - 1 \ \nonumber \\
			&\ \overset{(\ref{eq:r1})}{=}\  r_u + \|\langle u, w \rangle \circ S_w[w, v] \| - 1 + \|S[v, f]\| \ \nonumber \\
			&\ \ = \ \ \ r_u + \|\langle u, w \rangle \circ S_w[w, v] \| - 1 + r' \ \nonumber \\
			&\ \ = \ \ \ r + r'  \label{eq:r2}
		\end{align}
		
		\Cref{eq:r} and \Cref{eq:r2} conclude the proof. 
		
              \end{proof}

	\begin{lemma}\label{lem:tenacity2}
		 Assuming that \textit{I.H.} holds, let $e=\set{u,w}$ be an arbitrary edge over which a flow is generated, where $u$ sends a token to $w$ in round $r_u$, and $w$ sends a token to $u$ in round $r_w$ in execution $\calE$. 
		 If $\max\set{r_u, r_w}<r$, tenacity of every node that has two adjacent free edges on the concatenation of $e$ and a shortcut of $u$ or $w$ is less than $tn(e)$. 
	\end{lemma}
	\begin{proof}
		Let $e$ be a $\psi$-edge for any $\psi \in \set{0,1}$. 
		The nodes are provided with the well-formed $(r-1)$-radius free node clustering of $\calG$ and $\max\set{r_u, r_w}<r$. 
		Moreover, $u$ and $w$ send tokens over the adjacent $\psi$-edge $e$ in rounds $r_u$ and $r_w$ respectively. 
		Therefore, $u$ is $(r_u-1)$-$(1-\psi)$-reachable, and $w$ is $(r_w-1)$-$(1-\psi)$-reachable. 
		Hence, $tn(e) = r_u + r_w -1$. 
		Without loss of generality, consider path $S$, the concatenation of $e$ and an arbitrary shortcut of $u$ (the argument for $w$ is symmetric). 
		Let $x$ be an arbitrary node with two adjacent $0$-edges on $S$.
		Let $x$ be $\ell_0$-$0$-reachable and $\ell_1$-$1$-reachable. 
		For simplicity, let us separately study the cases where $u = x$ and $u \neq x$. 
		First consider the case where $u = x$. 
		It is easy to see that $\ell_1 + 1 = r_u$. 
		Moreover, since $w$ is not a predecessor of $u$, it must hold that $\ell_0$, that is the first round of receiving any token for $w$ is less than $r_w$. 
		Therefore, $\ell_0 + \ell_1 < r_u + r_w - 1 = tn(e)$. 
		
		Now let us consider the case when $u \neq x$. 
		Then, $\ell_1 < R(u) < r_u$. 
		Moreover, $\ell_0 < R(u) < r_w$ since $w$ is not a predecessor $u$ and hence $u$ must not receive a token from $w$ in the first round of receiving tokens. 
		Therefore, $\ell_0 + \ell_1 < r_u + r_w - 1 = tn(e)$. 	
		
              \end{proof}

	Here we present the proof of \Cref{lem:step4} as the final step of this section:
	
	\begin{proof}[Proof of \Cref{lem:step4}]
		We show that an arbitrary node is $r$-$\vartheta$-reachable if $r^{(\vartheta)}_v = r$ after $r$ rounds of $\calE$.
		Let us assume that $v$ is in the cluster centered at $f$. 
		There are two ways that $r^{(\vartheta)}_v$ can be set to $r$; either $r$ is the first round in which $v$ receives a token or it is the first incomplete round for $v$. 
		Let us first consider the former case, and let $v'$ be a neighbor of $v$ that sends token $f$ to $v$ in round $r$. 
		Then, $r^{(1-\vartheta)}_{v'} = r-1$, and since $\calD(r-1)$ is the well-formed $(r-1)$-radius free node clustering of $\calG$, $v'$ is $(r-1)$-$(1-\vartheta)$-reachable. 
		Let $P_{v'}$ be a shortest uniform $(1-\vartheta)$-path of length $r-1$ of $v'$. 
		Then, $v\not\in P_{v'}$ since otherwise $v$ should have received a token before round $r$, which is contradictory. 
		Therefore, path $P_v$, the concatenation of $P_{v'}$ and $\langle v', v\rangle$, is a uniform $\vartheta$-path of length $r$ of $v$. 
		Path $P_v$ is a shortest uniform $\vartheta$-path of $v$, since otherwise, $r^{(\vartheta)}_v$ would have been set to an integer smaller than $r$.
		Hence, $v$ is $r$-$\vartheta$-reachable. 
		
		For the rest of the proof, we consider the latter case where $r$ is the first incomplete round for $v$.
		Let $v$ be $r'$-reachable for some integer $r'<r$. 
		Let $e$ be an edge such that $v$ receives an incomplete flow of edge $e$ in round $r$. 
		Since $v$ receives a flow of edge $e$, all the shortcuts of the two endpoints of $e$ have no $l$-reachable common node for any $l\geq r'$. 
		That is because otherwise based on \Cref{lem:discards}, the common node of all the shortcuts of the two endpoints of $e$ that has maximum reachability discards the whole flow of $e$, and $v$ never receives a flow of $e$, which contradicts the fact that $v$ receives a flow of $e$.
		Therefore, based on \Cref{lem:disjoint-shortcuts}, there exist a shortcut $T_u$ of $u$ and a shortcut $T_w$ of $w$ with no common $l$-reachable node for any $l\geq r'$ such that exactly one of them contains $v$. 
		Without loss of generality, let us assume that $T_w$ contains $v$. 
		Let $e$ be a $\varphi$-edge for an integer $\varphi \in \set{0,1}$. 
		Then, we run \textit{the path construction algorithm}, whose pseudocode is given by \Cref{alg:disc}, to construct a shortest uniform $\vartheta$-path of length $r$ of $v$.
		
		\begin{algorithm}[htp]
  		\SetAlgoLined\DontPrintSemicolon
  		\SetKwFunction{algo}{Path-Discovery}\SetKwFunction{proc}{Resolve}
  		\SetKwProg{myalg}{}{}{}  %\SetKwProg{myalg}{Algorithm}{}{}
  		\myalg{\algo{$G, T_u, T_w, v, \vartheta$}}{
  			$P_v \leftarrow$ a shortest $(1-\vartheta)$-path of $v$; \\
			$P_u \leftarrow \langle w, u \rangle \circ T_u$; \\
			$P_w \leftarrow \langle u, w \rangle \circ T_w[w, v] \circ P_v$; \\
			Let $U$ be the set of nodes with two adjacent $0$-edges on $P_u$\\
			Let $W$ be the set of nodes with two adjacent $0$-edges on $P_w$\\
			Let $T \leftarrow U\cup W$\\ 
			\While{$T\neq \emptyset$}{
				Let $s\in T$ be the node with the smallest $1$-reachability\;
				\texttt{Resolve}$(P_w, P_u, s)$\;
				$T \leftarrow T\setminus \set{s}$\;
			}

  			\nl \KwRet $P_u[u, f_i] \circ \langle u, w \rangle \circ P_w[w, v]$\;}{}
  			\;
  		\setcounter{AlgoLine}{0}
  		\SetKwProg{myproc}{}{}{}  %  \SetKwProg{myproc}{Procedure}{}{}
  		\myproc{\proc{$P_u, P_w, s$}}{
  			\If{$s\in P_u$}{
  				\If{$s$ has a shortest uniform $1$-path $P'$ disjoint from $P_w[v , w]$}{
	 				$P_u \leftarrow P' \circ P_u[s , u]$\; 
         			}\;
  			}
  			\ElseIf{$s\in P_w$}{
  				\If{$s$ has a shortest uniform $1$-path $P'$ containing $v$ and disjoint from $P_u[f , u]$}{
	 				$P_w \leftarrow P' \circ P_w[s , w]$\; 
         			}
  			}
 
		}
		\caption{The Path Construction Algorithm.}
		\label{alg:disc}
	 \end{algorithm} 
		
		The algorithm gradually constructs a shortest uniform $(1-\varphi)$-path of $u$, i.e., $P_u[f, u]$, and a shortest uniform $(1-\varphi)$-path of $w$, i.e., $P_w[f, w]$, such that $P_w[f, w]$ contains $v$, $P_u[f, u]$ does not contain $v$, $P_u[f, u]$ and $P_w[v, w]$ have no common node, and $|P_u[f, u]|+|P_w[f, w]| = r+r'-1$.
		Then, the concatenation of $P_u[f, u]$, $\langle u, w \rangle$ and $P_w[v, w]$ is a uniform $\vartheta$-path of length $r$ of $v$. 
		Every iteration of the while-loop in the algorithm is called \textit{successful} if the procedure \texttt{Resolve} successfully updates one of the paths $P_w$ or $P_u$. 
		Then, to prove that the algorithm returns a shortest uniform $\vartheta$-path of length $r$ of $v$, we show that all iterations of the while-loop are successful, and hence the concatenation of $P_u[f, u]$, $\langle u,w \rangle$, and $P_w[w, v]$ is a uniform $\vartheta$-path of length $r$ of $v$. 
		Initially, $P_u[f, u]$ and $P_w[v, w]$ are disjoint. 
		If one iteration of the while is successful in updating one of the paths, $P_u[f, u]$ and $P_w[v, w]$ remain disjoint. 
		Considering an arbitrary iteration $j$ of the while
                loop, we show that if all previous iterations were successful in updating $P_u[f, u]$ and $P_w[v, w]$, iteration $j$ is also successful in updating $P_u[f, u]$ or $P_w[v, w]$. 
		Now let us consider the following two cases separately when procedure \texttt{Resolve$(T_w, T_u, s)$} is called. 
		\begin{enumerate}
			\item[(I)] $s\in P_u$
			
				Here we show that there exists some shortest uniform $1$-path of $s$ that is disjoint from $P_w[v , w]$.
				
				Let us first show that there is a shortest uniform $1$-path of $s$ that does not contain $v$. 
				To do so we search and find such a path in a procedure explained as follows. 
				First observe that $s$ must have a shortest uniform $1$-path. 
				Let $Z$ be an arbitrary shortest uniform $1$-path of $s$. 
				If $Z$ does not contain $v$, we are done with the search and $Z$ is one of such paths. 
				Let us thus assume that $Z$ contains $v$. 
				%Then, we find a shortest uniform $1$-path of $s$ as follows. 
				We will first show that the adjacent edge of $v$ on $Z[v , s]$ is a $\vartheta$-edge. 
				We will then show that the concatenation of $Z[v , s]$ and $P_u[f , s]$, that is an alternating walk, is of length less than $r$. 
				Thus, $Z[v , s]$ and $P_u[f , s]$ must have a common edge. 
				Then, we will argue how having such a common edge leads to the existence of a shortest uniform $1$-path of $s$ that does not contain $v$. 
				
				Let us show that the adjacent edge of $v$ on $Z[v , s]$ is a $\vartheta$-edge.
				Node $s$ sends a token to its neighbor in $P_u[s , w]$ in round $|Z| + 1$. 
				Moreover, considering \Cref{lem:reach2}, all the nodes in $P_u[s , u]$ excluding $s$ have reachability greater than $|Z|$, and hence they can only send tokens in rounds greater than $|Z|$. 
				Therefore, node $u$ that is in $P_u[s , u]$ sends a token to $w$ in a round greater than $|Z|$, i.e., $r_u > |Z|$. 
				We also know that $r_u \leq r$ since $T_w$ contains $v$ and $v$ receives a flow of $e$ along $T_w$ in round $r$.
				This overall concludes that $|Z| < r$. 
				Then, since $|Z[f , v]| < |Z|$, path $Z[f , v]$ is an alternating path of length less than $r$ of $v$. 
				Hence, the adjacent edge of $v$ on $Z[f , v]$ is not a $\vartheta$-edge, since otherwise $v$ would have a uniform $\vartheta$-path of length less than $r$ and hence have received an incomplete flow in a round before round $r$. 
				 
				Now let us show that the concatenation of $Z[v , s]$ and $P_u[f , s]$ is of length less than $r$.
				Path $P_u[f , s]$ is a shortest uniform $0$-path of $s$.
				Thus, $tn(s) = |P_u[f , s]| + |Z|$. 
				Based on \Cref{lem:tenacity2}, $tn(s) < tn(e)$, and based on \Cref{lem:tenacity1}, $tn(e)< tn(v)$. 
				Moreover, $tn(v) = r + r'$. 
				Therefore, $|P_u[f , s]| + |Z| < r + r'$. 
				Note that since $Z[f , v]$ is a uniform alternating path of $v$, $|Z[f , v]| \geq r'$.
				Hence, $|Z[v , s]| + |P_u[f , s]| < r$. 
				This concludes that $Z[v , s]$ and $P_u[f , s]$ must have a common edge as otherwise $v$ would have a uniform $\vartheta$-path of length less than $r$. 
				
				The last step is to show that this common edge leads to the existence of a shortest uniform $1$-path of $s$ that does not contain $v$. 
				Due to \Cref{lem:sep}, the closest common edge of $P_u[f , s]$ and $Z[v , s]$ to $f$ on $P_u$ is a separator of $Z$ and $P_u[f , s]$. 
				Let this edge be $\set{t_1, t_2}$.
				Then, the concatenation of $P_u[f , t_1]$ and $Z[t_1 , s]$ is a shortest uniform $1$-path of $s$ that does not contain $v$. 
				Let $\calS$ be the set of all shortest uniform $1$-paths of $s$ that do not contain $v$. 
				
				Let us now consider the following two cases separately: (recall that $T$ is the set of nodes that have two adjacent $0$-edges on current $P_w$ or $P_u$.)

				\begin{enumerate} 
					\item $P_w\cap T = \emptyset$ \\
						Here we show that there is some path in $\calS$ that is disjoint from $P_w[v , w]$. 
						We first show that there is a path in $\calS$ that does not have a common edge with $P_w[v , w]$.
						This implies that the path dose not contain any node in $P_w(v , w)$ since the path is an alternating path.
						Moreover, since the path is in $\calS$, it does not contain $v$. 
						Then, at the end we show that the path does not also contain $w$. 
						
						For the sake of contradiction, let us assume that there is no path in $\calS$ that has no common edge with path $P_w[v , w]$. 
						Considering $Q$ an alternating path starting at $f$, let the closest common edge of $Q$ and any path $Q'$ to $f$ on $Q$ be denoted by $\partial(Q \ on \ Q')$.  
						Let $e_1:=\set{s_1, s_1'}\in P_w[v , w]$ be the closest edge to $w$ such that $e_1$ is $\partial(S \ on \ P_w[v , w])$ for some $S\in \calS$. 
						Path $P_w[f , w]$ is a shortest uniform $(1-\varphi)$-path of $w$. 
						Moreover, $S$ is a shortest uniform $1$-path of $s$ with length less than the length of a shortest $(1-\varphi)$-path of $u$. 
						Hence, $|P_w[f , w]| + |S| < tn(e) < tn(v)$. 
						Then, based on \Cref{lem:sep}, $e_1$ is a separator. 
						Therefore, $L_1:= P_w[f , s_1] \circ S[s_1 , s]$ is a shortest uniform $1$-path of $s$ that contains $v$.
						Since $|L_1| + |P_u[f , u]| = tn(s) < tn(e) < tn(v) = r+r'$ and $|L_1[f , v]| \geq r'$, the concatenation of $L_1[v , s]$ and $P_u[f , s]$ is an alternating walk of length less than $r$ from $f$ to $v$ that contains a $\vartheta$-edge of $v$. 
						Therefore, this concatenation cannot be a path. 
						Hence, $P_u[f , s]$ and $L_1[v , s]$ must have a common edge. 
						Then, based on \Cref{lem:sep}, the closest common edge of $P_u[f , s]$ and $L_1[v , s]$ to $f$ on $P_u$ is a separator of $L_1$ and $P_u[f , s]$.
						Let $e_2 = \set{s_2, s_2'}$ be the edge. 
						Then, $L_2 := P_u[f , s_2] \circ L_1[s_2 , s]$ is a shortest uniform $1$-path of $s$ that does not contain $v$. 
						Note that $L_2[s_2 , s]$ is the same as $L_1[s_2 , s]$ and the same as $S[s_2 , s]$. 
						Therefore, since $e_1$ is a separator of $S$ and $P_w$, $L_2[s_2 , s]$ has no common edge with $P_w[f , s_1]$. 
						Moreover, $L_2[s_2 , s]$ has no common edge with $P_w[s_1 , w]$ as otherwise it contradicts the choice of $e_1$. 
						Therefore, $L_2$ is a shortest uniform $1$-path of $s$ that does not contain $v$ and has no common edge with $P_w[v , w]$.   
						
						It is left to show that $L_2$ does not also contain $w$. 
						Since $P_w\cap T = \emptyset$, it holds that $P_w[f , u]$ is alternating, and hence $w$ has an adjacent $1$-edge on $P_w$. 
						For the sake of contradiction, let us assume that $L_2$ contains $w$, and consequently it must contain the adjacent $1$-edge of $w$. 
						All the nodes in $L_2[f , s)$ have reachability of less than $R(s)$, and hence less than $R(u)$. 
						Therefore, $L_2[f , s)$ cannot contain $u$, and consequently cannot also contain $\set{u, w}$. 
						Therefore, the $1$-edge of $w$ is on $P_w[f , w]$, and hence $L_2$ must contain an edge on $P_w[v , w]$, which contradicts $L_2$ having no common edge with $P_w[v , w]$.  
						As a result, $L_2$ does not contain $w$.

					\item $P_w\cap T \neq \emptyset$ \\
						Let node $s'\in P_w \cap T$ be the closest node to $f$ on $P_w$. 
						Note that the reachability of every node in $P_w(s' , w]$ is greater than the $1$-reachability of $s'$. 
						Moreover, every node in any shortest uniform $1$-path of $s$ has reachability at most the $1$-reachability of $s$. 
						Therefore, since the $1$-reachability of $s$ is smaller than that of $s'$, any shortest uniform $1$-path of $s$ has no node in $P_w(s' , w]$.
						We need to show that there exists a path in $\calS$ that has no node in $P_w[v , s']$. 
						
						Let $s''$ be the matched neighbor of $s'$.  
						Then the concatenation of $P_w[f , s']$ and $\langle s', s'' \rangle$ is a shortest uniform $1$-path of $s''$. 
						Observe that $s''$ cannot be the same node as $s$ since otherwise $P_u[f , s]\circ \langle s, s' \rangle \circ P_w[v , s']$ would be a uniform $\vartheta$-path of length less than $r$ of $v$. 
						Then, we argue similarly to part (a) to show that there is a path in $\calS$ that does not have a common edge with $P_w[f , s'] \circ \langle s', s'' \rangle$. 
						This concludes that that there is a path in $\calS$ that is disjoint from $P_w[v , s']$ and hence $P_w[v , w]$.
\newpage
				\end{enumerate}
		
			\item[(II)] $s\in P_w$
			
				Here we show that there exists some shortest uniform $1$-path of $s$ that is disjoint from $P_u[f, u]$ and contains $v$.				
				Let us first show that there is a shortest uniform $1$-path of $s$ that contains $v$. 
				To do so we search and find such a path as follows. 
				First observe that $s$ must have a shortest uniform $1$-path. 
				Let $Z$ be an arbitrary shortest uniform $1$-path of $s$. 
				If $Z$ contains $v$, we are done with the search and $Z$ is one of such paths. 
				Let us thus assume that $Z$ does not contain $v$. 
				Since $Z$ is a shortest uniform $1$-path and $P_w[f , s]$ is a shortest uniform $0$-path of $s$, $tn(s) = |Z| + |P_w[f , s]|$. 
				Based on \Cref{lem:tenacity2}, $tn(s) < tn(e)$, and based on \Cref{lem:tenacity1}, $tn(e)< tn(v)$. 
				Moreover, $tn(v) = r + r'$. 
				Therefore, $|Z| + |P_w[f , s]| < r + r'$.
				Hence, since $|P_w[w , v]| = r'$, $|Z| + |P_w[v , s]| < r$.
				This concludes that the concatenation of $Z$ and $P_w[v , s]$ is an alternating walk of length less than $r$ from $f$ to $v$ that contains a $\vartheta$-edge of $v$. 
				Therefore, $Z$ and $P_w[v , s]$ must have a common edge since otherwise $v$ would have a uniform $\vartheta$-path of length less than $r$. 
				Let $g_1 = \set{h_1, h_1'}$ be $\partial(Z \ on \ P_w[v , s])$. 
				Then, based on \Cref{lem:sep}, $g_1$ is a separator of $Z$ and $P_w[f, s]$. 
				Hence, the concatenation of $P_w[f , h_1]$ and $Z[h_1 , s]$ is a shortest uniform $1$-path of $s$ that contains $v$. 
				Let $\calS'$ be the set of all shortest uniform $1$-paths of $s$ that contain $v$. 
				
				Let us now consider the following two cases separately: (recall that $T$ is the set of nodes that have two adjacent $0$-edges on current $P_w$ or $P_u$.)
				 
				\begin{enumerate} 
					\item[(a$'$)] $P_u\cap T = \emptyset$ \\
						Here we show that there is some path in $\calS'$ that is disjoint from $P_u[f , u]$. 
						We first show that there is a path in $\calS'$ that does not have a common edge with $P_u[f , u]$.
						This implies that the path dose not contain any node in $P_u[f , u)$ since the path is an alternating path.
						Then, we will show that the path does not also contain $u$. 
						
						For the sake of contradiction, let us assume that there is no path in $\calS'$ that has no common edge with path $P_u[f , u]$. 
						Let $g_2=\set{h_1, h_1'}\in P_u[f , u]$ be the closest edge to $u$ such that $g_2$ is $\partial(S' \ on \ P_u[f , u])$ for some $S'\in \calS'$. 
						Path $P_u[f , u]$ is a shortest uniform $(1-\varphi)$-path of $u$. 
						Moreover, $S'$ is a shortest uniform $1$-path of $s$ with length less than the length of a shortest $(1-\varphi)$-path of $w$. 
						Hence, $|P_u[f , u]| + |S'| < tn(e) < tn(v)$. 
						Then, based on \Cref{lem:sep}, $g_2$ is a separator. 
						Therefore, $L_1':= P_u[f , h_2] \circ S'[h_2 , s]$ is a shortest uniform $1$-path of $s$ that does not contain $v$. 
						Since $|L_1'| + |P_w[f , s]| = tn(s) < tn(e) < tn(v) = r + r'$ and $|P_w[f , v]| = r'$, the concatenation of $L_1'$ and $P_w[v , s]$ is an alternating walk of length less than $r$ from $f$ to $v$ that contains a $\vartheta$-edge of $v$. 
						Therefore, this concatenation cannot be a path. 
						Hence, $P_w[v , s]$ and $L_1'$ must have a common edge. 
						Then, based on \Cref{lem:sep}, the closest common edge of $P_w[v , s]$ and $L_1'$ to $f$ on $L_1'$ is a separator of $L_1'$ and $P_w[f , s]$.
						Let $g_3 = \set{h_3, h_3'}$ be the edge. 
						Then, $L_2' := P_w[f , h_3] \circ L_1'[h_3 , s]$ is a shortest uniform $1$-path of $s$ that contains $v$. 
						Note that $L_2'[h_3 , s]$ is the same as $L_1'[h_3 , s]$ and the same as $S'[h_3 , s]$. 
						Therefore, since $g_2$ is a separator of $S'$ and $P_u$, $L_2'[h_3 , s]$ has no common edge with $P_u[f , h_2]$. 
						Moreover, $L_2'[h_3 , s]$ has no common edge with $P_u[h_2 , u]$ as otherwise it contradicts the choice of $g_2$. 
						Therefore, $L_2'$ is a shortest uniform $1$-path of $s$ that contains $v$ and has no common edge with $P_u[f , u]$.   
						
						It is left to show that $L_2'$ does not also contain $u$. 
						Observe that since $P_u[f , w]$ is alternating, $u$ has an adjacent $1$-edge on $P_u$. 
						For the sake of contradiction, let us assume that $L_2'$ contains $u$, and consequently it must contain the adjacent $1$-edge of $u$. 
						All the nodes in $L_2'[f , s)$ have reachability of less than $R(s)$, and hence less than $R(w)$. 
						Therefore, $L_2'[f , s)$ cannot contain $w$, and consequently cannot also contain $\set{u, w}$. 
						Therefore, the $1$-edge of $u$ is on $P_u[f , u]$, and hence $L_2'$ must contain an edge on $P_u[f , u]$, which is contradictory.  
						As a result, $L_2'$ does not contain $u$.

					\item[(b$'$)] $P_u\cap T \neq \emptyset$ \\
						Let node $s'\in P_u \cap T$ be the closest node to $f$ on $P_u$. 
						Note that reachability of every node in $P_u(s' , u]$ is greater than $1$-reachability of $s'$. 
						Moreover, every node in any shortest uniform $1$-path of $s$ has reachability at most the $1$-reachability of $s$. 
						Therefore, since the $1$-reachability of $s$ is smaller than that of $s'$, any shortest uniform $1$-path of $s$ has no node in $P_u(s' , u]$.
						We need to show that there exists a path in $\calS'$ that has no node in $P_u[f , s']$. 
						
						Let $s''$ be the matched neighbor of $s'$.  
						Then the concatenation of $P_u[f , s']$ and $\langle s', s'' \rangle$ is a shortest uniform $1$-path of $s''$. 
						Observe that $s''$ cannot be the same node as $s$ since otherwise $P_u[f , s']\circ \langle s, s' \rangle \circ P_w[v , s]$ would be a uniform $\vartheta$-path of length less than $r$ of $v$. 
						Then, we argue similarly to part (a$'$) to show that there is a path in $\calS'$ that does not have a common edge with $P_u[f , s'] \circ \langle s', s'' \rangle$. 
						This concludes that that there is a path in $\calS'$ that is disjoint from $P_u[f , s']$.						 
				\end{enumerate}				 
		\end{enumerate}
\end{proof}

\subsection{Adaptation to the CONGEST model}\label{sec:congest-adapt}
	In this section, we employ randomness to adapt the DFNC algorithm to the \CONGEST model and prove \Cref{lem:path}. 
\begin{proof}[Proof of \Cref{lem:path}]

	Two types of messages are sent in the DFNC algorithm execution; the tokens and the flow messages. 
	Over every edge, at most one token is sent in each direction, which is the ID of some free node in the network. 
	Therefore, the dissemination of the tokens does not violate the congestion restriction of the \CONGEST model.  
	However, the described DFNC algorithm utilizes large flow messages. 
	First note that a flow message might contain a large number of flows (i.e., key-value pairs). 
	Moreover, since every time a node receives a flow, it divides its value in partitioning and forwarding the flow to its predecessors, the value of a flow might become a very small real number that needs a large number of bits to be represented.  
	In this section, we modify the DFNC algorithm to resolve this problem and only use $O(\log n)$-bit flow messages, and we show that this modification does not influence the desired effects of the DFNC algorithm. 
	
	We only modify the content of the flow messages, but all other rules and regulations regarding token dissemination and flow forwarding remain the same. 
	In the modified DFNC algorithm, the flow message that is sent by a node to its neighbor is only an integer in the ring of integers modulo $\gamma$, i.e., $\mathbb{Z}/\gamma \mathbb{Z}$, where $\gamma:=n^c$ for some large enough constant $c$.  
	That is, all flows of different edges that are sent over an edge in a single round are together aggregated and replaced with a single integer. 
	At the beginning of the execution, every node $v$ sends a distinct uniformly at random chosen integer in $\mathbb{Z}/\gamma \mathbb{Z}$, denoted by $\tau_{v, e}$, over each of its adjacent edges $e$. 
	Now let us consider a flow generation event over an arbitrary edge $e= \set{u,w}$, where $u$ and $w$ send tokens to each other in rounds $r_u$ and $r_w$ respectively.
	Without loss of generality, let $u<v$. 
	Then, we define the flow generation as node $w$ receiving flow $\tau_{u,e}$ in round $r_u$ and $u$ receiving flow $(- \tau_{u,e})\mod{\gamma}$ in round $r_w$. 
	
	Regarding flow forwarding, let us consider an arbitrary node $v$.
	We explain in the following how a node processes its incoming flows and how it sends the processed flows. 
	Node $v$ has an output buffer $O_v(t)$ for every round $t$ which is initially set to $0$. 
	Regarding processing the incoming flows, let us assume that $v$ receives flows in an arbitrary round $r$. 
	Let $z$ be the sum modulo $\gamma$ of all the received flows by $v$ in round $r$. 
	If $z\neq 0 \mod{\gamma}$, then it updates its output buffers as follows. 
	If the edges over which $v$ receives flows in round $r$ and the edges connecting $v$ to its predecessors are all $0$-edges, then $v$ sets $O_v(r+d+1)$ to $\big(O_v(r+d+1) + z\big) \mod{\gamma}$, where $d$ is the difference of $v$'s shortest uniform $0$-paths and its shortest uniform $1$-paths. 
	Otherwise, $v$ sets $O_v(r+1)$ to $\big(O_v(r+1) + z\big)\mod{\gamma}$.  
	Now regarding forwarding the received and processed flows, let us consider an arbitrary round $r'$ for $v$.
	If $O_v(r') \neq 0$, then $v$ forwards flows to its predecessors as follows.
	Let $p$ be the number of $v$'s predecessors. 
	If $p=1$, $v$ just forwards $O_v(r')$ to its only predecessor. 
	Otherwise if $p>1$, $v$ selects one of its predecessors uniformly at random which we call the \textit{poor predecessor} of $v$. 
	It independently chooses $p-1$ uniformly random integers $\alpha_1, \dots , \alpha_{p-1}$ from $\mathbb{Z}/\gamma \mathbb{Z}$ and considers $\alpha_p$ to be $(O_v(r')-\sum_{j=1}^{p-1}\alpha_j) \mod{\gamma}$. 
	It forwards $\alpha_p$ to its poor predecessor and forwards $\alpha_1, \dots , \alpha_{p-1}$ to the rest of its predecessors, one to each. 	
	
	Recall that $\calE$ is the DFNC algorithm execution on $\calG$ and $\calM$. 
	Let $\mathcal{R}$ be the execution of the modified DFNC algorithm on $\calG$ and $\calM$. 
	Considering this flow circulation, when the sum modulo $\gamma$ of some flows is not $0$ in $\mathcal{R}$, it corresponds to a flow with positive value of less than $1$ in $\calE$ and otherwise to a flow with value $1$ in $\calE$. 
	Fix an arbitrary edge $e = \set{u, w}$ over which a flow is generated.
	Let us define layers $L_1(e), L_2(e), \dots$ on $e$ as follows. 
	$L_1(e):= \set{u, w}$, and for all $t>1$, $L_t$ is the set of predecessors of the nodes in $L_{t-1}$. 
	Moreover, let $s(e)$ denote the smallest integer such that $|L_{s(e)}| = 1$. 
	Note that for all edges $e$, there exists such integer $s(e)$ since the cluster center of $e$'s endpoints is the single element of the last layer of $e$. 
	Note that the single node in layer $L_{s(e)}$ does not send any integer that is influenced by the flow of $e$ since the sum modulo $\gamma$ of the generated integers regarding the flow of $e$ is $0$ and it is all received and assigned to be sent by the node in a single round due to \Cref{lem:discards}. 
	Therefore, it is ``discarded'' by the node (actually cancelled out) when the sum is $0\mod{\gamma}$. 
	
	Now let us define $L(e):= \bigcup_{j<s(e)} L_j$.
	The nodes in $L(e)$ are the nodes for which the sum of the received flows of $e$ is positive but less than $1$ in $\calE$. 
	We show by induction on the layer numbers that for these nodes the sum modulo $\gamma$ of the received flows are not $0$ in $\mathcal{R}$. 
	In addition, we show that sum modulo $\gamma$ of the received flows of $e$ by these nodes are actually uniformly at random chosen numbers in $\mathbb{Z}/\gamma \mathbb{Z}$. 
	This actually shows that these nodes would receive a positive non-zero modulo sum even if the received flows of $e$ are aggregated with other flows of different edges.
	 
	Regarding the induction base, both $u$ and $w$ receive uniformly at random chosen numbers from $\mathbb{Z}/\gamma \mathbb{Z}$ as flows over $e$. 
	Now regarding the induction step, let us assume that for any positive integer $t<s(e)$, for every node $v\in \bigcup_{j=1}^{t-1} L_j$, the modulo sum of the received flows of $e$ is a uniformly at random chosen number from $\mathbb{Z}/\gamma \mathbb{Z}$.
	Then, we show that for every node $v'\in L_t$, the sum of the received flows of $e$ is also a uniformly at random chosen number from $\mathbb{Z}/\gamma \mathbb{Z}$.
	Let us fix an arbitrary node $v'\in L_t$. 
	Then, it is enough to show that the sum of the flows of $e$ that $v'$ receives from its neighbors in $L_{t-1}$ is a uniformly at random chosen number from $\mathbb{Z}/\gamma \mathbb{Z}$.
	As pictured below, let the edge from every node in $L_{t-1}$ to its poor predecessor (if any) be a dotted lines while the edges to the rest of its predecessors be solid lines. 
	
	 \begin{figure}[hbt!] \label{fig:layers}
		\centering	
	\begin{tikzpicture}
		\node at (-4, 0) {$L_t$};
		\shade[shading=ball, ball color=black] (-3,0) circle (.08);
		\shade[shading=ball, ball color=black] (-2,0) circle (.08);
		\shade[shading=ball, ball color=black] (-1,0) circle (.08);
		\shade[shading=ball, ball color=black] (0,0) circle (.08);
		\shade[shading=ball, ball color=black] (1,0) circle (.08);
		\shade[shading=ball, ball color=black] (2,0) circle (.08);
		\shade[shading=ball, ball color=black] (3,0) circle (.08);
		\node at (1, .3) {$v'$};
		
		\node at (-4, -1.5) {$L_{t-1}$};
		\shade[shading=ball, ball color=black] (-2.5,-1.5) circle (.08);
		\shade[shading=ball, ball color=black] (-1.5,-1.5) circle (.08);
		\shade[shading=ball, ball color=black] (-.5,-1.5) circle (.08);
		\shade[shading=ball, ball color=black] (.5,-1.5) circle (.08);
		\shade[shading=ball, ball color=black] (1.5,-1.5) circle (.08);
		\shade[shading=ball, ball color=black] (2.5,-1.5) circle (.08);
		
		\draw[color=red, >=stealth, ->, line width=.3mm] (-2.5,-1.45) - - (-3,-.05);		
		\draw[color=red, >=stealth, ->, line width=.3mm] (-2.5,-1.45) - - (1,-.05);
		\draw[color=blue, dashed, >=stealth, ->, line width=.3mm] (-1.5,-1.45) - - (-2,-.05);	
		\draw[color=red, >=stealth, ->, line width=.3mm] (-1.5,-1.45) - - (-3,-.05);	
		\draw[color=red, >=stealth, ->, line width=.3mm] (-1.5,-1.45) - - (0,-.05);
		\draw[color=red, >=stealth, ->, line width=.3mm] (-1.5,-1.45) - - (-1,-.05);	
		\draw[color=red, >=stealth, ->, line width=.3mm] (-0.5,-1.45) - - (-1,-.05);
		\draw[color=red, >=stealth, ->, line width=.3mm] (0.5,-1.45) - - (1,-.05);
		\draw[color=red, >=stealth, ->, line width=.3mm] (0.5,-1.45) - - (1,-.05);	
		\draw[color=blue, dashed, >=stealth, ->, line width=.3mm] (1.5,-1.45) - - (2,-.05);	
		\draw[color=red, >=stealth, ->, line width=.3mm] (1.5,-1.45) - - (1,-.05);	
		\draw[color=blue, dashed, >=stealth, ->, line width=.3mm] (2.5,-1.45) - - (0,-.05);
		\draw[color=red, >=stealth, ->, line width=.3mm] (2.5,-1.45) - - (3,-.05);
		\draw[color=blue, dashed, >=stealth, ->, line width=.3mm] (-0.5,-1.45) - - (1,-.05);
		\draw[color=blue, dashed, >=stealth, ->, line width=.3mm] (-2.5,-1.45) - - (-1,-.05);				
	\end{tikzpicture}	
		\caption{Layers $t-1$ and $t$ of edge $e$ in execution $\mathcal{R}$.}
		\label{fig:longer-path}
	\end{figure}
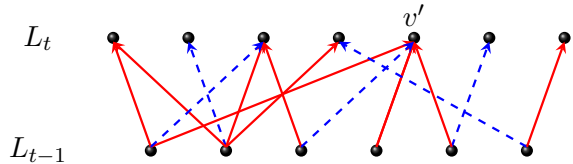
	
	Note that over each solid edge a uniformly at random chosen number from $\mathbb{Z}/\gamma \mathbb{Z}$ is sent to the nodes in layer $t$. 
	Hence, if node $v'$ has at least one incoming solid edge from layer $t-1$, then despite other received flows, the modulo sum of the received flows by $v'$ is a uniformly random number in $\mathbb{Z}/\gamma \mathbb{Z}$. 
	Now let us consider the case when all the incoming edges to $v'$ from layer $t-1$ are dotted edges. 
	Then, a node in $L_{t - 1}$ that has a dotted edge to $v'$ definitely has a solid edge to some other node in $L_t$. 
	Therefore, the sum of the flows of $e$ received by the nodes in $L_{t}\setminus{v'}$ from the nodes in $L_{t-1}$ is a uniformly random number in $\mathbb{Z}/\gamma \mathbb{Z}$. 
	Moreover, it is easy to see that the sum of all the received flows of $e$ by the nodes in $L_{t-1}$ is $0 \pmod{\gamma}$.
	Hence, the total sum of the received flows of $e$ by the nodes in $L_t$ is also $0 \mod{\gamma}$. 
	Therefore, since the sum of the received flows of $e$ by the nodes in $L_t\setminus \set{v'}$ is a uniformly random number in $\mathbb{Z}/\gamma \mathbb{Z}$, the sum of the received flows of $e$ by node $v'$ from the nodes in $L_{t-1}$ must also be a uniformly random number in $\mathbb{Z}/\gamma \mathbb{Z}$. 
	
	We showed that the modified DFNC algorithm that is represented in this proof maintains the desired effects of the flow circulation in the DFNC algorithm by just changing the flows content. 
	Considering \Cref{lem:DFNC-proof}, it thus concludes the proof of \Cref{lem:path}. 
	
      \end{proof}

%%% Local Variables:
%%% mode: latex
%%% TeX-master: "main"
%%% End: